\documentclass[11pt]{llncs}

\usepackage{makeidx}  
\usepackage{amssymb}  
\usepackage[dvips]{graphicx}
\usepackage{algorithmic,algorithm} 
\usepackage{amsmath}
\usepackage{comment}
\usepackage{subfigure}
\usepackage{multirow}

\usepackage{epsfig}

\usepackage{float}

\newcommand{\Real}{\mathbb{R}}

\makeatletter
  
  \@addtoreset{algorithm}{section}
\makeatother

\usepackage{comment}	
\begin{document} 

\title{Pattern Formation Problem for Synchronous Mobile Robots in 
the Three Dimensional Euclidean Space}
\author{Yukiko Yamauchi \thanks{Corresponding author. 
Address: 744 Motooka, Nishi-ku, Fukuoka 819-0395, Japan. 
Fax: +81-92-802-3637. Email: \texttt{yamauchi@inf.kyushu-u.ac.jp}}
\and 
Taichi Uehara 
\and 
Masafumi Yamashita}
\institute{Kyushu University, Japan. }

\maketitle 

\begin{abstract}
 Self-organization of a swarm of mobile computing entities in the 
 three-dimensional Euclidean space (3D-space) such as drones and satellites 
 attracts much attention as such systems are required to accomplish 
 more complicated tasks.
 We consider a swarm of autonomous mobile robots each of which is
 an anonymous point in 3D-space 
 and synchronously executes a common distributed algorithm. 
 We investigate the {\em pattern formation problem} that
 requires the robots to form a given target pattern
 from an initial configuration and characterize the problem 
 by showing a necessary and sufficient condition for 
 the robots to form a given target pattern.
 
 The pattern formation problem in the two dimensional Euclidean 
 space (2D-space) has been investigated 
 by Suzuki and Yamashita (SICOMP 1999, TCS 2010), and
 Fujinaga et al. (SICOMP 2015). 
 The symmetricity $\rho(P)$ of a configuration (i.e., the positions
 of robots) $P$ 
 is the order of the cyclic group that acts on $P$
 with the exception that when a robot is on the center of the smallest
 enclosing circle of $P$, $\rho(P)=1$. 
 It has been shown that  
 fully-synchronous (FSYNC) robots can form
 a target pattern $F$ from an initial configuration $P$
 if and only if $\rho(P)$ divides $\rho(F)$. 
 
 We extend the notion of symmetricity to 3D-space by 
 using the {\em rotation groups} each of which is defined by a set of 
 rotation axes and their arrangement. 
 We define the symmetricity $\varrho(P)$ of configuration $P$ 
 in 3D-space as the set of rotation groups that acts on $P$ and
 whose rotation axes do not contain any robot. 
 We show the following necessary and sufficient condition 
 for the pattern formation problem which is a natural extension
 of the existing results of the pattern formation problem in 2D-space: 
 FSYNC robots in 3D-space can form a target pattern $F$
 from an initial configuration $P$ 
 if and only if 
 $\varrho(P) \subseteq \varrho(F)$.
 This result guarantees that, for example,  
 from an initial configuration
 where the robots form a cube (i.e., the robots occupy the vertices of
 a cube), 
 they can form a regular octagon that consists of points on a plane 
 or a square anti-prism that has a vertical axis.
 In other words, these target patterns have lower symmetry than the
 initial configuration. 
 For solvable instances, we present a pattern formation 
 algorithm for oblivious FSYNC robots. 
 The insight of this paper is that symmetry of mobile robots 
 in 3D-space is sometimes lower than the symmetry of their positions  
 and the robots can show their symmetry by their movement. 

\medskip

\end{abstract}

\newpage


\section{Introduction}
\label{sec:intro}

Distributed control of a system consisting of autonomous mobile computing
entities in the three dimensional Euclidean space (3D-space)
is one of the most challenging problems 
in distributed computing theory and robotics.
One of the most important properties that is expected to such
systems is {\em self-organization ability} that enables the system
to obtain the coordination by itself. 
For example, drones are becoming widely available and 
their applications in sensing, monitoring, and rescuing 
in harsh environment such as disaster area and active volcanoes, 
where they are required to coordinate themselves without
human intervention, are attracting much attention. 

As one of the most fundamental tasks in 3D-space,
this paper considers the {\em pattern formation problem}
that requires a swarm of robots to
form a given 3D target pattern. 
A robot is a point in 3D-space that autonomously moves according to a
given rule.
We adopt the conventional computation model~\cite{FYOKY15,SY99,YS10},
i.e., 
each robot repeats a {Look-Compute-Move cycle},
where it observes the positions of other robots (Look phase), 
computes its next position with a given algorithm (Compute phase),
and moves to the next position (Move phase). 
Each robot is {\em anonymous} in the sense that they have no
identifiers and the robots are {\em uniform} in the sense that all robots
execute a common algorithm. 
Each robot has no access to the global $x$-$y$-$z$ coordinate system 
(like GPS) and its observation and movement are done in terms of its 
{\em local $x$-$y$-$z$ coordinate system}.
The origin of the local coordinate system of a robot is its current
position and the local coordinate system has arbitrary
directions and unit distance.
However we assume that all local coordinate systems are right-handed.
Thus each local coordinate system is either a
uniform scaling, transformation, rotation, or their combination of the
global coordinate system. 
A robots is {\em oblivious} if its local memory is refreshed at
the end of each cycle, otherwise {\em non-oblivious}.
Hence the input to the algorithm at an oblivious robot is
the observation obtained in the current cycle. 
Suzuki and Yamashita pointed out that oblivious mobile robot system 
is self-stabilizing~\cite{D74}, that guarantees self-organization and fault
tolerance against finite number of transient faults~\cite{SY99}.
In a Move phase, each robot reaches the next position 
and in this paper we do not care for the track of movement.\footnote{
This type of movement is called rigid movement. 
On the other hand, non-rigid movement
allows a robot to stop en route after moving unknown minimum moving
distance $\delta$ in a Move phase. If the track to the
next position is shorter than $\delta$, non-rigid movement makes a 
robot stop at the next position.}
We consider the {\em fully-synchronous (FSYNC)} model
where the robots execute the
$t$-th Look-Compute-Move cycle at the same time
with each of the Look, Compute, and Move phases completely synchronized. 
Here a {\em configuration} of the robots is 
the positions of the robots observed in the global coordinate system. 
These assumptions mean that the robots do not have
explicit communication medium and 
they have to tolerate inconsistency among local coordinate systems
so that they coordinate themselves by building some agreement
by using inconsistent observations. 

The pattern formation problem was first introduced by Suzuki and
Yamashita for the robots moving in the two-dimensional Euclidean
space (2D-space)~\cite{SY99}.
They characterized the class of formable patterns by using
the notion of {\em symmetricity} of an initial configuration. 
The symmetricity of a configuration is essentially the order of the
cyclic group that acts on it. 
Let $P$ be an initial configuration of robots without any
multiplicity.\footnote{Throughout this paper, we assume that
any initial configuration contains no multiplicity. 
It is impossible to break up multiple robots on a single position
as we assume all robots execute the same algorithm. }
We consider the decomposition of $P$ into regular $m$-gons 
centered at the center of the smallest enclosing circle of $P$.
The symmetricity $\rho(P)$ of $P$ is the maximum value of such $m$
with an exception that when a single point of $P$ is at 
the center of the smallest enclosing circle of $P$, $\rho(P)=1$.
We consider a point as a regular $1$-gon with an arbitrary center
and a set of two points as a regular $2$-gon with 
the center at the midpoint of the two points. 
This exception is derived from an easy symmetry breaking algorithm,
i.e., the robot on the center leaves its current position. 
Then they showed that FSYNC robots can form
a target pattern $F$ from a given initial configuration $P$
if and only if $\rho(P)$ divides $\rho(F)$
regardless of obliviousness. 
This impossibility is by the fact that since $\rho(P)$ divides the
robots into regular $\rho(P)$-gons, 
these symmetric $\rho(P)$ robots cannot break their symmetry. 
Thus robots in 2D-space cannot break
rotational symmetry of an initial configuration. 

Yamauchi et al. first showed that rotational symmetry of robots in
3D-space causes the same impossibility~\cite{YUKY15}. 
They considered the {\em plane formation problem} that requires the
robots to land on a common plane without making any
multiplicity.  
In 3D-space, there are five-kinds of rotation groups with finite order, 
i.e., 
{\em the cyclic groups}, {\em the dihedral groups},
{\em the tetrahedral group},
{\em the octahedral group},
and {\em the icosahedral group}~\cite{A88,C73,C97}.
Given a configuration $P$ in 3D-space,
its {\em rotation group} $\gamma(P)$ is the rotation group that acts on $P$
and none of its supergroup in these five kinds of rotation groups 
acts on $P$. 
They called the cyclic groups and the dihedral groups
{\em two-dimensional (2D)},
while the remaining three rotation groups
{\em three-dimensional (3D)}, 
because 3D rotation groups do not act on points on a plane. 
Then they showed that the robots cannot form a plane from an initial 
configuration $P$ 
if and only if 
$\gamma(P)$ is a 3D rotation group and all robots are not on the
rotation axes of $\gamma(P)$. 
The results showed that even when the robots form a regular polyhedron 
(except an regular icosahedron) in an initial configuration,
they can break their 3D rotation group and form a plane. 

In this paper, we define the {\em symmetricity} $\varrho(P)$ of a
configuration $P$ in 3D-space as a set of rotation groups
that acts on the positions of robots and 
that consists of rotation axes containing no robot.
We will give the following necessary and sufficient condition
for the pattern formation problem in 3D-space. 
\begin{theorem}
\label{theorem:main} 
Regardless of obliviousness, 
FSYNC robots can form a target pattern $F$ from 
an initial configuration $P$ 
if and only if 
$\varrho(P) \subseteq \varrho(F)$. 
\end{theorem}

The impossibility is derived from symmetry among robots in the same
way as 2D-space. 
For the solvable instances, we present a pattern formation algorithm
for oblivious FSYNC robots that non-oblivious robots can 
execute correctly by just ignoring its memory contents. 
Theorem~\ref{theorem:main} guarantees, for example, that
the robots can form a square anti-prism or a regular octagon 
from an initial configuration where they form a cube, because
their symmetricity is identical while the rotation groups of the
target patterns are dihedral and that of the initial configuration is
the octahedral group. Thus the rotation group of
the robots decreases during the formation.
(See Figure~\ref{fig:exrho}.)
We will show that the robots can translate a given initial configuration
$P$ into another configuration $P'$ with $\gamma(P') \in \varrho(P)$.
From the definition, $\gamma(P') \preceq \gamma(P)$. 
The symmetry breaking algorithm is based on the ``go-to-center''
algorithm proposed in \cite{YUKY15}. 
Theorem~\ref{theorem:main} guarantees
that such $\gamma(P')$ is an element of $\varrho(F)$
which means that $\gamma(P')$ is a subgroup of $\gamma(F)$.
We will show that the robots can easily form $F$ from such $P'$
by embedding an image of $F$ into $P'$ and building an agreement on
a perfect matching between $F$ and $P'$.

\begin{figure}[t]
\centering 
\subfigure[]{\includegraphics[width=2cm]{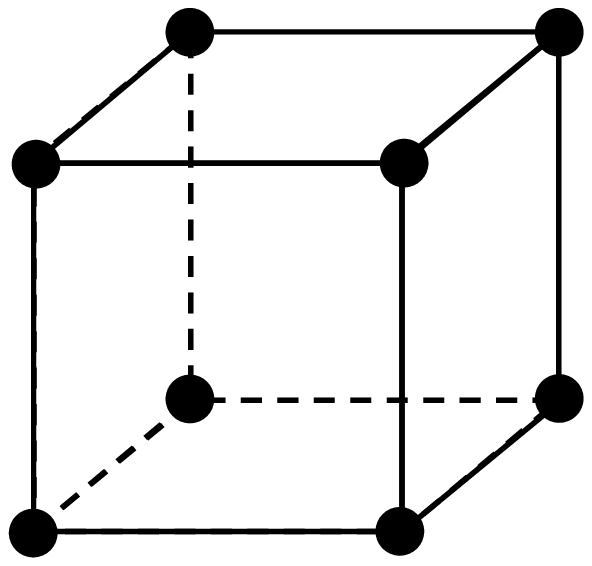}\label{fig:rho-cube}}
\hspace{1mm}
\subfigure[]{\includegraphics[width=2cm]{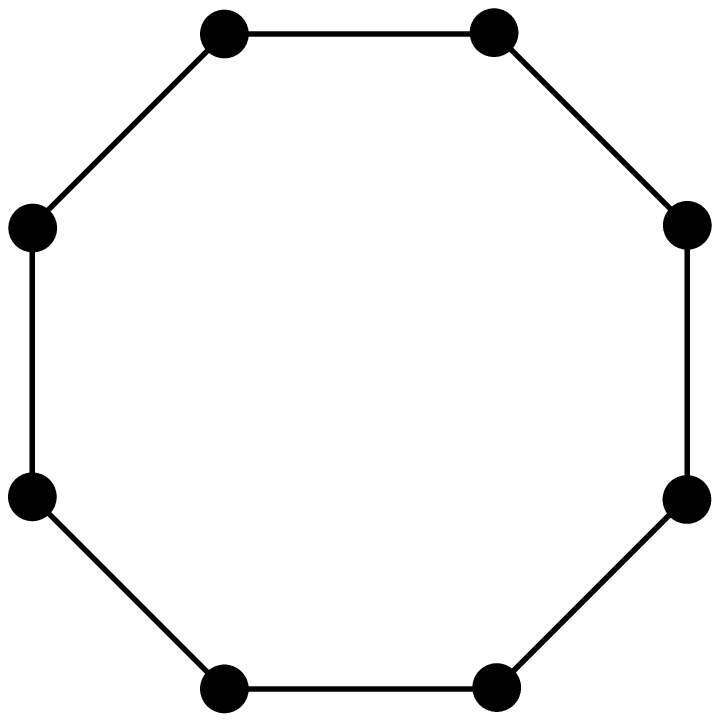}\label{fig:rho-octa}}
\hspace{1mm}
\subfigure[]{\includegraphics[width=2cm]{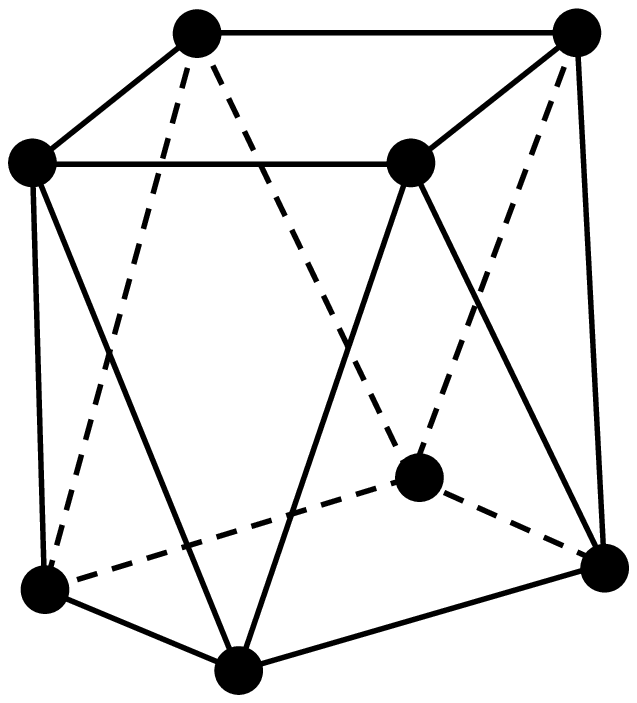}\label{fig:rho-anti}}
 \caption{Example of an initial configuration and target patterns with
 the same symmetricity. 
 (a) Initial configuration $P$ where robots form a cube and 
 $\gamma(P)$ is the octahedral group. 
 (b)(c) Target patterns $F$ and $F'$
 where robots form a regular octagon and a regular square anti-prism.
 Thus $\gamma(F)$ and $\gamma(F')$ are dihedral groups. 
 }
\label{fig:exrho}
\end{figure}

The main contribution of this paper is the fact that
the symmetry of moving points is different from
symmetry of their positions and
robots can show their symmetry by their movement.
We finally note that our result is a 
generalization of existing results for FSYNC robots in
2D-space~\cite{SY99}. 

\medskip
\noindent{\bf Related work.~}
The only existing paper on autonomous mobile robot systems in
3D-space~\cite{YUKY15}
considers the plane formation problem motivated by
the fact that autonomous mobile robot systems in 2D-space  
has been extensively investigated.
We mainly survey the results on formation problems in 2D-space. 
The main research interest has been the computational power of robot
systems and minimum requirements for the robots to accomplish a 
given task. 
Many fundamental distributed tasks have been introduced, for example,
{\em gathering}~\cite{SY99},
{\em pattern formation}~\cite{SY99},
{\em partitioning}~\cite{EP09},
{\em covering}~\cite{IKY14}, and so on. 
The book by Flocchini et al.~\cite{FPS12} contains almost all results 
on autonomous mobile robot systems up to year 2012.

Asynchrony among robots is classified into three models; 
the fully synchronous (FSYNC) model,
the {\em semi-synchronous} (SSYNC) model,
and the {\em asynchronous} (ASYNC) model. 
The robots are SSYNC if some robots do not start 
the $i$-th Look-Compute-Move cycle for some $i$, 
but all robots that have started the cycle synchronously 
execute their Look, Compute, and Move phases~\cite{SY99}. 
The robots are ASYNC if no assumptions are made on the execution of 
Look-Compute-Move cycles~\cite{FPSW08}. 

Yamashita et al. characterized the pattern formation problem 
for each of the FSYNC, SSYNC, and ASYNC models~\cite{FYOKY15,SY99,YS10},
that are summarized as follows:
(1) For non-oblivious FSYNC robots,
pattern $F$ is formable from an initial configuration $P$
if and only if $\rho(P)$ divides $\rho(F)$.
(2) Pattern $F$ is formable from $P$ by oblivious ASYNC (thus SSYNC)
robots 
if $F$ is formable from $P$ by non-oblivious FSYNC robots,
except for $F$ being a point of multiplicity 2.

This exceptional case is called the rendezvous problem.
Indeed, it is trivial for two FSYNC robots,
but it is unsolvable for two oblivious SSYNC (and hence ASYNC)
robots while oblivious SSYNC (and ASYNC) robots can converge to a
point~\cite{SY99}.  
However, more than two robots can form a point in the
SSYNC model~\cite{SY99} and in the ASYNC model~\cite{CFPS12}. 
In terms of symmetricity, the point formation problem is one of
the easiest problems (except the rendezvous problem), 
since $\rho(F)=n$ when $F$ is a point of multiplicity $n$ 
and $\rho(P)$ is always a divisor of $n$ by the definition of the symmetricity,
where $n$ is the number of robots.

The other easiest case is a regular $n$-gon, 
which is also called the circle formation problem, 
since $\rho(F) = n$.
Recently the circle formation problem for $n$ oblivious ASYNC 
robots ($n \neq 4$) 
is solved without agreement of clockwise direction,
i.e., chirality~\cite{FPSV14}.

Yamauchi et al. showed a randomized pattern formation
algorithm for oblivious ASYNC robots that breaks the symmetricity
of the initial configuration and forms any
target pattern with probability $1$~\cite{YY14}.

The notion of {\em compass} was first introduced in~\cite{FPSW05} 
that assumes agreement of the direction and/or the orientation
(i.e., the positive direction) of $x$-$y$ local coordinate systems. 
Flocchini et al. showed that if oblivious ASYNC robots agree on the 
directions and orientations of $x$-$y$ axes, 
they can form any arbitrary target pattern~\cite{FPSW08}. 

Das et al. characterized the formation of a sequence of 
patterns by oblivious SSYNC robots 
in terms of symmetricity~\cite{DFSY15}.
They showed that symmetricity of each pattern of
a formable sequence should be identical and a multiple of the
symmetricity of an initial configuration. 
Such sequence of patterns is a geometric global memory formed
by oblivious robots. 

All above results are based on unlimited visibility of robots. 
A robot has {\em limited visibility} if it can observe other robots
within the unknown fixed visibility range~\cite{FPSW05}.
Yamauchi et al. showed that oblivious FSYNC (thus SSYNC and
ASYNC) robots with limited visibility 
have substantially weaker formation power 
than FSYNC robots with unlimited visibility~\cite{YY13}. 
Ando et al. proposed a convergence algorithm for oblivious SSYNC
robots with limited visibility~\cite{AOSY99} 
while Flocchini et al. assumed consistent compass for
convergence of oblivious ASYNC robots with limited
visibility~\cite{FPSW05}. 

Peleg et al. first introduced the 
{\em luminous robot model} where each robot is equipped with
externally and/or internally visible lights~\cite{P05}.
Das et al. provided an algorithms for oblivious luminous
robots to simulate robots without lights in stronger synchronization
model and showed that two ASYNC (thus SSYNC) luminous robots
can form a point~\cite{DFPSY16}.

\medskip
\noindent{\bf Organization.~~} 
We define the mobile robot model in Section~\ref{sec:prel} and 
we introduce the rotation groups and symmetricity of a
configuration in Section~\ref{sec:3Dsym}. 
Then we show that the robots can reduce their rotation group to some
element of the symmetricity of an initial configuration by the
``go-to-center'' algorithm in Section~\ref{sec:show-sym}.
We prove the necessity of Theorem~\ref{theorem:main} in
Section~\ref{sec:nec} and sufficiency of Theorem~\ref{theorem:main}
in Section~\ref{sec:suf} by showing a pattern formation algorithm
for oblivious FSYNC robots.
Finally Section~\ref{sec:concl} concludes this paper.

\section{Preliminary} 
\label{sec:prel}

Let $R = \{r_1, r_2, \ldots, r_n\}$ be a set of $n \geq 3$ robots 
each of which is represented by a point in 3D-space. 
Each robot is anonymous and there is no way to distinguish them.
We use the indexes just for description. 

By $Z_0$ we denote the global $x$-$y$-$z$ coordinate system.
Let $p_i(t) \in \Real^3$ be the position of $r_i$ at time $t$ in $Z_0$, 
where $\Real$ is the set of real numbers.
A {\em configuration} of $R$ at time $t$ is denoted by a multiset 
$P(t) = \{p_1(t), p_2(t), \ldots, p_n(t)\}$.
Let ${\cal P}_n^3 = ({\mathbb R}^3)^n$ be the set of all
configurations. 
We assume that the robots initially occupy distinct positions,
i.e., $p_i(0) \not= p_j(0)$ for all $1 \leq i < j \leq n$. 
\footnote{This assumption is necessary because
it is impossible to break up multiple oblivious FSYNC robots 
(with the same local coordinate system)
on a single position as long as they execute the same algorithm. 
The proposed pattern formation algorithm does not make any
multiplicity during the formation. However, we have to consider
configurations with multiplicity when we prove impossibility
by checking executions of any arbitrary algorithm. 
}
The robots have no access to $Z_0$. 
Instead, each robot $r_i$ observes the positions of other robots in its
local $x$-$y$-$z$ coordinate system $Z_i$,
where the origin is always its current position,
while the direction of each positive axis and the magnitude of 
the unit distance are arbitrary but never
change.\footnote{
Since $Z_i$ changes whenever $r_i$ moves,
notation $Z_i(t)$ is more rigid,
but we omit parameter $t$ to simplify its notation.}
We assume that $Z_0$ and all $Z_i$ are right-handed. 
Thus $Z_i$ is either a uniform scaling, transformation, rotation,
or their combinations of $Z_0$. 
By $Z_i(p)$ we denote the coordinates of a point $p$ in $Z_i$.

We consider discrete time $0, 1, 2, \cdots$ and at each time step
the robots execute a Look-Compute-Move cycle with each of
Look, Compute, and Move phases completely synchronized, i.e., 
we consider the {\em fully-synchronous (FSYNC)} robots in this paper.
We specifically assume without loss of generality 
that the $(t+1)$-th Look-Compute-Move cycle starts 
at time $t$ and finishes before time $t+1$.
At time $t$,
each $r_i \in R$ obtains a multiset 
$Z_i(P(t)) = \{ Z_i(p_1(t)), Z_i(p_2(t)), \ldots , Z_i(p_n(t)) \}$
in the Look phase. 
We call $Z_i(P(t))$ the {\em local observation} of $r_i$ at $t$.
Then $r_i$ computes its next position using an algorithm $\psi$,
which is common to all robots. 
If $\psi$ uses only $Z_i(P(t))$, we say that $r_i$ is {\em oblivious}. 
Otherwise, we say $r_i$ is {\em non-oblivious}, i.e., 
$r_i$ can use past local observations and past outputs of $\psi$. 
Finally, $r_i$ moves to $\psi(Z_i(P(t)))$ in $Z_i$ before time $t+1$.
Thus the movement of robots is {\em rigid}. 
In this paper, we do not care for the track of the movement of robots, 
rather each robot jumps to its next position. 
An infinite sequence of configurations
${\cal E}: P(0), P(1), \ldots$ is called an {\em execution} 
from an {\em initial configuration} $P(0)$.
Observe that the execution $\cal E$ is uniquely determined,
once initial configuration $P(0)$, 
local coordinate systems of robots at time $0$,
initial local memory contents (if any), 
and algorithm $\psi$ are fixed.

\noindent{\bf Pattern formation problem.~} 
The {\em pattern formation problem} is to make the robots 
form a given target pattern $F$ from an initial configuration $P$. 
The target pattern $F$ is given to each robot as a set of
coordinates of $n$ points in $Z_0$. 
We assume that $F$ does not contain any multiplicity, but as we will discuss
in Section~\ref{sec:concl}, we can easily extend the results to target
patterns with multiplicities. 
Because robots do not have access to the global coordinate 
system, it is impossible to form $F$ itself. 
Let ${\cal T}$ be the set of all rotations, translations, 
uniform scalings, and their combinations. 
We say $F'$ is {\em similar} to $F$ if there exists $Z \in {\cal T}$ 
such that $F' = Z(F)$, which we denote by $F' \simeq F$. 
We say that the robots form a target pattern $F$ 
from an initial configuration $P$,
if, regardless of the choice of local coordinate systems and
memory contents (if any) of
robots in the initial configuration, 
any execution $P(0)(=P), P(1), \ldots$ reaches 
a configuration $P(t)$ that is similar to $F$ in finite time. 

For any (multi-)set of points $P$, by $B(P)$ and $b(P)$,
we denote the {\em smallest enclosing ball} of $P$ and its center,
respectively. 
A point on the sphere of a ball is said to be {\em on} the ball 
and we assume that the {\em interior} or the {\em exterior}
of a ball does not include its sphere. 
The {\em innermost empty ball} $I(P)$ of $P$ is the ball centered at $b(P)$ 
and contains no point of $P$ in its interior, but
contains at least one point of $P$ on it. 
When all points are on $B(P)$, 
we say that $P$ is {\em spherical}.
Given a ball $B$, $rad(B)$ denotes the radius of $B$. 
We denote a ball centered at an arbitrary point $b$ and 
with radius $r$ by $Ball(b, r)$. 

\section{Symmetricity in 3D-Space} 
\label{sec:3Dsym}

In this section, we define the {\em rotation group} and the
{\em symmetricity} 
of a set of points and investigate the relation between 
the two notions. 
We start with an arbitrary set of points because
any initial configuration and any target pattern contain no
multiplicity, 
and then extend
these notions to multiset of points since we should consider 
an arbitrary algorithm that may produce multiplicity when we discuss
impossibility. 

In 2D-space, the symmetricity $\rho(P)$ of a configuration $P$
considers the worst-case arrangement of local coordinate system of $P$,
that is caused by the rotations around the center of the
smallest enclosing circle of $P$, denoted by $c(P)$,
i.e., the {\em cyclic group} of order $\rho(P)$. 
Hence $\rho(P)$ is redefined as follows: 
For an initial configuration $P$ identified as a set of points, 
$\rho(P)$ is the maximum order of the cyclic group
that acts on $P$ with the exception such that
when $c(P) \in P$, $\rho(P) = 1$. 
This exceptional case means that 
whenever $c(P) \in P$, 
the robot on $c(P)$ can translate $P$ into another asymmetric
configuration $P'$ with $\rho(P') = 1$ by leaving 
$c(P)$ which is on the rotation axes of the cyclic group. 

\begin{figure}[t]
\centering 
\includegraphics[width=3cm]{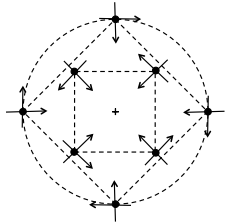}
\caption{A symmetric initial configuration in 2D-space,
whose symmetricity is $4$.
Eight robots and their local coordinate systems are 
symmetric with respect to the center of their smallest enclosing circle. 
There are two groups consisting of $4$ symmetric robots 
and the robots in each group cannot break their symmetry. }
\label{fig:2dim-sym}
\end{figure}

In 3D-space, 
we consider rotation groups so that we check all possible symmetric
arrangement of local coordinate systems. 
There are only five kinds of finite-order rotation groups in 3D-space; 
the cyclic groups, the dihedral groups, 
the tetrahedral group, the octahedral group, and 
the icosahedral group.\footnote{These five kinds of groups are 
proper subgroup of $SO(3)$, which is defined by  
rotations of a unit ball and its order is infinite}
Symmetry operations in 3D-space consist of rotation around an axis, 
reflection for a mirror plane (bilateral symmetry),
reflections for a point (central inversion),
and rotation-reflections~. 
However, we consider symmetry among robots, specifically,
whether the robots have identical local observation or not. 
Because all local coordinate systems are right-handed, 
it is sufficient to consider transformations
that preserve the center of the
smallest enclosing ball of robots and keep Euclidean distance and
handedness,
in other words, direct congruent transformations.
Such symmetry operations consist of rotations around some axes and
we consider above five kinds of rotation groups. 
(See, for example, \cite{C73,C97} for more detail.)

In the following, we first define the rotation group $\gamma(P)$
of a set of points $P$,
which is the symmetry that the robots can agree by just observing
$P$ in their local coordinate systems.
Then we define the rotation group $\sigma(P)$ of the arrangement
of local coordinate systems of $P$,
which is the symmetry that the robots can never break. 
However, the robots do not agree on $\sigma(P)$ by just observing
the set of points $P$ in their local coordinate systems. 
We define the symmetricity $\varrho(P)$ of $P$ that consists of
all possible rotation groups of the arrangement of
local coordinate systems of $P$.
Intuitively the maximal elements in $\varrho(P)$ are the worst-case 
symmetry of the robots. 
The maximality of $G \in \varrho(P)$ means that
there is no proper supergroup of $G$ in $\varrho(P)$ and 
$\varrho(P)$ actually has multiple such maximal elements.
Based on these notions, we present the first impossibility result
that shows that FSYNC robots can never reduce $\sigma(P)$ of an initial
configuration $P$ by any arbitrary algorithm.

\subsection{Rotation group of a set of points}
\label{sec:def-rg}

We formally define the five kinds of rotation groups. 
The rotation group $SO(3)$ has five kinds of
subgroups of finite order~\cite{C73,C97}; 
the cyclic groups $C_k$ ($k=1, 2, \cdots$), 
the dihedral groups $D_{\ell}$ ($\ell = 2, 3, \cdots$), 
the tetrahedral group $T$, 
the octahedral group $O$, and 
the icosahedral group $I$. 
Each of these groups is identified by the rotations 
of a regular pyramid with a regular $k$-gon base, 
a regular prism with regular $\ell$-gon bases,
a regular tetrahedron, 
a regular octahedron, and a regular icosahedron, respectively.
(See Figure~\ref{fig:csym}.)
For example, consider a regular pyramid that has a regular $k$-gon as its base. 
The rotation operations for this regular pyramid are rotations 
by $2\pi i/k$ for $i = 1, 2, \cdots, k$ 
around an axis containing the apex and the center of the base. 
We call such an axis {\em $k$-fold axis}. 
Let $a^i$ be the rotation by $2\pi i/k$ 
around this $k$-fold axis with $a^k = e$ where $e$ is the identity element. 
Then, $a^1, a^2, \ldots, a^k$ form the {\em cyclic group} $C_k$. 

\begin{figure}[t]
\centering 
\subfigure[]{\includegraphics[width=2cm]{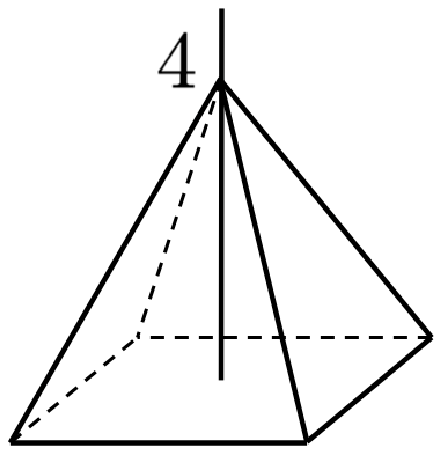}\label{fig:axis-cyc}}
\hspace{1mm}
\subfigure[]{\includegraphics[width=2cm]{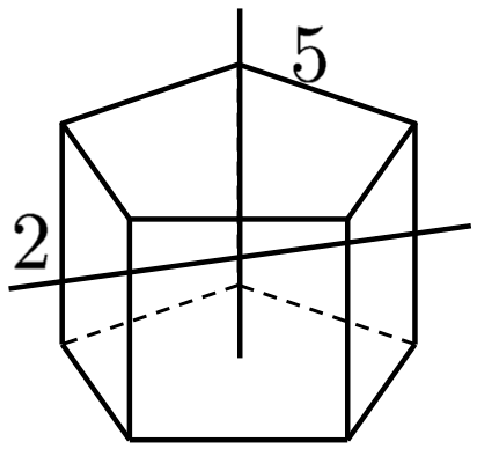}\label{fig:axis-dih}}
\hspace{1mm}
\subfigure[]{\includegraphics[width=2cm]{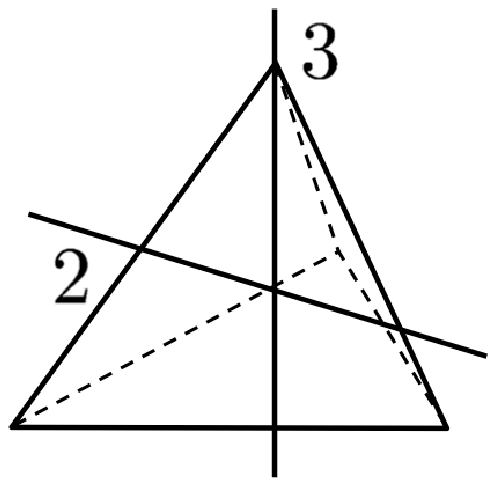}\label{fig:axis-tetra}}
\hspace{1mm}
\subfigure[]{\includegraphics[width=2cm]{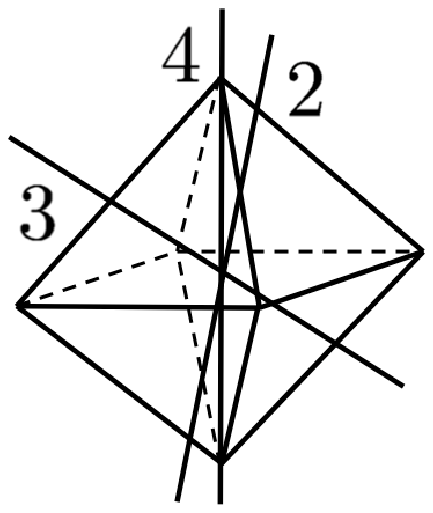}\label{fig:axis-octa}}
\hspace{1mm}
\subfigure[]{\includegraphics[width=2cm]{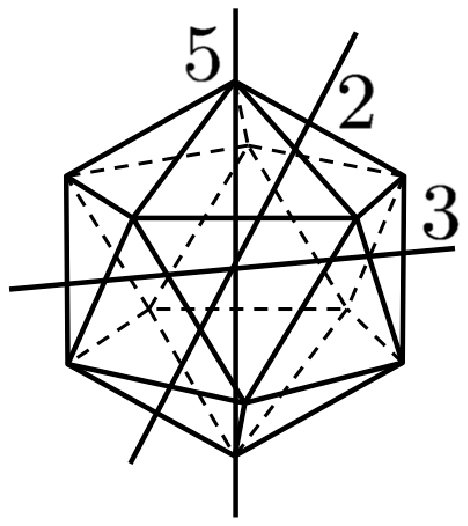}\label{fig:axis-icosa}}
\caption{Rotation groups.  
(a) the cyclic group $C_4$, 
(b) the dihedral group $D_5$, 
(c) the tetrahedral group $T$, 
(d) the octahedral group $O$, 
and (e) the icosahedral group $I$. 
Figures show only one axis for each fold of axes.}
\label{fig:csym}
\end{figure}

A regular prism (except a cube) has two parallel 
regular $\ell$-gons as its top and bottom bases 
and has two types of rotation axes, 
one is the $\ell$-fold axis containing the centers of its top and bottom
bases, 
and the others are $\ell$ $2$-fold axes that exchange the top and the 
bottom. 
We call this $\ell$-fold axis {\em principal axis} and 
the remaining $\ell$ $2$-fold axes {\em secondary axes}. 
These rotation operations on a regular prism form the 
{\em dihedral group} $D_{\ell}$. 
When $\ell=2$, we can define $D_2$ in the same way but in the group theory
we do not distinguish the principal axis. 

The remaining three rotation groups $T$, $O$, and $I$ are called 
the polyhedral groups. 
Table~\ref{table:elements} shows the number of rotation axes and the 
number of elements around each type of rotation axes 
for each of the polyhedral groups.

\begin{table}[t]
\begin{center}
\caption{Three polyhedral groups. 
The number of elements around $k$-fold axes
excluding the identity element is shown. 
The number in the brackets is the number of rotation axes.} 
\label{table:elements} 
\begin{tabular}{c|c|c|c|c|c}
\hline 
Polyhedral group & $2$-fold axes & $3$-fold axes & $4$-fold axes &
$5$-fold axes & Order \\ 
\hline 
$T$ & 3(3) & 8(4) & - & - & 12 \\ 
$O$ & 6(6) & 8(4) & 9(3) & - & 24 \\ 
$I$ & 15(15) & 20(10) & - & 24(6) & 60\\ 
\hline
\end{tabular}
\end{center}
\end{table} 

Let ${\mathbb S} = \{C_{k}, D_{\ell}, T, O, I \ | k = 1,2,
\ldots, and \ \ell = 2, 3, \ldots\}$ 
be the set of rotation groups,
where $C_1$ is the rotation group with order $1$;
its unique element is the identity element (i.e., $1$-fold rotation).

When $G'$ is a subgroup of $G$ ($G, G' \in {\mathbb S}$), 
we denote it by $G' \preceq G$. 
If $G'$ is a proper subgroup of $G$ (i.e., $G' \neq G$), 
we denote it by $G' \prec G$. 
For example, we have $T \prec O$, $T \prec I$, but $O \not\prec I$. 
If $G \in {\mathbb S}$ has a $k$-fold axis, then $C_k \preceq G$. 
Clearly, $C_{k'} \preceq C_k$ if $k'|k$, i.e., $k'$ divides $k$, 
which also holds for dihedral groups. 
The relation $\prec$ is asymmetric and transitive. 
Figure~\ref{fig:subgroups} shows the structure of subgroups of 
polyhedral groups. 

\begin{figure}[t]
\centering 
\includegraphics[width=4cm]{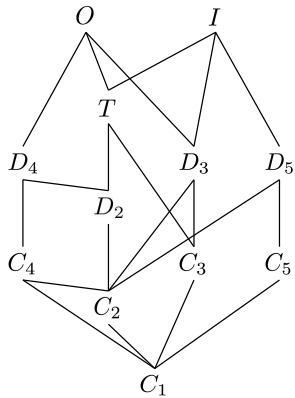}
\caption{The subgroups of polyhedral groups. Two groups $G$ and $G'$ are 
linked by a line if there is no proper subgroup $G''$ of $G$ 
that satisfies $G' \prec  G''$, and $G$ is put above $G'$. }
\label{fig:subgroups}
\end{figure}

For a set of points $P$,
we consider rotations on $P$ that produces $P$ itself. 
\begin{definition}
Let $P \in {\cal P}_n^3$ be a set of points. 
The rotation group $\gamma(P)$ of $P$ 
is the rotation group in ${\mathbb S}$ that acts on $P$ and 
none of its proper supergroup in ${\mathbb S}$ acts on $P$. 
 \end{definition}

 Clearly, if $\gamma(P) \succ C_1$, all rotation axes of
 $\gamma(P)$ contain $b(P)$, which is the single intersection of them. 
 From the definition, we can uniquely determine $\gamma(P)$
 irrespective of (local) coordinate system to observe $P$. 
For example, when $P$ forms a regular pyramid with 
a regular square base, $\gamma(P) = C_4$, 
when $P$ forms a square, $\gamma(P) = D_4$,   
and when $P$ forms a cube, $\gamma(P) = O$. 
When the robots are on a line and symmetric against $b(P)$, 
$\gamma(P)=D_{\infty}$, and $\gamma(P)=C_{\infty}$ if they are 
asymmetric against $b(P)$.

We say a rotation axis of $\gamma(P)$ is {\em occupied} when
it contains some points of $P$ and {\em unoccupied} otherwise. 
For example, when $P$ forms a cube, 
all $3$-fold axes of $\gamma(P) = O$ are occupied while
all $2$-fold axes and all $4$-fold axes are
unoccupied. 

In the group theory, we do not distinguish the principal axis 
of $D_2$ from the other two $2$-fold axes. 
Actually, since we consider the rotations on a set of points
in 3D-space, 
we can recognize a principal axis of $D_2$. 
Consider a sphenoid consisting of $4$ congruent triangles 
(Figure~\ref{fig:sphenoid}).
A rotation axis of such a sphenoid contains the midpoints of 
opposite edges and there are three $2$-fold axis 
perpendicular to each other. 
Hence the rotation group of such a sphenoid is $D_2$. 
However we can recognize, for example, the vertical $2$-fold axis 
from the others by their lengths (between the midpoints connecting).
The polyhedra on which only $D_2$ can act are 
rectangles and the family of such sphenoids, and we can always
recognize the principal axis. 
Other related polyhedra are squares and regular tetrahedra, but 
$D_4$ also acts on a square and $T$ acts on a regular tetrahedron. 
Hence their rotation groups are proper supergroup of $D_2$. 
We can show the following property regarding the principal axis of
$D_2$. See Appendix~\ref{app:rotation-groups} for the proof. 

\begin{property}
\label{property:d2-principal}
Let $P \in {\cal P}_n^3$ be a set of points. 
If $D_2$ acts on $P$ and we cannot distinguish the principal axis of 
(an arbitrary embedding of) $D_2$, then $\gamma(P) \succ D_2$. 
\end{property}

\begin{figure}[t]
\centering 
\includegraphics[width=2.5cm]{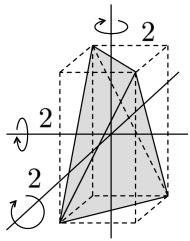}
\caption{A sphenoid consisting of $4$ congruent triangles.
Its rotation group is $D_2$.
Since the vertices are not placed equidistant positions from the three axes,
we can distinguish an axis as the principal axis from the others.}
\label{fig:sphenoid}
\end{figure}

Given a set of points $P$, $\gamma(P)$ determines the arrangement of its 
rotation axes in $P$. 
We thus use $\gamma(P)$ and its arrangement in $P$ interchangeably. 
If $\gamma(P) = C_k$ ($k > 1$), the single rotation axis of $C_k$ has 
a ``direction'' in the sense that $P$ is asymmetric against $b(P)$.
(Otherwise, $\gamma(P)$ is $D_k$.) 
An example is when $P$ forms a pyramidal frustum with regular $k$-gon 
bases which we cannot rotate to exchange the top and bottom bases. 
We say the single rotation axis of $C_k$ is {\em oriented}. 
The secondary axes of $D_{\ell}$ ($\ell > 2$) are oriented 
if $\ell$ is odd, otherwise not oriented. 
This is because the only rotation axis perpendicular to each 
secondary axis is the principal axis that has a $\pi$ rotation 
if $\ell$ is even, and there is no such rotation axis, if $\ell$ is
odd. 
Additionally, the $3$-fold axes of $T$ are oriented. 
Remember that $D_3$ is not a subgroup of $T$. 
On the other hand, $2$-fold axis of $T$ are not oriented 
because $D_2$ is a subgroup of $T$. 
The rotation axes of $O$ and $I$ are not oriented 
because each rotation axes has at least one $2$-fold axes 
that is perpendicular to it. 

\medskip

\noindent{\bf Embedding of a rotation group.~} 
We define an embedding of a rotation group to an arrangement of its supergroup. 
For two groups $G, H \in {\mathbb S}$, 
an {\em embedding} of $G$ to $H$ is an embedding of each rotation 
axis of $G$ to one of the rotation axes of $H$ 
so that any $k$-fold axis of $G$ overlaps a $k'$-fold axis of $H$ 
satisfying $k|k'$ with keeping the arrangement. 
If rotation axes of $H$ are oriented, 
the corresponding rotation axes of $G$ should keep the orientation. 
If the rotation axes of $H$ are not oriented, 
we do not care for the orientation of $G$. 
For example, we can embed $T$ to $O$ so that each $3$-fold axis of $T$ 
overlaps a $3$-fold axis of $O$ and each $2$-fold 
axis of $T$ overlaps a $4$-fold axis of $O$. 
Note that there may be many embeddings of $G$ to $H$. 
There are six embeddings of $C_4$ to $O$ depending on the choice of 
the $4$-fold axis and its orientation. 
Observe that we can embed $G$ to $H$ if and only if $G \preceq H$. 

\medskip

\noindent{\bf Transitivity.~}
We say that a set of points $P$ is {\em transitive} regarding a rotation
group $G$ if it is an orbit of $G$ through some seed point $s$, i.e.,
$P= Orb(s) = \{g*s \in P: g \in G\}$
for some $s \in P$.\footnote{For a transitive set of
points $P$, arbitrary point $s \in P$ can be a seed point.}
For a transitive set of points $P$ and any $p \in P$,
we call
$\mu(p) = |\{ g \in G : g*s = p \}|$ 
the {\em folding} of $p$.\footnote{
In group theory, 
the folding of a point $P$ is the size of 
the stabilizer of $p$ defined by
$G(p) = \{ g \in G: g*p = p \}$~\cite{A88}.}
We of course count the identity element of $G$ for $\mu(p)$ 
and $\mu(p) \geq 1$ holds for all $p \in P$. 
If $p \in P$ is at $b(P)$, its folding is $|\gamma(P)|$ and
if $p$ is on a $k$-fold axis of $\gamma(P)$,
its folding is $k$.
We have the following lemma. 
\begin{lemma}
\label{lemma:fulding}{\cite{YUKY15}}
Let $P$ be the transitive set of points generated 
by a rotation group $G \in {\mathbb S}$ 
and a seed point $s \in \Real^3$. 
If $p \in P$ is on a $k$-fold axis of $G$ for some $k$,
so are the other points $q \in P$
and $\mu(p) = \mu(q) = k$ holds.
Otherwise, if $p \in P$ is not on any axis of $G$, 
so are the other points $q \in P$,
and $\mu(p) = \mu(q) = 1$ holds.
\end{lemma}

Hence $\mu(p)$ for $p \in P$ is identical for 
a transitive set of points $P$ generated by a rotation group $G$ 
and a seed point $s$. 
We abuse $\mu$ to a transitive set of points $P$ 
and $\mu(P)$ represents $\mu(p)$ for $p \in P$. 
When $\mu(P) >1$, the positions of points of $P$ is uniquely
determined in the arrangement of $G$ if we ignore uniform
scalings that keep the center of $G$. 
Additionally, we have $|P| = |G|/\mu$. 
Table~\ref{table:vt-sets} shows the set of 
points generated by the five kinds of rotation groups. 

\begin{table}[t]
\begin{center}
\caption{The folding of seed points and transitive sets of points}
\label{table:vt-sets}
\begin{tabular}{c|c|c|c|p{6cm}}
\hline 
Rotation Group & Order & Folding & Cardinality & Polyhedron \\
\hline 
Any $G$ & $|G|$ & $|G|$ & $1$ & Point \\
\hline 
\multirow{2}{*}{$C_k$} & \multirow{2}{*}{$k$} & $k$ & 1 & Point \\ 
\cline{3-5}
 & & 1 & $k$ & Regular $k$-gons \\
\hline 
\multirow{3}{*}{$D_2$} & \multirow{3}{*}{4} & 2 & 2 & Line \\ 
\cline{3-5}
 & & 1 & 4 & Regular tetrahedron, infinitely many sphenoids, infinitely 
many rectangles \\ 
\hline 
\multirow{3}{*}{$D_{\ell}$} & \multirow{3}{*}{$2{\ell}$} & $\ell$ & 2 & Line  \\ 
\cline{3-5}
 & & 2 & $\ell$ & Regular $\ell$-gon \\ 
\cline{3-5}
 & & 1 & $2\ell$ & Infinitely many polyhedra \\ 
\hline 
\multirow{3}{*}{$T$} & \multirow{3}{*}{12} & 3 & 4 & Regular tetrahedron \\ 
\cline{3-5}
 & & 2 & 6 & Regular octahedron \\
\cline{3-5}
 & & 1 & 12 & Infinitely many polyhedra \\ 
\hline 
\multirow{4}{*}{$O$} & \multirow{4}{*}{24}& 4 & 6 & Regular octahedron \\
\cline{3-5}
 & & 3 & 8 & Cube \\ 
\cline{3-5}
 & & 2 & 12 & Cuboctahedron \\ 
\cline{3-5}
 & & 1 & 24 & Infinitely many polyhedra \\ 
\hline 
\multirow{4}{*}{$I$} & \multirow{4}{*}{60} & 5 & 12 & Regular icosahedron \\
\cline{3-5}
 &  & 3 & 20 & Regular dodecahedron \\ 
\cline{3-5}
 &  & 2 & 30 & Icosidodecahedron \\ 
\cline{3-5}
 &  & 1 & 60 & Infinitely many polyhedra \\ 
\hline
\end{tabular}
\end{center}
\end{table} 

\medskip

\noindent{\bf $\gamma(P)$-decomposition of $P$.~} 
Yamauchi et al. showed that a set of points $P$ can be decomposed into
transitive subsets~\cite{YUKY15}: 
For a point $p \in P$, let $Orb(p) = \{g * p \in P: g \in \gamma(P)\}$
be the orbit of the group action of $\gamma(P)$ through $p$.
Then we let $\{P_1, P_2, \ldots, P_m\} = \{Orb(p) : p \in P\}$ be
its orbit space. From the definition, each $P_i$ is transitive
regarding $\gamma(P)$ and $\{P_1, P_2, \ldots, P_m\}$ is a partition of
$P$.
Such partition is unique and
we call it the {\em $\gamma(P)$-decomposition} of $P$. 
Clearly, each element of the $\gamma(P)$-decomposition
is one of the polyhedron for $\gamma(P)$ 
shown in Table~\ref{table:vt-sets}.
Note that the sizes of the elements of the $\gamma(P)$decomposition of
$P$ may be different. 
The $\gamma(P)$-decomposition of $P$
does not depend on the local coordinate systems and 
each robot can recognize it.

For the $\gamma(P)$-decomposition $\{P_1, P_2, \ldots, P_m\}$ of $P$, 
we denoted the ball centered at $b(P)$ and contains $P_i$ on it
by $Ball(P_i)$ for each $1 \leq i \leq m$. 

In~\cite{YUKY15}, the authors showed the following theorem. 

\begin{theorem}{\cite{YUKY15}}
\label{theorem:decomposition}
Let $P \in {\cal P}_n^3$ be a configuration of robots recognized as a
set of points. 
Then $P$ can be decomposed into subsets $\{P_1, P_2, \ldots , P_m\}$  
in such a way that each $P_i$ is a transitive set of points
regarding $\gamma(P)$. 
 Furthermore, the oblivious FSYNC robots
 can agree on a total ordering among the
elements of the $\gamma(P)$-decomposition of $P$. 
\end{theorem}

To let the robots agree on the total ordering among the elements of the
$\gamma(P)$-decomposition of $P$, 
the authors introduced the ``local view'' of each robot, 
which is determined by $P$ 
independently of the local coordinate systems so that 
each robot $r_i$ can compute the local view of $r_j \in R$ 
although $r_i$ observes $P$ in its local coordinate system $Z_i$.
We will briefly describe how the robots compute local views.
The local view of robot $r_i \in R$ is constructed by considering
the innermost empty ball $I(P)$ as the earth and line
$\overline{p_i b(P)}$ as the earth's axis,
where $p_i$ is the position of $r_i$.
Then the positions of each robot is represented by its amplitude,
longitude, and latitude.
To determine the meridian, $r_i$ selects a
robot nearest to $I(P)$ whose projection on $I(P)$
determines the meridian. 
If there are multiple candidates for a meridian robot,
$r_i$ selects one of them that minimizes $r_i$'s local view,
which is defined as follows: 
The local view of $r_i$ is an $n$-tuple of positions 
of robots where the first element is the position of $r_i$,
the second element is the position of the meridian robot, 
and the positions of the remaining robots are  
sorted in the increasing order. 
Because all local coordinate systems are right-handed, 
each robot can compute the local view of all robots. 
Then the robots agree on the lexicographic order of local view
that guarantees the following properties: 
For a set of points $P$ and its $\gamma(P)$-decomposition
$\{P_1, P_2, \ldots, P_m\}$, 
\begin{enumerate}
\item
All robots in $P_i$ have the same local view 
for each $i = 1, 2, \ldots , m$.
\item
Any two robots, one in $P_i$ and the other in $P_j$, 
have different local views for all $i \not= j$.
\end{enumerate}
The robots agree on a total ordering among the elements of the 
$\gamma(P)$-decomposition of $P$ by the ordering among local views 
so that the ordering satisfies the following properties. 
\begin{property}
\label{property:ordering}
Let $P \in {\cal P}_n^3$ and $\{P_1, P_2, \ldots, P_m\}$
be a configuration of robots recognized as a set of points and 
its $\gamma(P)$-decomposition of $P$. 
Then, the local view defined in~\cite{YUKY15} 
guarantees the following properties: 
\begin{enumerate}
\item $P_1$ is on $I(P)$ and $P_m$ is on $B(P)$. 
\item For each $P_i$ and $P_{i+1}$, all points in $P_{i+1}$ are  
on or in the exterior of $Ball(P_i)$.
\end{enumerate} 
\end{property}

Finally, we consider a decomposition of a set of points $P$
by a rotation group $G \preceq \gamma(P)$.
Given an embedding of $G$ into $\gamma(P)$,
we consider the orbit of $G$ through each element $p \in P$ and
the orbit space $\{Orb(p) :p \in P\} = \{P_1, P_2, \ldots, P_{\ell}\}$. 
Clearly, $\{P_1, P_2, \ldots, P_{\ell}\}$ is a partition of $P$ and
each subset in the family is a transitive set of points
regarding $G$.
We call this decomposition the {\em $G$-decomposition} of $P$. 
We will use such decomposition when we discuss impossibility
in Section~\ref{sec:def-sym}.

\subsection{Rotation group of local coordinate systems} 
\label{sec:def-sigma}

We introduce the rotation group of local coordinate systems of
robots. 
Of course each robot recognizes neither the local coordinate systems
of other robots nor the rotation group of them by 
just observing the positions of the robots.
We use this notion when we discuss impossibility. 

We denote an arrangement of local coordinate systems by
a set of four-tuples 
$Q = \{(p_i, x_i, y_i, z_i) : r_i \in R\}$ where
$p_i$ represents the position of $r_i \in R$ in $Z_0$ and
$x_i, y_i, z_i$ are the positions 
$(1,0,0)$, $(0,1,0)$, and $(0,0,1)$ of $Z_i$ observed in $Z_0$. 
An arrangement of local coordinate systems encodes the positions of
the robots since the current position of the robot is the
origin of its local coordinate system. 
We also use the set of points $P = \{p_1, p_2, \ldots, p_n\}$ to
denote the positions of robots of $Q$.  

We consider rotations on $Q$ that 
produces the same arrangement of local coordinate systems. 
\begin{definition}
Let $Q$ be a set of local coordinate systems of robots. 
The rotation group $\sigma(Q)$ of $Q$
is the rotation group in ${\mathbb S}$
that acts on $Q$ and none of its proper supergroup in ${\mathbb S}$
acts on $Q$. 
 \end{definition} 

Clearly, $\sigma(Q)$ is uniquely determined for any set of local
coordinate systems $Q$ and 
it also determines the arrangement of rotation axes of $\sigma(Q)$ in
$Q$ that decomposes $Q$ into disjoint subsets
by the group action of $\sigma(Q)$.
We call this decomposition {\em $\sigma(Q)$-decomposition} of $Q$. 
We focus on the decomposition of $P$ by $\sigma(Q)$ rather than
the decomposition of $Q$ by $\sigma(Q)$. 
When it is clear from the context, we denote $\sigma(Q)$ by
$\sigma(P)$ and the $\sigma(Q)$-decomposition of $Q$ by
$\sigma(P)$-decomposition of $P$. 

Clearly we have the following property.
\begin{property}
\label{prop:sigma-observation}
Let $P \in {\cal P}_n^3$ and $\{P_1, P_2, \ldots, P_{\ell}\}$
be a configuration of robots recognized as set of points
and its $\sigma(P)$-decomposition. 
For each $P_i$ ($1 \leq i \leq \ell$), 
the robot forming $P_i$ have the same local observation. 
\end{property}

We show several relation between $\gamma(P)$ and $\sigma(P)$. 
The following property is clear from the definition. 
 \begin{property}
  \label{prop:sub-sigma}
 Let $Q = \{(p_i, x_i, y_i, z_i) : r_i \in R\}$ and
 $P = \{p_i : r_i \in R\}$
 be an arbitrary set of $n$ local coordinate systems and the set of
 $n$ positions of robots where $p_i$ is the position of $r_i \in R$. 
 Then we have $\sigma(P) (= \sigma(Q)) \preceq \gamma(P)$,
 thus there is an embedding of $\sigma(P)$ to $\gamma(P)$. 
 \end{property}

If a $k$-fold axis of $\sigma(P)$ contains a point of $P$,
there should be $(k-1)$ other robots on that point so that
we can apply rotations around the axis.
When $P$ is a set of points, it does not contain such multiplicity. 
\begin{property}
\label{prop:size-sigma}
Let $P \in {\cal P}_n^3$ and $\{P_1, P_2, \ldots, P_{\ell}\}$
be a configuration of robots recognized as a set of points 
and its $\sigma(P)$-decomposition, respectively. 
Then, we have $|P_i| = |\sigma(P)|$ for each $1 \leq i \leq \ell$. 
\end{property}

\subsection{Rotation group of a multiset of points}
\label{subsection:rho-multiset}

Though the initial configuration and the target pattern contain
no multiplicity,
when we prove impossibility, 
we should consider executions of any arbitrary pattern formation algorithm
that may generate multiplicity. 
In this section, we extend the rotation group and the symmetricity
to a multiset of points. 

Let $P \in {\cal P}_n^3$ be a multiset of $n$ points. 
Intuitively, the rotation group $\gamma(P)$ of $P$ and
the rotation group $\sigma(P)$ of local coordinate systems of $P$ 
are straightforward generalization of those for a set of points.
We consider rotations that keep the positions and multiplicities
of $P$.\footnote{We do not define the 
decomposition of a multiset of points $P$ by $\gamma(P)$ or
$\sigma(P)$ in this section, but we need careful treatment
as shown in Section~\ref{sec:concl}. 
Additionally, the robots cannot agree on the ordering of the
elements of the $\gamma(P)$-decomposition of $P$
since the second condition of Property~\ref{property:ordering}
does not always hold because of multiplicities. }
\begin{definition}
Let $P \in {\cal P}_n^3$ be a multiset of points. 
The rotation group $\gamma(P)$ of $P$  
is the rotation group in ${\mathbb S}$ that acts on $P$ and 
none of its proper supergroup in ${\mathbb S}$ acts on $P$. 
\end{definition}

\begin{definition}
Let $Q$ be a multiset of local coordinate systems. 
The rotation group $\sigma(Q)$ of 
$Q$ is the rotation group ${\mathbb S}$ that acts on $Q$
and none of its proper supergroup in ${\mathbb S}$
acts on $Q$. 
\end{definition} 

Clearly, we have Lemma~\ref{prop:sigma-observation} and
Lemma~\ref{prop:sub-sigma} for the rotation group of
multiset of points. 

We will show the first impossibility result.
In the following, for a configuration of robots recognized as
a (multi-)set of points $P$, we use the robot $r_i \in R$
and its position $p_i \in P$ interchangeably. 
\begin{lemma}
\label{lemma:sigma-impossibility}
Let $P \in {\cal P}_n^3$ be an arbitrary initial configuration of 
oblivious FSYNC robots. 
For an arbitrary deterministic algorithm $\psi$ and its execution 
$P(0)(=P), P(1), P(2), \ldots$, 
we have $\sigma(P(t)) \succeq \sigma(P)$ for any $t \geq 0$, 
thus $\gamma(P(t)) \succeq \sigma(P)$. 
\end{lemma}
\begin{proof}
Let $P$ and $\{P_1, P_2, \ldots, P_{\ell}\}$ be an initial configuration 
and its $\sigma(P)$-decomposition of $P$.
We consider an execution $P(0)(=P), P(1), P(2), \ldots$ 
of an arbitrary algorithm $\psi$. 
We focus on an arbitrary element $P_i$ and $p_j \in P_i$.
Because $P$ is a set of $n$ points,  
from Property~\ref{prop:size-sigma}, 
for any $p_k \in P_i$, there exists an element 
$g_k \in \sigma(P)$ that satisfies $g_k * p_j = p_k$ and 
$g_k \neq g_{k'}$ if $p_k \neq p_{k'}$. 
We will show that the movement of each $p_k \in P_i$ is 
symmetric regarding $\sigma(P)$ and the robots of $P_i$ 
keep the rotation axes of $\sigma(P)$ in $P(1)$. 

Consider the Compute phase at time $0$ and let 
$\psi(Z_j(P(0))) = d_j$ at $p_j$. 
From Property~\ref{prop:sigma-observation}, 
each robot $p_k \in P_i$ have the same local observation and 
$\psi(Z_k(P(0))) = \psi(Z_j(P(0))) = d_j$ at $p_k$. 
Because $p_k = g_k * p_j$,
$Z_0(\psi(Z_k(p(0)))) = g_k * Z_0(\psi(Z_j(P(0))))$ 
and after the movement the positions of robots that formed $P_i$ 
are symmetric regarding the same arrangement of $\sigma(P)$. 
Additionally, the local coordinate system of these robots 
are still symmetric regarding the same arrangement of $\sigma(P)$. 
Let $P_i(1) \subseteq P(1) $ be the positions of robots 
that formed $P_i$ in $P(0)$. 
Hence, we have $\sigma(P_i(1)) = \sigma(P(0))$. 
Because this property holds for all $P_i$ ($1 \leq i \leq \ell$), 
we have $\sigma(P(1)) \succeq \sigma(P(0))$. 
Note that $P(1)$ can be a multiset of points. 

By repeating the above discussion, 
we can show that $\sigma(P(t)) \succeq \sigma(P(0))$ for any 
$t \geq 1$, thus 
$\gamma(P(t)) \succeq \sigma(P(0))$ for any $t \geq 1$. 
\qed 
\end{proof}

We now consider the non-oblivious version of
Lemma~\ref{lemma:sigma-impossibility}. 

\begin{lemma}
\label{lemma:sigma-impossibility-memory}
Let $P \in {\cal P}_n^3$ be an arbitrary initial configuration of 
non-oblivious FSYNC robots. 
For an arbitrary deterministic algorithm $\psi$ and its execution 
$P(0)(=P), P(1), P(2), \ldots$, 
we have $\sigma(P(t)) \succeq \sigma(P)$ for any $t \geq 0$, 
thus $\gamma(P(t)) \succeq \sigma(P)$. 
\end{lemma}
\begin{proof}
 Let $P$ and $\{P_1, P_2, \ldots, P_{\ell}\}$
 be an initial configuration where the content of
 local memory at all robots are identical and 
 the $\sigma(P)$-decomposition of $P$, respectively. 
Let $P(0)(=P), P(1), P(2), \ldots$ be an execution of 
an arbitrary algorithm $\psi$. 
In the same way as the proof of Lemma~\ref{lemma:sigma-impossibility}, 
the robots forming each $P_i$ ($i \in \{1, 2, \cdots, \ell\}$)
obtain the same local observation and same output of a common
algorithm in $P(0)$, and their next positions keep the same arrangement
of $\sigma(P)$. 
Hence in $P(1)$, their memory contents remain identical and 
these robots obtain the same local observation, 
and same output of a common algorithm. 
By repeating this discussion, robots can never break 
$\sigma(P)$ in any $P(t)$ ($t \geq 0$). 
\qed 
\end{proof}

\subsection{Symmetricity of a set of points} 
\label{sec:def-sym} 

We start with the following observation. 
Consider a set of points $P$ that forms a cube, thus $\gamma(P) = O$. 
We can embed $D_4$ to $\gamma(P)$ by selecting one $4$-fold axis of
$\gamma(P)$ as the principal axis and the other two $4$-fold axes of
$\gamma(P)$ as the secondary axes.
The remaining two secondary axes overlap the $2$-fold axes of
$\gamma(P)$.
The $D_4$-decomposition of $P$ consists of two elements and 
we construct an arrangement of local coordinate systems of robots
so that $\sigma(P) = D_4$ by selecting one point $p$ for each element of
the $D_4$-decomposition of $P$, 
fixing its local coordinate system, 
and applying the rotations of $D_4$ to $p$'s local
coordinate system. 
From Lemma~\ref{lemma:sigma-impossibility}, 
the robots cannot reduce their rotation group to any subgroup of $D_4$
from such initial arrangement of local coordinate systems.
Consequently, robots cannot eliminate an arbitrary rotation group
$G \preceq \gamma(P)$ 
that produces $|G|$ symmetric local coordinate systems regarding $G$. 
Based on this observation, we define the {\em symmetricity} of a 
configuration as follows. 

\begin{definition}
\label{def:symmetricity}
Let $P \in {\cal P}_n^3$ be a set of points. 
The symmetricity $\varrho(P)$ of $P$
is the set of rotation groups $G \in {\mathbb S}$ 
that acts on $P$ and there exists an embedding of $G$ to
$\gamma(P)$ such that each element of $G$-decomposition of $P$ 
is a $|G|$-set. 
\end{definition}

We define $\varrho(P)$ as a set because the ``maximal'' rotation
group that satisfies the definition is not uniquely determined. 
Maximality means that there is no proper supergroup in
${\mathbb S}$ that satisfies the condition of
Definition~\ref{def:symmetricity}. 
When it is clear from the context, we denote $\varrho(P)$ by
the set of such maximal elements. 
For example, if $P$ forms a regular icosahedron, 
$\varrho(P) = \{C_1, C_2, C_3, D_2, D_3, T\}$ and we 
denote it by $\varrho(P) = \{D_3, T\}$.
(See Figure~\ref{fig:subgroups}.)
From the definition, $\varrho(P)$ always contains $C_1$ and 
if $G \in \varrho(P)$, $\varrho(P)$ contains every element of
${\mathbb S}$ that is a subgroup of $G$.

Because $G \in \varrho(P)$ acts on $P$, $G$ is a subgroup of $\gamma(P)$ 
and any initial configuration $P$ is a set of $n$ points, 
we can rephrase the above definition as follows: 
For an initial configuration $P$,
$\varrho(P)$ is the set of rotation groups $G \in {\mathbb S}$
that has an embedding to unoccupied rotation axes of $\gamma(P)$ 
and if all rotation axes of $\gamma(P)$ is occupied, 
$\varrho(P) = \{C_1\}$.

From Lemma~\ref{lemma:sigma-impossibility}, 
we will show that for any $G \in \varrho(P)$, 
there exists an arrangement of local coordinate systems of $P$ 
that satisfies $\sigma(P) = G$ and
does not allow the robots to reduce their rotation group to 
a subgroup of $G$. 
\begin{lemma}
\label{lemma:rho-conf}
Let $P$ be an arbitrary initial configuration of oblivious FSYNC robots. 
For each $G \in \varrho(P)$, 
there exists an arrangement of local coordinate systems of $P$ 
such that 
for an arbitrary algorithm and its execution 
$P(0)(=P), P(1), P(2), \ldots$, 
$\gamma(P(t)) \succeq G$ for all $t \geq 0$. 
\end{lemma}
\begin{proof}
From the definition, there exists an embedding of $G$ to 
the unoccupied rotation axes of $\gamma(P)$.
We select one of such embeddings and 
construct an initial arrangement of local coordinate systems for
$P(=P(0))$ that satisfies $\sigma(P) = G$:  
Let $\{P_1, P_2, \ldots, P_{\ell}\}$ be the $G$-decomposition of $P$
for the embedding of $G$. 
For each element $P_i$ ($i = 1, 2, \cdots, \ell$), 
we choose a point $p \in P_i$ and its local coordinate system, 
then apply the rotations of $G$. 
We obtain an arrangement of local coordinate systems
with $\sigma(P) = G$.

From Lemma~\ref{lemma:sigma-impossibility}, 
for any execution $P(0)(=P), P(1), P(2), \ldots$  of an arbitrary
algorithm, 
$\gamma(P(t)) \succeq G$ for all $t \geq 0$. 
\qed
\end{proof}

In the same way as Lemma~\ref{lemma:sigma-impossibility-memory},
local memory at robots does not help the robots
break the symmetricity.
Hence we have the following non-oblivious version of
Lemma~\ref{lemma:rho-conf}. 

\begin{lemma}
\label{lemma:rho-conf-memory}
Let $P$ be an arbitrary initial configuration of 
non-oblivious FSYNC robots. 
For each $G \in \varrho(P)$, 
there exists an arrangement of local coordinate systems of $P$ 
such that 
for an arbitrary algorithm and its execution 
$P(0)(=P), P(1), P(2), \ldots$, 
$\gamma(P(t)) \succeq G$ for all $t \geq 0$. 
\end{lemma}

\section{Algorithm to show $\varrho(P)$} 
\label{sec:show-sym}

As shown in Section~\ref{subsection:rho-multiset},
the symmetricity of an initial configuration $P$ consists of the
rotation groups that the robots can never break. 
However there exists an initial configuration $P$ such that
$\varrho(P)$ does not contain $\gamma(P)$. 
For example, when $P$ forms a cube,
$\gamma(P) = O$ while
$\varrho(P) = \{D_4\}$. 
In this case, there seems to be a possibility that the robots can
reduce the symmetry of their positions. 
In this section, we will show that the robots can reduce symmetry of
their positions and
agree on some rotation group $G \in \varrho(P)$
with a very simple algorithm. 

In 2D-space, the robots can reduce symmetry of their
positions only when there is a robot on the center of the 
smallest enclosing circle of their positions, 
i.e., on the single rotation axis of their cyclic group. 
We will show that in 3D-space, in the same way,
robots can reduce the rotation group of their positions 
by leaving rotation axes. 
For example, consider a configuration $P$ where robots form a 
regular pyramid. Hence $\gamma(P)=C_k$ if the base is a regular
$k$-gon. 
If the robot at the apex leaves the single rotation axis, 
the rotation group of the new configuration is $C_1$ that
matches the symmetricity of the initial configuration. 
We will show such simple movement reduces the rotation group of 
the configuration.

\begin{lemma}
\label{lemma:show-sym}
There exists an algorithm $\psi_{SYM}$ for oblivious FSYNC robots
that 
translates an arbitrary initial configuration $P$ 
to a configuration $P'$ that satisfies $\gamma(P') \in \varrho(P)$. 
\end{lemma}

The proposed algorithm is based on the ``go-to-center'' algorithm in 
\cite{YUKY15}, that considers symmetry breaking in an initial
configuration $P$ whose $\gamma(P)$-decomposition
contains a regular tetrahedron, regular octahedron,
a cube, a regular dodecahedron, or an icosidodecahedron. 
When the robots form these polyhedra, 
they are on some rotation axes of $\gamma(P)$ and the
``go-to-center'' algorithm sends the robots to some point not on
any rotation axis. Because the number of robots is less than
$|\gamma(P)|$, the positions of robots are not transitive regarding
$\gamma(P)$ and the rotation group of any resulting configuration
is no more $\gamma(P)$. 
We apply the ``go-to-center'' algorithm to all configurations 
where some rotation axes of $\gamma(P)$ are occupied.
Specifically, we add a cuboctahedron and a regular icosahedron 
to the target of the algorithm and analyze the rotation group
of resulting configurations.
We also add similar symmetry breaking procedures for configurations with
a 2D rotation group. 
We will show that 
as is expected from the results in 2D-space, 
when some rotation axes of $\gamma(P)$ are occupied, 
the robots on the rotation axes can remove the rotation axes
by leaving their current positions, 
thus the symmetricity of a resulting configuration is a proper subgroup of 
the rotation group of the previous configuration. 
By repeating this procedure, the robot system reaches a configuration $P'$
with $\gamma(P') \in \varrho(P)$.
We first analyze the configuration after one-step execution of the
``go-to-center'' algorithm when the robots form one of the above seven
(semi-)regular polyhedra in Section~\ref{subsec:base} and
then we consider any initial configuration where some of its
rotation axes are occupied in 
Section~\ref{subsec:composition}. 

Let $\{P_1, P_2, \ldots, P_m\}$ be the $\gamma(P)$-decomposition of
an initial configuration $P$. 
We focus on the elements that consist of points on rotation axes 
of $\gamma(P)$. In other words, 
we focus not on the coordinates of each point but 
on $\gamma(P)$ and the folding of each element of its
$\gamma(P)$-decomposition. 
We denote a polyhedron generated by a rotation group $G$ 
and a seed point $s$ with folding $\mu$ as $U_{G, \mu}$. 
Table~\ref{table:basic-symmetricity} shows the list of $U_{G, \mu}$ 
for $G \in {\mathbb S}$ and the symmetricity of 
$U_{G, \mu} \cup U_{G, 1}$ as an example. 
For example, a configuration $U_{O, 1} \cup U_{O, 3}$ consists
of a truncated cube and a cube with a common arrangement of $O$. 
We note that we need $U_{G, 1}$ for the 2D rotation groups and
when $G$ is a 3D rotation group,
$\varrho(U_{G, \mu})$ and
$\varrho(U_{G, 1} \cup U_{G, \mu})$ are identical. 

\begin{table}[t]
\begin{center}
\caption{Symmetricity of $U_{G, \mu} \cup U_{G, 1}$.} 
\label{table:basic-symmetricity}
\begin{tabular}{c|c|l|c|c}
\hline 
Rotation group & Folding & Polyhedra formed by $U_{G,\mu}$ &
 Notation & $\varrho(U_{G, 1} \cup U_{G, \mu})$\\ 
\hline 
Any $G \in {\mathbb S}$ & $|G|$ & Point at the center & $U_{G,|G|}$ & $\{C_1\}$ \\ 
\hline 
\multirow{2}{*}{$C_k$} & $k$ & Point on the single rotation axis & $U_{C_k,k}$ & $\{C_1\}$ \\ 
\cline{2-5}
 & $1$ & Regular $k$-gon & $U_{C_k,1}$ & $\{C_k\}$ \\ 
\hline 
\multirow{4}{*}{$D_2$} & $2$ & Line on a rotation axis & $U_{D_2,2}$ & $\{C_2\}$ \\ 
\cline{2-5}
 & \multirow{3}{*}{1} & Regular tetrahedron, & \multirow{3}{*}{$U_{D_2, 1}$} & 
 \multirow{3}{*}{$\{D_2\}$} \\ 
 & & infinitely many sphenoids, and & &  \\ 
 & & infinitely many rectangles & & \\ 
\hline 
\multirow{5}{*}{$D_{\ell}$} & $\ell$ & Line on the principal axis & 
$U_{D_{\ell},{\ell}}$ & $\{C_2\}$ \\ 
\cline{2-5}
 & \multirow{3}{*}{$2$} & Regular $\ell$-gon perpendicular to & \multirow{3}{*}{$U_{D_{\ell}, 2}$} & 
\multirow{3}{*}{$\{C_{\ell}, D_{\ell/2}\}$ if $\ell$ is even, $\{C_{\ell}\}$ otherwise}\\ 
 & & the principal axis and & & \\ 
 & & containing the center & & \\ 
 \cline{2-5}
 & $1$ & Infinitely many polyhedra & $U_{D_{\ell}, 1}$ & $\{D_{\ell}\}$ \\ 
\hline 
\multirow{3}{*}{$T$} &  $3$ & Regular tetrahedron & $U_{T, 3}$ & 
$\{D_2\}$ \\ 
\cline{2-5}
 & $2$ & Regular octahedron & $U_{T, 2}$ & $\{D_3\}$ \\ 
\cline{2-5}
 & $1$ & Infinitely many polyhedra & $U_{T, 1}$ & $\{T\}$\\ 
\hline 
\multirow{4}{*}{$O$} & $4$ & Regular octahedron & $U_{O, 4}$ & $\{D_3\}$ \\ 
\cline{2-5}
 & $3$ & Cube & $U_{O, 3}$ & $\{D_4\}$\\ 
\cline{2-5}
 & $2$ & Cuboctahedron & $U_{O, 2}$ & $\{T, C_4, C_3\}$\\ 
\cline{2-5}
 & $1$ & Infinitely many polyhedra  & $U_{O, 1}$ & $\{O\}$ \\ 
\hline 
\multirow{4}{*}{$I$} & $5$ & Regular icosahedron & $U_{I, 5}$ & $\{T, D_3\}$\\ 
\cline{2-5}
 & $3$ & Regular dodecahedron & $U_{I, 3}$ & $\{D_5, D_2\}$\\ 
\cline{2-5}
 & $2$ & Icosidodecahedron & $U_{I, 2}$ & $\{C_5, C_3\}$\\ 
\cline{2-5}
 & $1$ & Infinitely many polyhedra & $U_{I, 1}$ & $\{I\}$\\
\hline
\end{tabular}
\end{center}
\end{table} 

We first note that there is no way for oblivious FSYNC robots to 
reduce $\gamma(P)$ of an initial configuration $P$ 
when $P$ forms a regular $n$-gon. 
Here, $\gamma(P) = D_n$ and $\varrho(P) = \{C_n, D_{n/2}\}$ 
if $n$ is even, 
$\varrho(P) = \{C_n\}$ otherwise. 
Consider the case where $n$ is even. 
To show the symmetricity, 
the robots either show an orientation of the single rotation axis 
or divide themselves into two groups to form $U_{D_{n/2}, 1}$. 
However, when $\sigma(P) = C_n$, from Lemma~\ref{lemma:rho-conf}, 
the rotation group of robots is $C_n$ forever, thus they 
keep some regular $n$-gon forever. 
The robots neither show an agreement on the orientation of the 
principal axis nor divide themselves into two groups. 
When the robots are oblivious, they do not remember the previous 
trials without recognizing $\sigma(P)=C_n$ and 
they keep on trying to show their symmetricity forever. 
We have the same situation for odd $n$.
We avoid this infinite trials by leaving a regular $n$-gon as it is.
Hence the proposed algorithm do nothing
when $P$ forms a regular polygon. 
This is not a problem for Theorem~\ref{theorem:main} 
since for such $P$, the target pattern $F$ satisfies 
$C_n, D_{n/2} \in \varrho(F)$ and hence $\gamma(F) \succeq D_n$ 
and the robots do not need to break the symmetry.

\subsection{Transitive set of points}
\label{subsec:base}

We start with a symmetry breaking algorithm for a transitive initial 
configuration regarding a 3D rotation group, i.e, 
we consider $U_{G, \mu}$ for $G \in \{T, O, I\}$ and $\mu>1$, 
a regular tetrahedron, a regular octahedron, a cube, a cuboctahedron, 
a regular icosahedron, a regular dodecahedron, and an icosidodecahedron
(Figure~\ref{fig:seven-polyhedra}). 
The proposed algorithm is based on the ``go-to-center'' algorithm
of \cite{YUKY15} as shown in Algorithm~\ref{alg:go-to-center}. 
If a current configuration forms one of the above seven 
polyhedra, Algorithm~\ref{alg:go-to-center} makes
each robot select an adjacent face 
and approach the center of the selected face, 
but stops it $\epsilon$ before the center.
There are two restrictions, i.e., when the robots form a
cuboctahedron or an icosidodecahedron,
they select one face from adjacent regular triangle faces or
adjacent regular pentagon faces, respectively.
The aim of the algorithm is to put the robots
on some points that are not on any rotation axis.
The role of $\epsilon$ is to gather the robots
around some rotation axis and we fix $\epsilon$ to
$\ell/100$ for the simplicity of the proposed algorithm 
where $\ell$ is the length of an edge of the (semi-)regular polyhedron
that the robots initially form. 

\begin{figure}[t]
\centering 
\subfigure[$U_{T, 3}$]{\includegraphics[width=2cm]{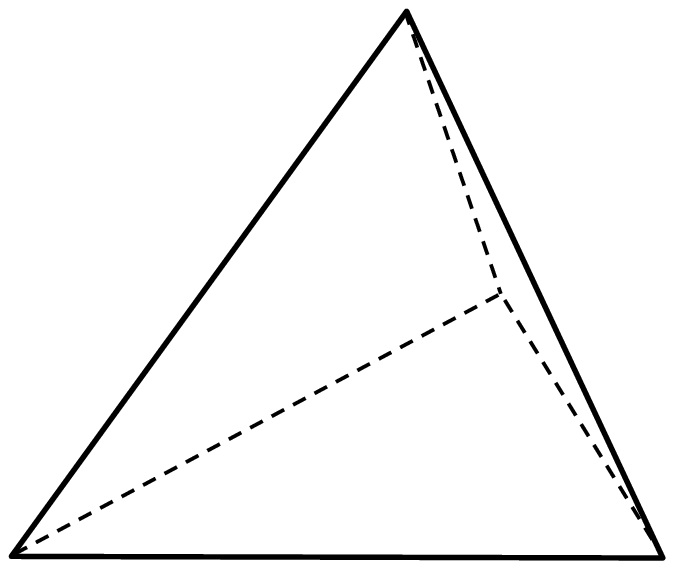}\label{fig:U-T3}}
\hspace{1mm}
\subfigure[$U_{O, 4}$]{\includegraphics[width=2cm]{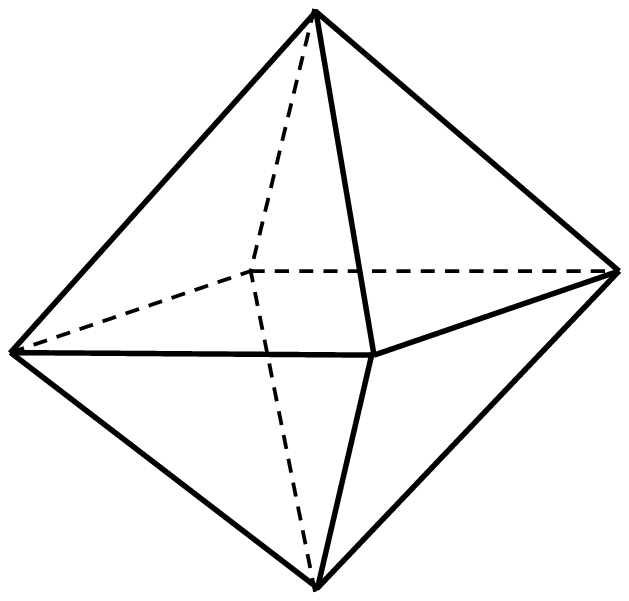}\label{fig:U-O4}}
\hspace{1mm}
\subfigure[$U_{O, 3}$]{\includegraphics[width=2cm]{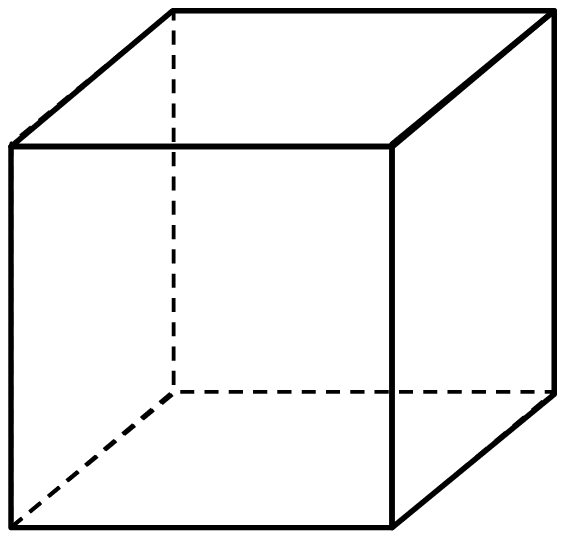}\label{fig:U-O3}}
\hspace{1mm}
\subfigure[$U_{O, 2}$]{\includegraphics[width=2cm]{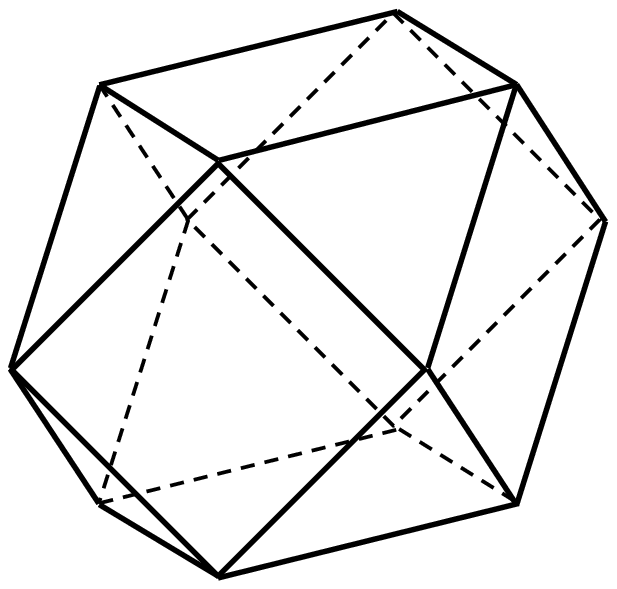}\label{fig:U-O2}}
\hspace{1mm}
\subfigure[$U_{I, 5}$]{\includegraphics[width=2cm]{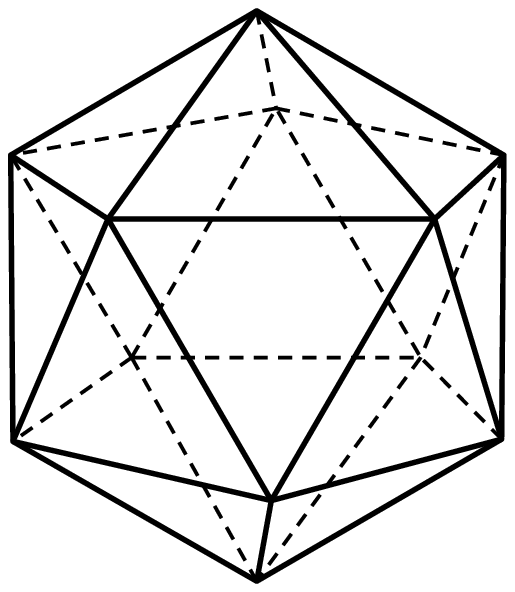}\label{fig:U-I5}}
\hspace{1mm}
\subfigure[$U_{I, 3}$]{\includegraphics[width=2cm]{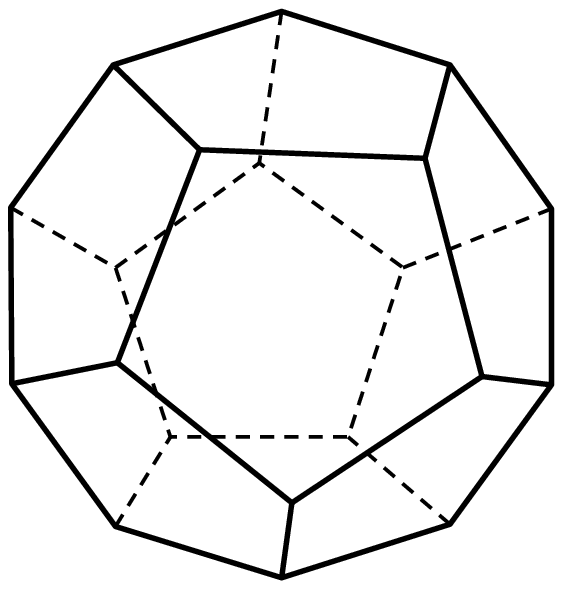}\label{fig:U-I3}}
\hspace{1mm}
\subfigure[$U_{I, 2}$]{\includegraphics[width=2cm]{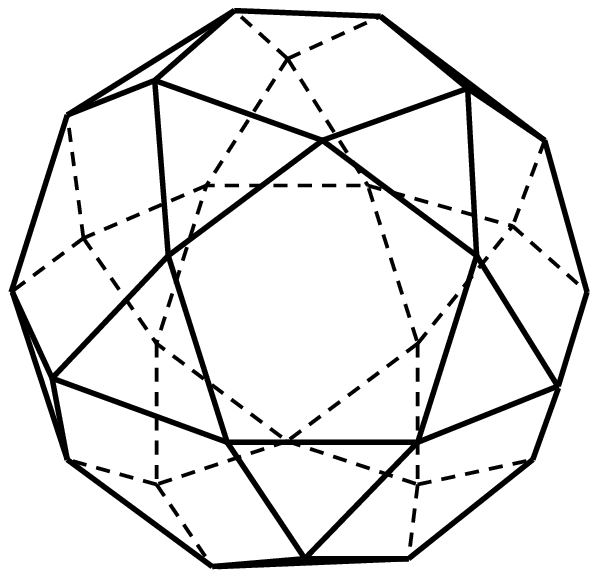}\label{fig:U-I2}}
\caption{Seven polyhedra}
\label{fig:seven-polyhedra}
\end{figure}

\begin{algorithm}[t]
\caption{\texttt{Go-to-center}$(P)$ for robot $r_i \in R$}
\label{alg:go-to-center} 
\begin{tabbing}
xxx \= xxx \= xxx \= xxx \= xxx \= xxx \= xxx \= xxx \= xxx \= xxx
\kill 
{\bf Notation} \\ 
\> $P$: The positions of robots forming an $U_{G, \mu}$ for 
$G \in \{T, O, I\}$ and $\mu>1$ observed in $Z_i$.\\ 
\> $p_i$: Current position of $r_i$. \\ 
\> $\epsilon$: $\ell/100$ where $\ell$ is the length of an edge 
of the polyhedron that $P$ forms. \\ 
\\ 
{\bf Algorithm}  \\ 
\> {\bf Switch} ($P$) {\bf do} \\ 
\> \> {\bf Case} cuboctahedron: \\ 
\> \> \> Select an adjacent triangle face. \\ 
\> \> \> Destination $d$ is the point $\epsilon$ before the center of 
the selected face on the line from $p_i$ to the center. \\ 
\> \> {\bf Case} icosidodecahedron: \\ 
\> \> \> Select an adjacent pentagon face. \\ 
\> \> \> Destination $d$ is the point $\epsilon$ before the center of 
the selected face on the line from $p_i$ to the center. \\ 
\> \> {\bf Default}: \\ 
\> \> \> Select an adjacent face. \\ 
\> \> \> Destination $d$ is the point $\epsilon$ before the center of 
the selected face on the line from $p_i$ to the center. \\ 
\> {\bf Enddo} 
\end{tabbing}
\end{algorithm}

\begin{lemma}
\label{lemma:show-symm-single}
Let $P$ be an arbitrary initial configuration of oblivious FSYNC robots 
that forms a 
$U_{G, \mu}$ for $G \in \{T, O, I\}$ and $\mu > 1$. 
One step execution of Algorithm~\ref{alg:go-to-center}
translates $P$ to 
another configuration $P'$ that satisfies 
$\gamma(P') \in \varrho(P)$. 
\end{lemma}
\begin{proof}
Let $P$, $P'$ be an initial configuration that forms one of the 
above seven (semi-)regular polyhedra and a configuration obtained by 
one-step execution of Algorithm~\ref{alg:go-to-center}. 
We will show that $\gamma(P') \in \varrho(P)$. 

The proof follows the same idea as~\cite{YUKY15}. 
Let $D$ be the set of all points that can be selected by the robots 
as their next positions in $P$. 
When $P$ is a regular polyhedron, 
the points of $D$ are placed around the vertices of the dual of $P$,
which we call the {\em base polyhedron} for $D$. 
For example, when $P$ is a cube,
the base polyhedron is a regular octahedron (Figure~\ref{fig:cocta}). 
When $P$ is a cuboctahedron (an icosidodecahedron, respectively), 
the next positions are placed around the $3$-fold 
axes of $O$ (the $5$-fold axes of $I$, respectively). 
In this case, we consider a cube (a regular icosahedron, respectively) 
as its base polyhedron.

\begin{figure}[t]
\centering 
\includegraphics[width=15cm]{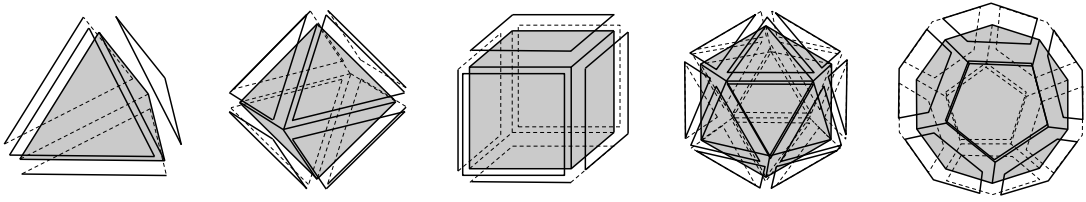}
\caption{Expansion of base polyhedra.}
\label{fig:expansion}
\end{figure}

\begin{figure}[t]
\centering 
\subfigure[$\epsilon$-expanded tetrahedron]
{\includegraphics[width=3.5cm]{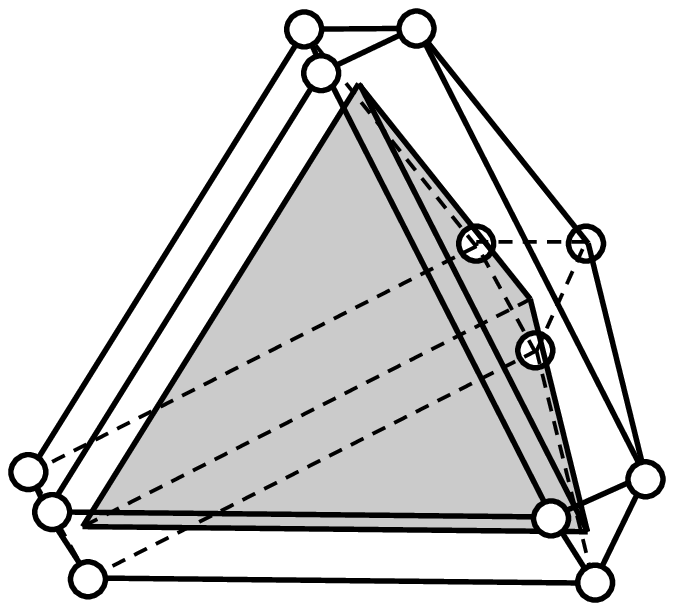}\label{fig:ctetra}}
\hspace{3mm}
\subfigure[$\epsilon$-expanded cube]
{\includegraphics[width=3.5cm]{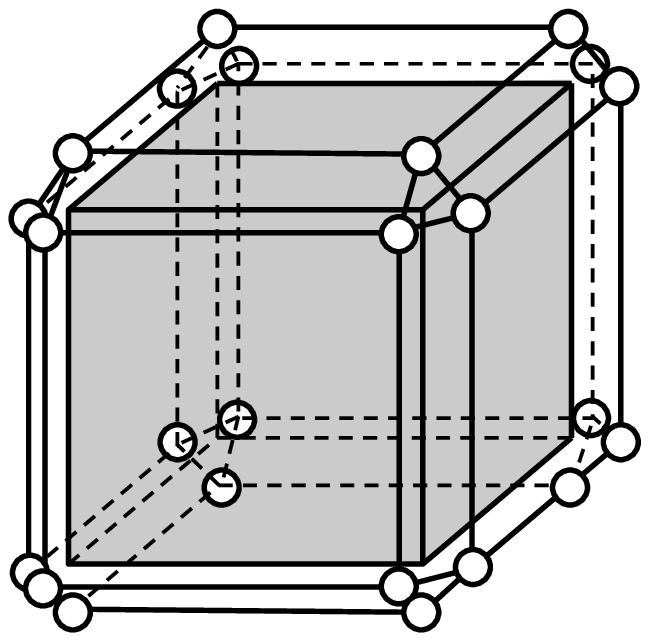}\label{fig:ccube}}
\hspace{3mm}
\subfigure[$\epsilon$-expanded octahedron]
{\includegraphics[width=3.5cm]{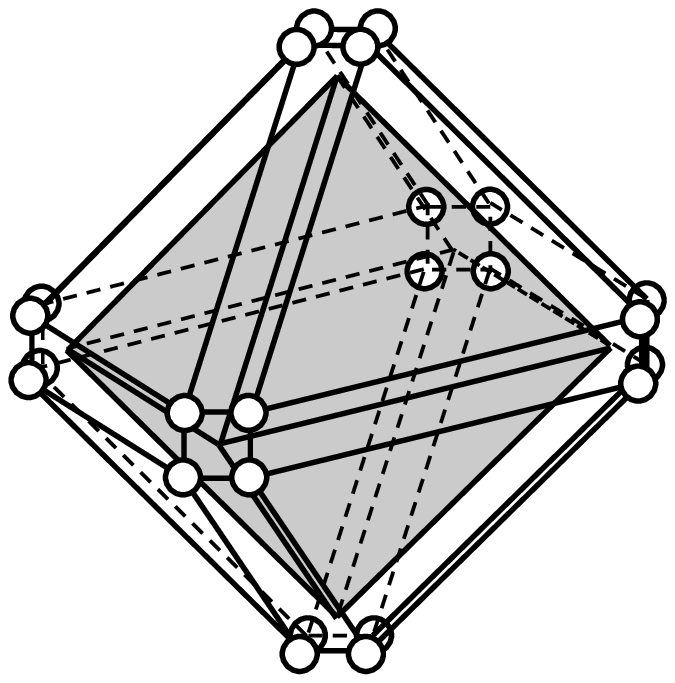}\label{fig:cocta}}
\hspace{3mm}
\subfigure[$\epsilon$-truncated cube]
{\includegraphics[width=3.5cm]{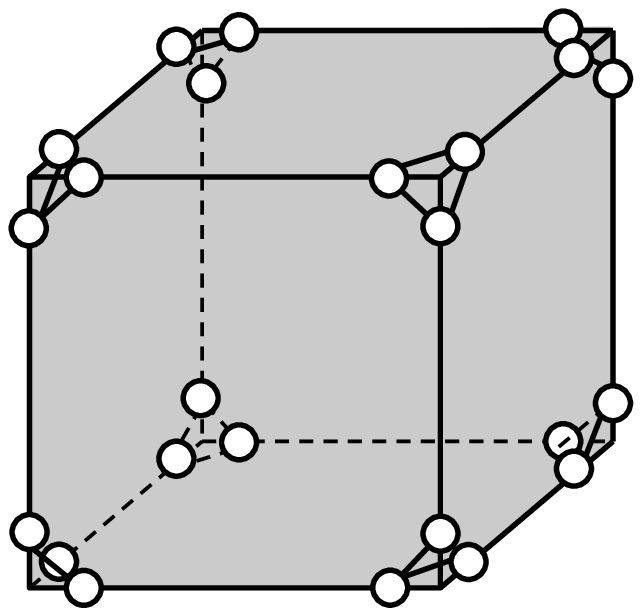}\label{fig:tcube}}
\\ 
\subfigure[$\epsilon$-expanded dodecahedron]
{\includegraphics[width=4cm]{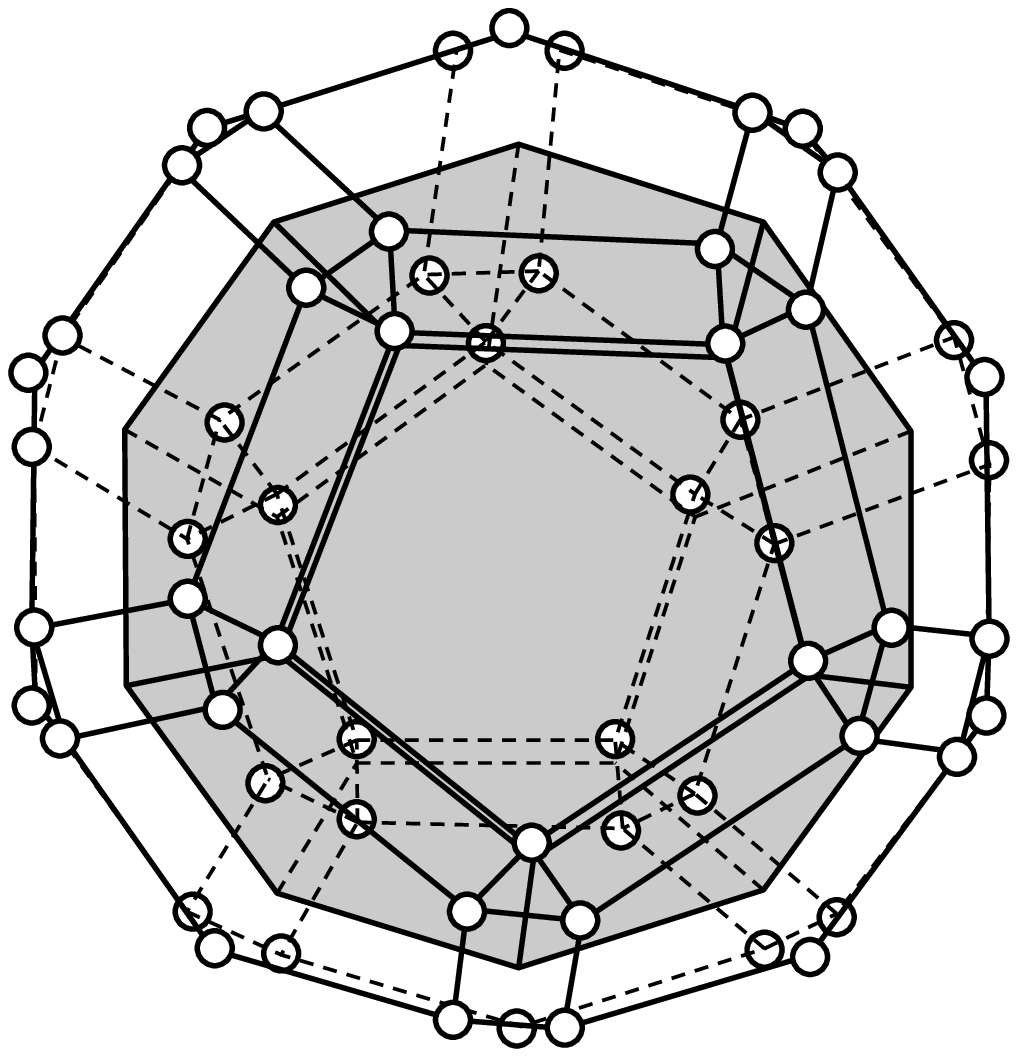}\label{fig:cdodeca}}
\hspace{3mm}
\subfigure[$\epsilon$-expanded icosahedron]
{\includegraphics[width=4cm]{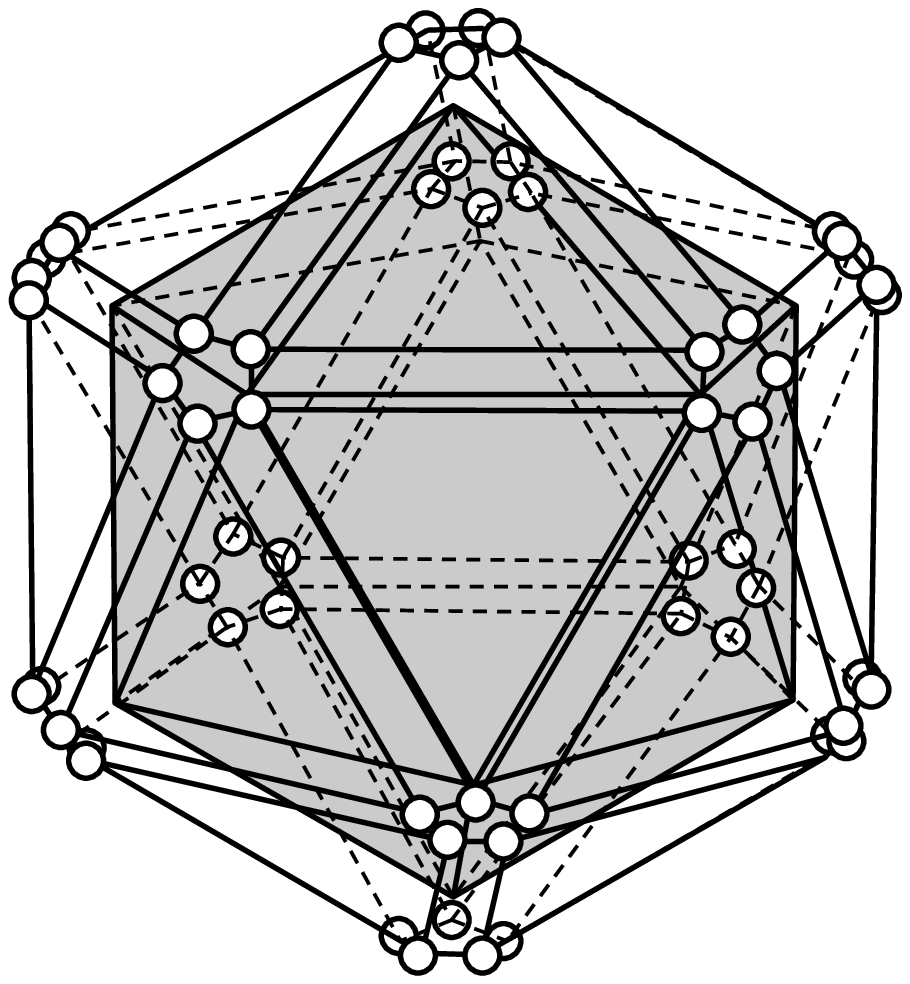}\label{fig:cicosa}}
\hspace{3mm}
\subfigure[$\epsilon$-truncated icosahedron]
{\includegraphics[width=4cm]{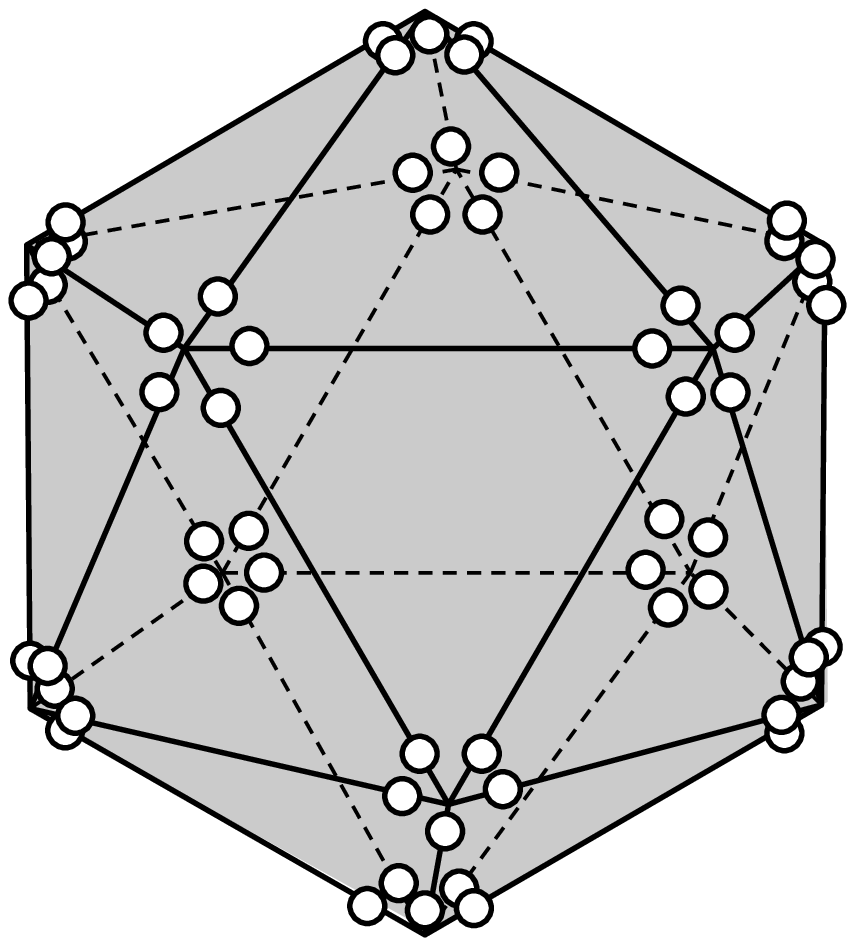}\label{fig:ticosa}}
\caption{Candidate set $D$ corresponding to $P$.}
\label{fig:fulldest}
\end{figure}
 
Figure~\ref{fig:fulldest} shows the base polyhedron and the set of
possible destinations $D$ 
for each of the seven initial configurations. 
When $P$ is a regular polyhedron, 
$D$ forms a polyhedron obtained 
by moving each face of the base polyhedron away from the center 
with keeping the center.
Then the obtained polyhedron consists of the faces of the base
polyhedron 
and new faces formed by the separated vertices and the separated
edges of the base 
polyhedron.\footnote{ 
 The operation is also known as {\em cantellation}: 
 the convex hull of $D$ is obtained from the base polyhedron 
 by truncating the vertices and beveling the edges. See \cite{C73}. 
 }
See Figure~\ref{fig:expansion}. 
For example, when $P$ is a cube, $D$ forms a polyhedron 
obtained by a regular octahedron with the above operation 
and we call the polyhedron {\em $\epsilon$-expanded octahedron} 
(Figure~\ref{fig:cocta}). 
In the same way, if $P$ is a regular tetrahedron, 
$D$ forms an 
{\em $\epsilon$-expanded tetrahedron} (Figure~\ref{fig:ctetra}), 
if $P$ is a regular octahedron, 
$D$ forms an 
{\em $\epsilon$-expanded cube} (Figure~\ref{fig:ccube}), 
if $P$ is a regular icosahedron, 
$D$ forms an 
{\em $\epsilon$-expanded dodecahedron} (Figure~\ref{fig:cdodeca}), 
and
if $P$ is a regular dodecahedron, 
$D$ forms an 
{\em $\epsilon$-expanded icosahedron} (Figure~\ref{fig:cicosa}). 
On the other hand, when $P$ is a semi-regular polyhedron, 
$D$ forms a polyhedron obtained by cutting the vertices 
of the base polyhedron. 
For example, when $P$ is a cuboctahedron, 
$D$ is a polyhedron obtained from a cube by cutting its vertices 
and we call the polyhedron an {\em $\epsilon$-truncated cube} 
(Figure~\ref{fig:tcube}). 
If $P$ is a icosidodecahedron, 
$D$ forms an {\em $\epsilon$-truncated icosahedron} 
(Figure~\ref{fig:ticosa}). 
It is worth emphasizing that $D$ is a transitive set of points
regarding $\gamma(P)$, thus it is spherical. 
 
Algorithm~\ref{alg:go-to-center} makes the robots select 
a subset of size $|P|$ from $D$. To prove the correctness, 
we will show that the symmetricity of any such subset 
has the rotation group that satisfies the statement. 

Our basic idea is to check all possibilities of $\gamma(P')$, 
specifically, for each $G \not\in \varrho(P)$, 
we assume that $\gamma(P') = G$ and 
check the $\gamma(P')$-decomposition of $P'$. 
From Algorithm~\ref{alg:go-to-center}, 
 any resulting configuration $P'$ contains multiplicity
 since the sets of possible next positions of different robots are
 disjoint. 

\noindent{\bf Case A. $P$ is a regular tetrahedron:~} 
$D$ forms an $\epsilon$-expanded tetrahedron 
and we check the rotation group of any $4$-set of $D$. 
We will show $\gamma(P') \in \varrho(P) = \{D_2\}$. 

Assume that $\gamma(P')$ is $D_k$ or $C_k$ for some $k \geq 3$. 
Because we cannot find any regular $\ell$-gon for $\ell \geq 4$ in $D$, 
$k \leq 3$. 

Assume $\gamma(P') = D_3$, then the $D_3$-decomposition 
of $P'$ consists of $U_{D_3, 2}$ (cardinality $3$) and 
$U_{D_3, 6}$ (cardinality 1)  
because $P'$ consists of four points. 
We do not have the case where the $D_3$-decomposition of $P'$
consists of two $U_{D_3, 3}$'s (i.e., all points of $P'$ are on the
principal axis of $D_3$), because $\gamma(P') = D_{\infty}$. 
Thus $P'$ contains a regular triangle. 
Figure~\ref{fig:tri-etetra} shows all possible regular triangles 
in an $\epsilon$-expanded tetrahedron and 
we cannot find any regular triangle that have one 
point of $D$ at its center. 
Hence $\gamma(P') \neq D_3$. 

Assume $\gamma(P') = C_3$, then the $C_3$-decomposition 
of $P'$ consists of $U_{C_3, 1}$ (cardinality $3$) and 
$U_{C_3, 3}$ (cardinality $1$), because otherwise 
$\gamma(P')$ is not $C_3$. 
Thus $P'$ contains a regular triangle. 
In the same way, for any regular triangle in $D$, 
there is no point on the $3$-fold axis for the 
triangle. Hence $\gamma(P') \neq C_3$.
Consequently, if $\gamma(P')$ is a 2D rotation group,
it is $C_1$, $C_2$, or $D_2$.  

\begin{figure}[t]
\centering 
\includegraphics[width=14cm]{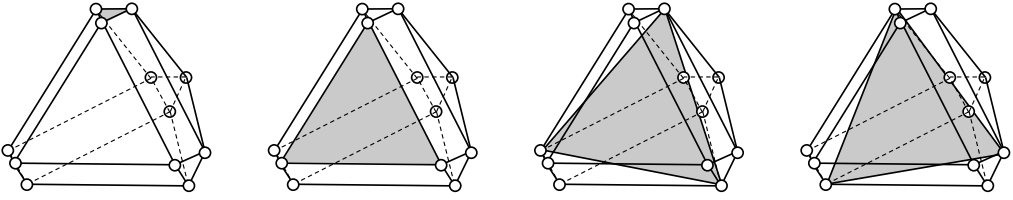}
\caption{Regular triangles in an $\epsilon$-expanded tetrahedron. }
\label{fig:tri-etetra}
\end{figure}

We check the 3D rotation groups. 
First, $\gamma(P') \neq T$ since any polyhedron with 
rotation group $T$ consists of more than four points and 
there is no $4$-set that forms a
regular tetrahedron (i.e., the base polyhedron) in 
an $\epsilon$-expanded tetrahedron. 
Second, $\gamma(P')$ is neither $O$ nor $I$ since 
any polyhedron with rotation group $O$ (or $I$) consists of 
at least $6$ ($12$, respectively) vertices. 
(See Table~\ref{table:vt-sets}.)

Consequently, $\gamma(P') \preceq D_2$.

\noindent{\bf Case B. $P$ is a regular octahedron:~} 
$D$ forms an $\epsilon$-expanded cube 
and we check the rotation group of any $6$-set of $D$. 
We will show $\gamma(P') \in \varrho(P) = \{D_3\}$. 

Assume that $\gamma(P)$ is $D_k$ or $C_k$ for some $k \geq 4$. 
Because we cannot find any regular $\ell$-gon for $\ell \geq 5$ in $D$, 
$k \leq 4$. 

Assume $\gamma(P') = D_4$, then the $D_4$-decomposition of $P'$ 
consists of $U_{D_4, 2}$ (cardinality $4$) and 
$U_{D_4, 4}$ (cardinality $2$) 
because $P'$ consists of six points. 
We do not have the case where $D$ consists of three $U_{D_4, 4}$'s 
since $P\gamma(P') = D_{\infty}$. 
Thus $P'$ contains a square. 
Figure~\ref{fig:sqr-ecube} shows all possible squares in an 
$\epsilon$-expanded cube and 
we cannot find any square that have two points 
on its $4$-fold axis. 
Hence $\gamma(P') \neq D_4$. 

Assume $\gamma(P') = C_4$, then the $C_4$-decomposition of $P'$ 
consists of $U_{C_4, 1}$ (cardinality $4$) and 
two $U_{C_4, 4}$'s (cardinality $1$). 
Thus $P'$ contains a square. 
In the same way, for any square in $D$, 
there is no point on the $4$-fold axis for the square. 
Hence $\gamma(P') \neq C_4$. 

\begin{figure}[t]
\centering 
\includegraphics[width=10.5cm]{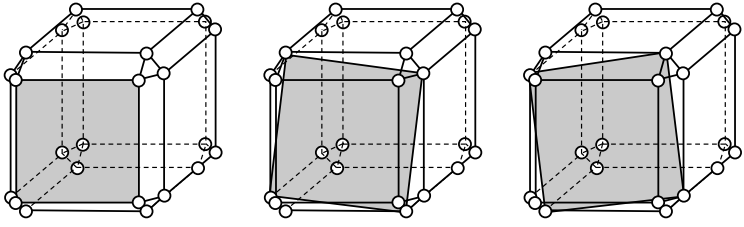}
\caption{Square in an $\epsilon$-expanded cube.}
\label{fig:sqr-ecube}
\end{figure}

\begin{figure}[t]
\centering 
\subfigure[]
{\includegraphics[width=3cm]{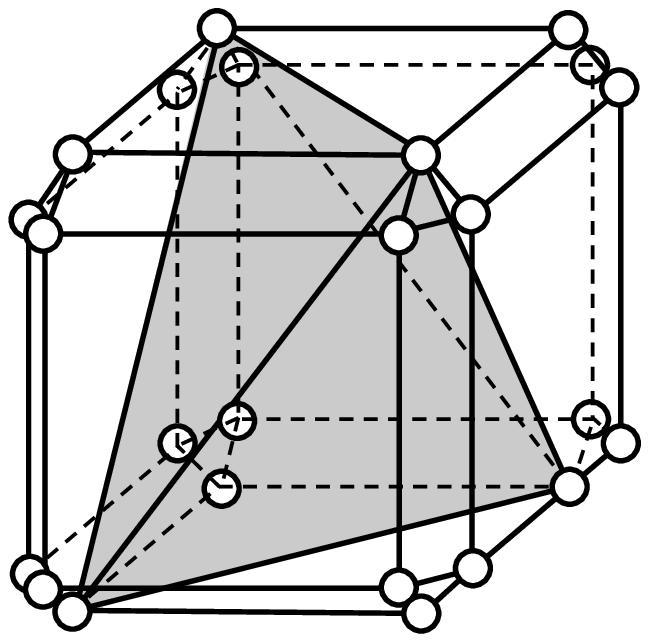}\label{fig:sphenoid-ecube}}
\hspace{3mm}
\subfigure[]
{\includegraphics[width=9cm]{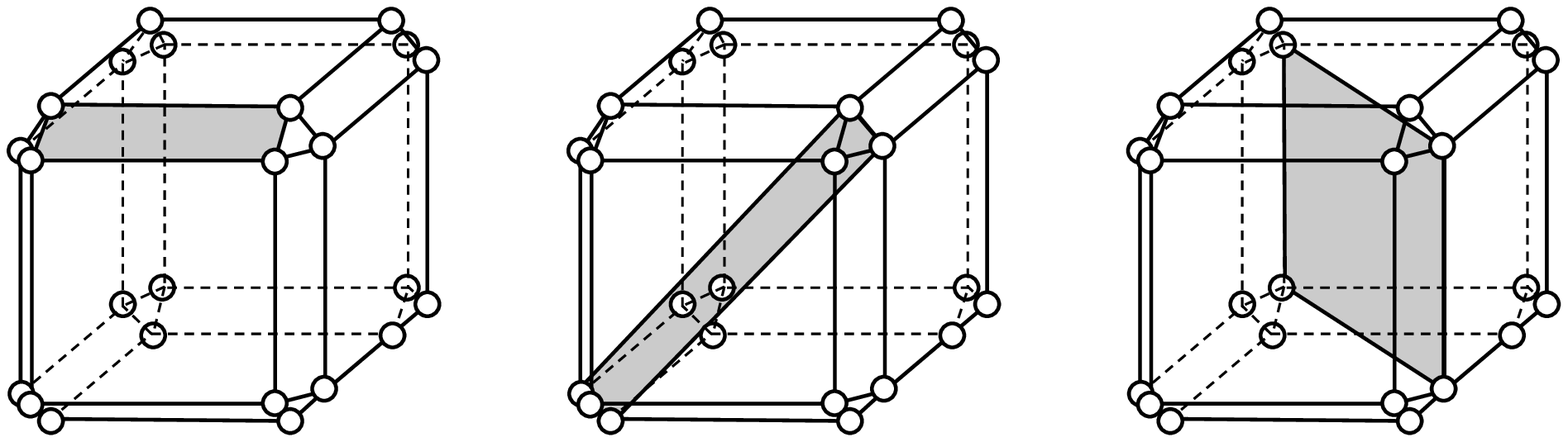}\label{fig:req-ecube}}
\caption{An example of $U_{D_2, 1}$ 
 in an $\epsilon$-expanded cube.
 (a) A sphenoid. (b) A rectangle. }
\label{fig:D2-ecube}
\end{figure}

Assume $\gamma(P') = D_2$, then the $D_2$-decomposition of $P'$ 
consists of (i) $U_{D_2, 1}$ (cardinality $4$) and 
$U_{D_2, 2}$ (cardinality $2$) or (ii)
three $U_{D_2, 2}$'s (cardinality $2$). 
We first check case (i) where
$U_{D_2, 1}$ forms a sphenoid, a regular tetrahedron, 
a rectangle, or a square. 
See Figure~\ref{fig:D2-ecube} as an example.  
Now consider $\epsilon \rightarrow 0$. 
If $U_{D_2, 1}$ forms a sphenoid or a regular tetrahedron,
when $\epsilon=0$, $U_{D_2, 1}$ forms a regular 
tetrahedron in the base polyhedron. 
(See Figure~\ref{fig:basecube-0}.)
However, there is no vertex of the base polyhedron 
on the $2$-fold axes of the regular tetrahedron.
When $\epsilon > 0$, the arrangement of edges of $U_{D_2, 1}$ 
slightly moves from the regular tetrahedron,
but we cannot find two points forming $U_{D_2, 2}$ on the
$2$-fold axes of $U_{D_2, 1}$. 

In the same way, if $U_{D_2, 1}$ forms a rectangle or a square, 
when $\epsilon=0$, $U_{D_2, 1}$ forms a line formed by two
vertices of the base polyhedron or a rectangle formed by
four vertices of the base polyhedron
(See Figure~\ref{fig:basecube-1} and Figure~\ref{fig:basecube-2}.) 
There are three possibilities for a line;
edges of the base polyhedron,
diagonals of the faces of the base polyhedron, or
diagonals of the base polyhedron (connecting two opposite vertices). 
In any of the three cases,
there is no pair of vertices of the base polyhedron
that form a perpendicular bisector or a line on it. 
When $\epsilon > 0$, the arrangement of edges of $U_{D_2, 1}$
slightly moves from the line, 
but we cannot find two points forming $U_{D_2, 2}$ on the
$2$-fold axes of $U_{D_2, 1}$. 
There are two possibilities for a square;
faces of the base polyhedron or the rectangle cutting the base
polyhedron into two triangular prisms.
In both cases, there are no pair of vertices of the base polyhedron
that are on some $2$-fold axis or $4$-fold axis
of these rectangles. 
When $\epsilon > 0$, the arrangement of edges of $U_{D_2, 1}$
slightly moves from the rectangles, 
but we cannot find two points forming $U_{D_2, 2}$ on the
$2$-fold axes of $U_{D_2, 1}$. 
 
We check case (ii) where $P'$ consists of three $U_{D_2, 2}$'s.
Because $D$ is spherical, 
there is no three points of $D$ that are on the same line. 
Hence the three $U_{D_2, 2}$'s are perpendicular to each other and
intersects at their midpoints, i.e.,
these three lines are in the interior of $D$. 
Consider $\epsilon \rightarrow 0$.
When $\epsilon = 0$, three $U_{D_2, 2}$'s degenerates to three
lines formed by the vertices of the base polyhedron and
perpendicular to each other. 
However we cannot find any three lines perpendicular to each other 
in a cube. (See Figure~\ref{fig:basecube-3}.)
When $\epsilon > 0$, the arrangement of lines slightly
moves from the rectangles, but we cannot find any three
lines perpendicular to each other.  
 
From case (i) and (ii), we have $\gamma(P') \neq D_2$. 

\begin{figure}[t]
\centering 
\subfigure[]
{\includegraphics[width=2.5cm]{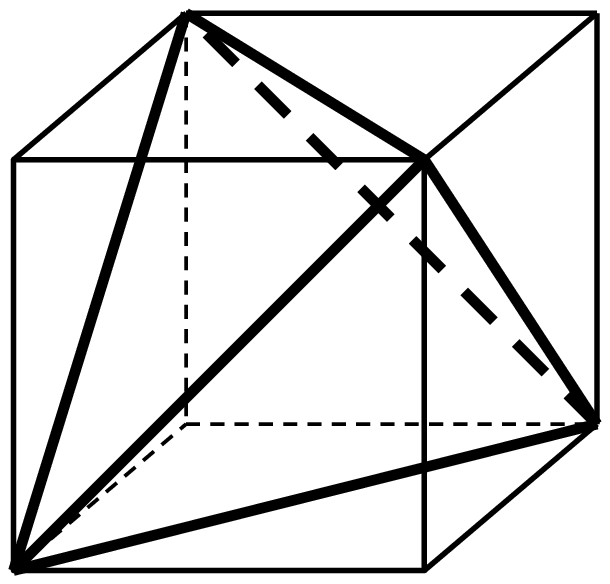}\label{fig:basecube-0}}
\hspace{3mm}
\subfigure[]
{\includegraphics[width=2.5cm]{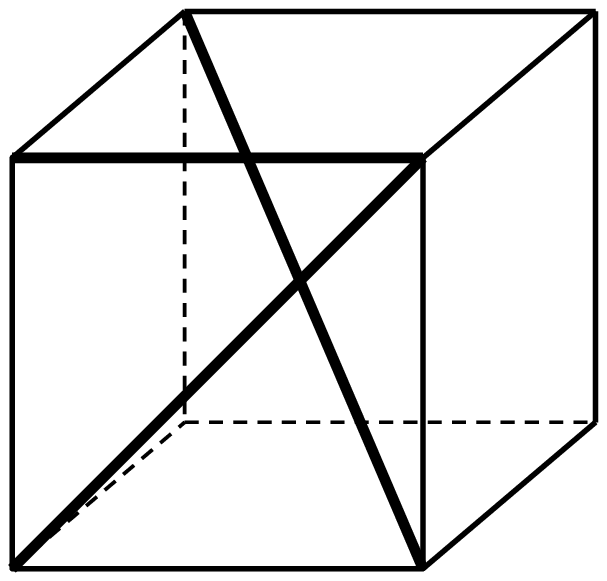}\label{fig:basecube-1}}
\hspace{3mm}
\subfigure[]
{\includegraphics[width=2.5cm]{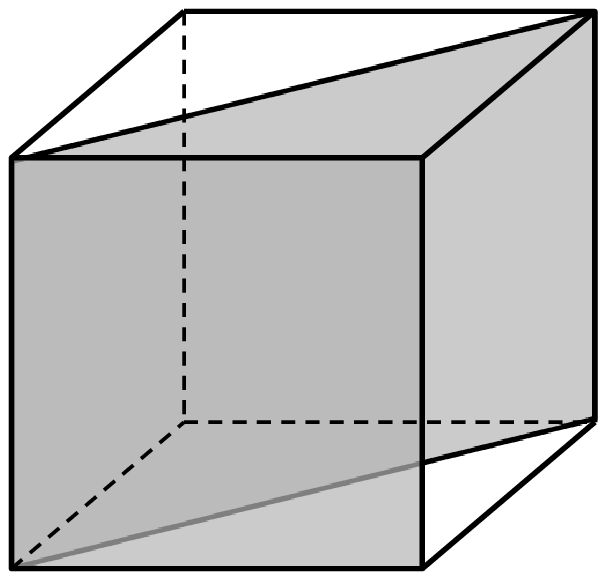}\label{fig:basecube-2}}
\hspace{3mm}
\subfigure[]
{\includegraphics[width=2.5cm]{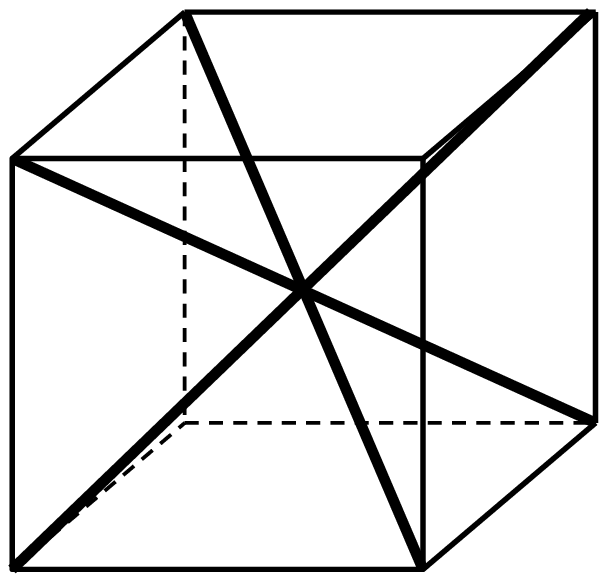}\label{fig:basecube-3}}
\hspace{3mm}
 \caption{Lines and rectangles of base polyhedron cube. (a) A regular
 tetrahedron in a cube. 
 (b) We have two kinds of rectangles formed by four vertices of a cube.
 (c) Three lines intersecting one point in the cube are not
 perpendicular to each other. }
\label{fig:basecube}
\end{figure}

We check the 3D rotation groups. 
First, assume $\gamma(P') = T$, then the $T$-decomposition of $P'$ 
consists of $U_{T, 2}$ since $P'$ consists of $6$ points. 
Because all points of $D$ are near the vertices of a cube, 
$P'$ does not form a regular octahedron (the dual of the base 
polyhedron).
Hence $\gamma(P') \neq T$. 
Second, assume $\gamma(P') = O$, then the $O$-decomposition of $P'$ 
consist of $U_{O, 4}$, but as already discussed, 
we cannot find any regular octahedron in $D$. 
Finally, $\gamma(P') \neq I$ since any polyhedron with rotation group $I$ 
consists of more than $12$ vertices. 

Consequently, $\gamma(P') \preceq D_3$.

\noindent{\bf Case C. $P$ is a cube:~} 
$D$ forms an $\epsilon$-expanded octahedron and 
we check the rotation group of any $8$-set of $D$. 
We will show $\gamma(P') \in \varrho(P) = \{D_4\}$. 

Assume $\gamma(P')$ is $D_k$ or $C_k$ for some $k \geq 5$. 
Because we cannot find any regular $\ell$-gon for $\ell \geq 5$ in $D$, 
$k \leq 4$. 

Assume $\gamma(P') = D_3$, then the $D_3$-decomposition of $P'$ 
 contains at least one $U_{D_3, 3}$ (cardinality $2$)
 or $U_{D_3, 6}$ (cardinality $1$) 
since $|P'|=8$ is not divided by $3$. 
Additionally, $P'$ contains $U_{D_3, 1}$ (cardinality $6$) or 
$U_{D_3, 2}$ (cardinality $3$), 
thus at least one regular triangle. 
Figure~\ref{fig:tri-eocta} shows all possible regular triangles 
in an $\epsilon$-expanded octahedron and 
we cannot find any regular triangle 
that have a point on its $3$-fold axis. 
Hence $\gamma(P') \neq D_3$. 

Assume $\gamma(P') = C_3$, then the $C_3$-decomposition of $P'$ 
 contains at least one $U_{C_3, 3}$ (cardinality $1$) since
 $|P'|=8$. 
Additionally, $P'$ contains $U_{C_3, 1}$, 
thus at least one regular triangle. 
In the same way, for any regular triangle in $D$, 
there is no point on the $3$-fold axis of the triangle. 
Hence $\gamma(P') \neq C_3$. 

\begin{figure}[t]
\centering 
\includegraphics[width=14cm]{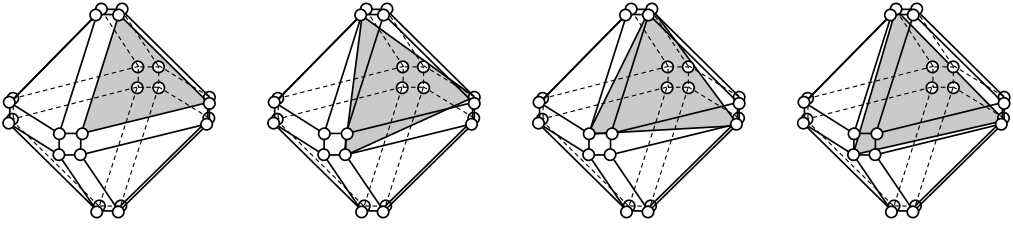}
\caption{Regular triangles in an $\epsilon$-expanded octahedron.}
\label{fig:tri-eocta}
\end{figure}

We check the 3D rotation groups. 
First, assume $\gamma(P') = T$, then the $T$-decomposition of $P'$ 
consists of two $U_{T, 3}$'s (cardinality $4$) 
because $P'$ consists of eight points. 
However, because the vertices of an $\epsilon$-expanded 
octahedron is around the vertices of a regular octahedron, 
we cannot find any regular tetrahedron in $D$. 
Hence $\gamma(P') \neq T$. 

Second, assume $\gamma(P') = O$, then the $O$-decomposition of $P'$ 
consists of one $U_{O, 3}$ (cardinality $8$), 
but an $\epsilon$-expanded octahedron does not contain any cube (its dual). 
Hence $\gamma(P') \neq O$. 

Finally, $\gamma(P') \neq I$ since any polyhedron with 
rotation group $I$ consists of at least $12$ vertices 
(See Table~\ref{table:vt-sets}). 

Consequently, $\gamma(P') \preceq D_4$.

\noindent{\bf Case D. When $P$ is a cuboctahedron:~}  
$D$ forms an $\epsilon$-truncated cube  
and we check the rotation group of any $12$-set of $D$. 
We will show $\gamma(P') \in \varrho(P) = \{T, C_4, C_3\}$. 

Assume that $\gamma(P')$ is $C_k$ or $D_k$ for some $k \geq 5$. 
Because we cannot find any regular $\ell$-gon for $\ell \geq 5$ in $D$, 
$k \leq 4$. 

Assume $\gamma(P') = D_4$, then the $D_4$-decomposition of $P'$ 
consists of 
(i) three $U_{D_4, 2}$'s (cardinality $4$), 
(ii) several $U_{D_4,4}$'s (cardinality $2$) and $U_{D_4,2}$ or $U_{D_4, 1}$ 
 (cardinality $8$) or 
(iii) one $U_{D_4, 1}$ and one $U_{D_4, 2}$ 
 (cardinality $4$).
 We do not have the case where we have six $U_{D_4, 4}$'s since
 $\gamma(P') = D_{\infty}$. 
In any case, $P'$ contains a square. 
Figure~\ref{fig:sqr-tcube} shows all possible squares in an 
$\epsilon$-truncated cube. 
We cannot find any three squares on a plane in $D$, 
 hence we do not have case (i). 
Additionally, we cannot find any square in $D$ that 
have points on its $4$-rotation axis 
and we do not have case (ii). 
Now we consider case (iii) where $P'$ contains a $U_{D_4, 1}$ 
consisting of two squares shown in 
Figure~\ref{fig:sqr-tcube}. 
$U_{D_4 ,1}$ consists two congruent squares, thus
they are on the same plane, 
or two congruent squares
or they are opposite against the center of $D$. 
$U_{D_4, 2}$ is on the $2$-fold axis of $D_4$, thus
on the plane between the two bases, 
but we cannot find any point of $D$ on such plane. 
Hence $\gamma(P') \neq D_4$. 
\begin{figure}[t]
\centering 
\includegraphics[width=10.5cm]{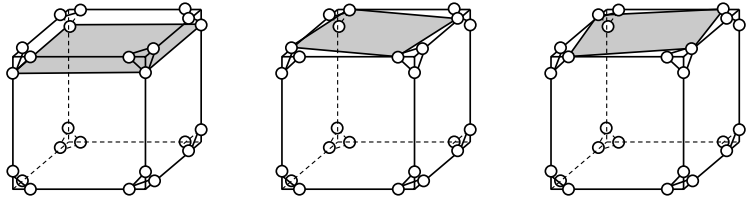} 
\caption{Squares in an $\epsilon$-truncated cube. }
\label{fig:sqr-tcube}
\end{figure}

Assume $\gamma(P') = D_3$, then the $D_3$-decomposition of $P'$ 
consists of 
(i) four $U_{D_3, 2}$'s (cardinality $3$), 
(ii) two $U_{D_3, 1}$'s (cardinality $6$), 
(iii) two $U_{D_3,2}$'s (cardinality $3$) and $U_{D_3, 1}$ 
(cardinality $6$) or 
(iv) contains $U_{D_3, 3}$ (cardinality $2$). 
In any case, $P'$ contains a regular triangle.  
Figure~\ref{fig:tri-tcube} shows all possible regular triangles 
in the $\epsilon$-truncated cube. 
We first note that we cannot find any four regular triangles on the 
same plane, hence we do not have case (i). 
Next, consider case (iii). 
Two $U_{D_3, 2}$'s are on the same plane 
and $U_{D_3, 1}$ is symmetric against that plane.
Find that any $U_{D_3, 1}$ in $D$ consists of opposite 
triangles against $b(D)$ that share the $3$-fold axis of
$D$. 
Hence there is no such two $U_{D_3, 2}$'s
in between the bases of $U_{D_3, 1}$ and 
 we do not have case (iii).
We do not have case (iv) 
because there is no regular triangle that have some point of $D$ 
on its $3$-fold axis as shown in Figure~\ref{fig:tri-tcube}. 

\begin{figure}[t]
\centering 
\includegraphics[width=14cm]{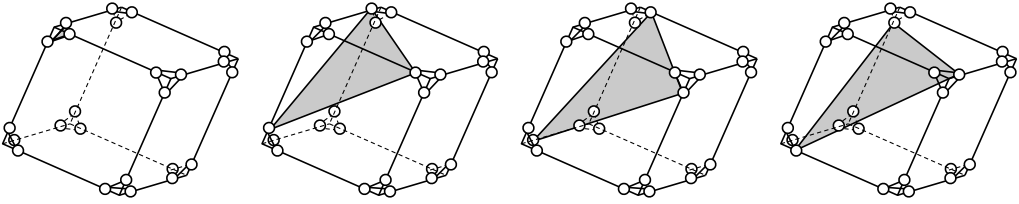}
\caption{Regular triangles in an $\epsilon$-truncated cube. }
\label{fig:tri-tcube}
\end{figure}

Finally, we consider case (ii). 
The only possibility to form two $U_{D_3, 1}$'s is to choose one 
base from the four triangles shown in Figure~\ref{fig:tri-tcube}.
 We divide the long edges of the $\epsilon$-truncated cube by a decomposition
 of the base polyhedron. We can cut a regular cube into two
 triangular pyramids and one triangular anti-prism. 
 The first group consists of the edges of the $\epsilon$-truncated cube
 contained 
 in the two triangular prisms and the second group consists of the
 edges of the $\epsilon$-truncated cube
 contained in the triangular anti-prism. 
 (See Figure~\ref{fig:cut-cube-into3}.) 
 Remember that the two endpoints of
 each long edge of the $\epsilon$-truncated cube 
 are the destinations of one robot and $P'$ contains just one of them. 
 From the endpoints of the first group 
 we can construct at most one triangular anti-prism 
 (i.e., $U_{D_3, 1}$). 
 Then we check the endpoints of the second group and we can form two 
 triangular anti-prisms each of which contains both endpoints of
 long edges of $D$ since the arrangement of $D_3$ for the first group
 and that of the second group should be common.
 (See Figure~\ref{fig:anti-tcube}.) 
 Thus $P'$ cannot contain two $U_{D_3,1}$'s and 
 $\gamma(P') \neq D_3$.

\begin{figure}[t]
\centering 
\subfigure[First group]
 {\includegraphics[width=3.5cm]{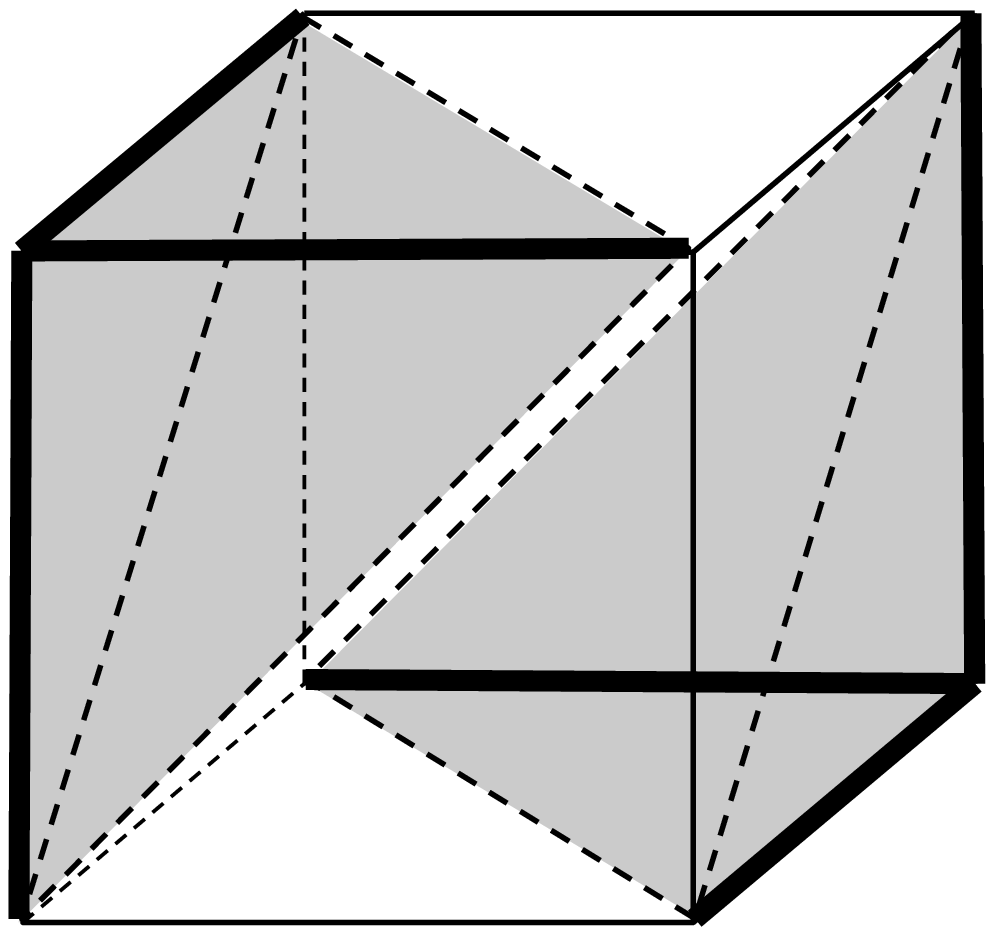}\label{fig:cut-cube-into3-1}}
\hspace{3mm}
\subfigure[Second group]
 {\includegraphics[width=3.5cm]{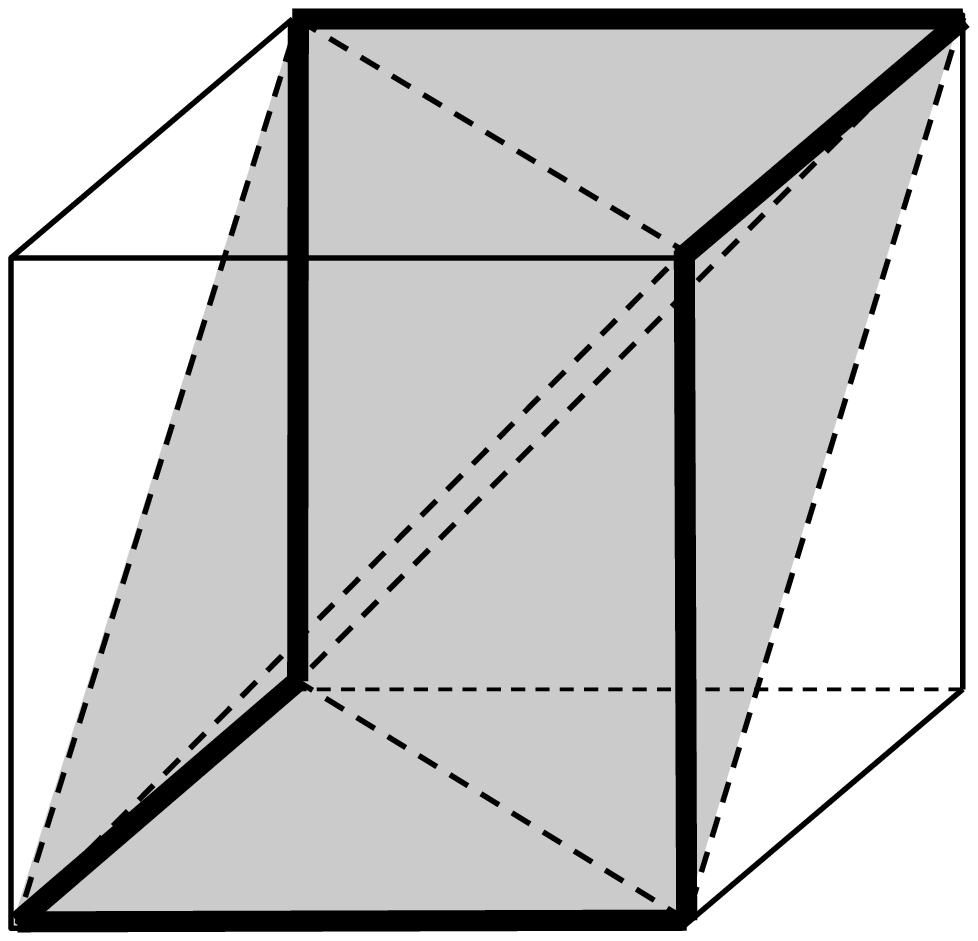}\label{fig:cut-cube-into3-2}}
\hspace{3mm}
 \caption{Decomposition of the edges of an $\epsilon$-truncated cube
 around a  $3$-fold axis 
 by using the base polygon. The bold edges show the member edges
 of each group.}
\label{fig:cut-cube-into3}
\end{figure}

\begin{figure}[t]
\centering 
\includegraphics[width=7cm]{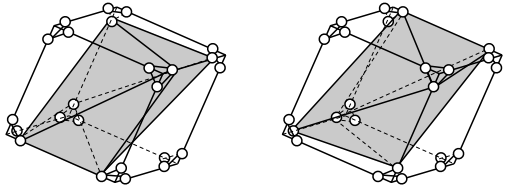}
\caption{$U_{D_3, 1}$ in the second group of an $\epsilon$-truncated cube. }
\label{fig:anti-tcube}
\end{figure}

We check the 3D rotation groups. 
First, assume $\gamma(P') = O$, then the $O$-decomposition of $P'$ 
 consists of two $U_{O, 4}$'s (cardinality $6$) or $U_{O, 2}$
 (cardinality $12$) 
because $|P'| = 12$. 
Because the vertices of an 
$\epsilon$-truncated cube is around the vertices of a regular cube, 
we can find neither any regular octahedron (i.e., its dual) in $D$ nor
any cuboctahedron in $D$. 

Next, assume $\gamma(P') = I$, then the $I$-decomposition of $P'$ 
consists forms $U_{I, 5}$ (cardinality $12$) 
because $|P'|=12$. 
Because the vertices of an 
$\epsilon$-truncated cube is around the vertices of a cube, 
we cannot find any regular icosahedron in $D$.

Consequently, $\gamma(P') \preceq C_4, C_3$, or $T$.

\noindent{\bf Case E. $P$ is a regular icosahedron:~}  
$D$ forms an $\epsilon$-expanded dodecahedron 
and we check the rotation group of any $12$-set of $D$. 
We will show $\gamma(P') \in \varrho(P) = \{T, D_3\}$. 

Assume that $\gamma(P')$ is $D_k$ or $C_k$ for some $k \geq 2$. 
Because we cannot find any regular $\ell$-gon for $\ell \geq 6$ in $D$, 
$k \leq 5$. 

 Assume that $\gamma(P')$ is $D_4$ or $C_4$,
 then $P'$ contains at least one square. 
We would like to check all possible squares in $D$, but 
clearly any edge of an $\epsilon$-expanded dodecahedron 
does not form any square because it does not have any 
adjacent edge with the same length and perpendicular to it. 
The remaining possibilities are the lines connecting the 
 vertices of the $\epsilon$-expanded dodecahedron.
 See Figure~\ref{fig:cube-dodeca} that shows an embedding of a cube into
 a regular dodecahedron (the base polyhedron).\footnote{
There are five such embeddings, but it is sufficient to check one of 
them because we just use it to check the arrangement of vertices. } 
Because the points of $D$ are obtained by expanding the faces of
a regular dodecahedron, there is no $4$-set of $D$ that forms a
square. 
Hence $\gamma(P') \neq C_4, D_4$. 
\begin{figure}[t]
\centering 
\includegraphics[width=3cm]{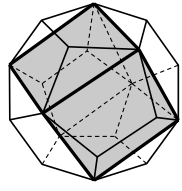}
\caption{A cube embedded in a regular dodecahedron.}
\label{fig:cube-dodeca}
\end{figure}

Assume that $\gamma(P')$ is $D_5$ or $C_5$. 
 Because $|P'|$ is not divided by $5$,
 $P'$ contains at least one $U_{D_5, 5}$ (cardinality $2$) or
 $U_{C_5, 5}$ (cardinality $1$). 
Figure~\ref{fig:penta-edodeca} shows all possible regular pentagons 
in $D$ and for each of the pentagons 
there is no vertex on its $5$-fold axis. 
Hence $\gamma(P') \neq C_5, D_5$. 

\begin{figure}[t]
\centering 
\includegraphics[width=14cm]{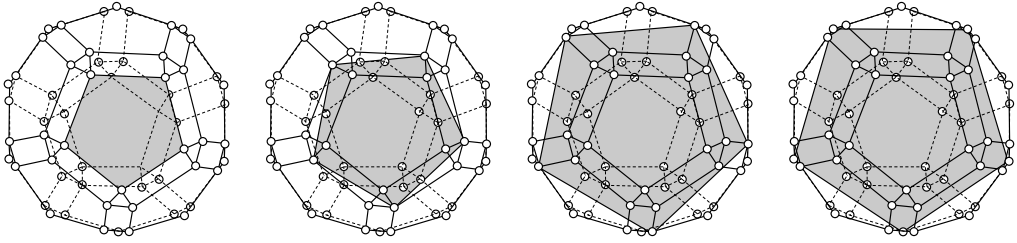}
\caption{Regular pentagons in an $\epsilon$-expanded dodecahedron}
\label{fig:penta-edodeca}
\end{figure}

We check the 3D rotation groups. 
First, assume $\gamma(P') = O$, then the $O$-decomposition of $P'$ 
consists of one $U_{O, 2}$ (cardinality $12$) or two $U_{O, 4}$'s 
(cardinality $6$) because $|P'|=12$. 
Because the vertices of an $\epsilon$-expanded dodecahedron is around 
the vertices of a regular dodecahedron, 
we can find neither a cuboctahedron nor regular octahedron in $D$. 

Next, assume $\gamma(P') = I$, then the $I$-decomposition of $P'$ 
consists of $U_{I, 5}$ (cardinality $12$) because $|P'|=12$. 
Because the vertices of an $\epsilon$-expanded dodecahedron is around 
the vertices of a regular dodecahedron, 
we cannot find any regular icosahedron (i.e., its dual) in $D$. 

Consequently, $\gamma(P') \preceq D_3$ or $T$. 

\noindent{\bf Case F. $P$ is a regular dodecahedron:~} 
$D$ forms an $\epsilon$-expanded icosahedron 
and we check the rotation group of any $20$-set of $D$. 
We will show $\gamma(P') \in \varrho(P) = \{D_5, D_2\}$. 

Assume that $\gamma(P')$ is $D_k$ or $C_k$ for some $k \geq 3$. 
Because we cannot find any regular $\ell$-gon in $D$ for $\ell=4$ and 
$\ell \geq 6$, $k=3$ or $5$. 

Assume that $\gamma(P')$ is $D_3$ or $C_3$. 
Because $|P'|$ is not divided by $3$, 
$P'$ contains at least one $U_{D_3, 3}$ (cardinality $2$) or
$U_{C_3, 1}$ (cardinality $1$). 
Figure~\ref{fig:tri-eicosa} shows all possible regular triangles 
in $D$ and for each of the triangles 
there is no vertex on its $3$-fold axis. 
Hence $\gamma(P') \neq C_3, D_3$. 

\begin{figure}[t]
\centering 
\includegraphics[width=10.5cm]{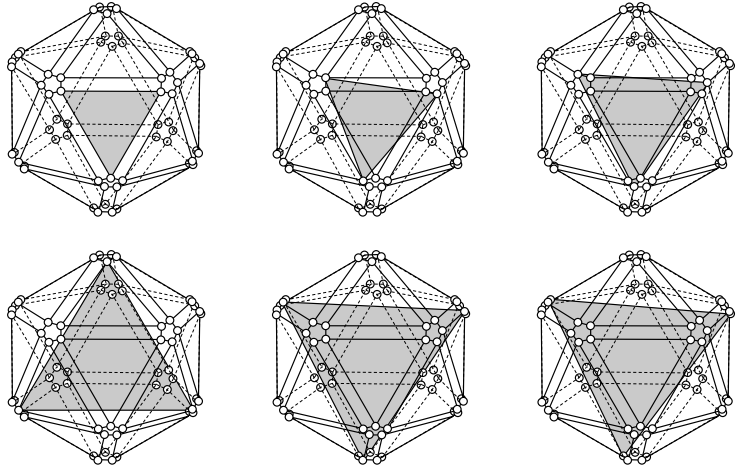}
\caption{Regular triangles in an $\epsilon$-expanded icosahedron.}
\label{fig:tri-eicosa}
\end{figure}

We check the 3D rotation groups. 
First, assume $\gamma(P') = T$, then the $T$-decomposition of $P'$ 
contains at least one $U_{T, 3}$ (cardinality $4$) or 
$U_{T, 2}$ (cardinality $6$) because $|P'|=20$ is not divided by $12$. 
Because the vertices of an $\epsilon$-expanded icosahedron is 
around the vertices of a regular icosahedron, 
we can find neither a regular tetrahedron nor a regular octahedron in $D$. 
Hence $\gamma(P') \neq T$. 

Second, assume $\gamma(P') = O$, then the $O$-decomposition of $P'$ 
consists of two $U_{O, 4}$'s (cardinality $6$) 
and one $U_{O, 3}$ (cardinality $8$) because $|P'|=20$. 
Because the vertices of an $\epsilon$-expanded icosahedron is 
around the vertices of a regular icosahedron, 
we can find neither a cube nor a regular octahedron in $D$. 
Hence $\gamma(P') \neq O$. 

Finally, assume $\gamma(P') = I$, then the $I$-decomposition of $P'$ 
forms a $U_{I, 3}$ (cardinality $20$) because $|P'|=20$. 
Because the vertices of an $\epsilon$-expanded icosahedron is 
around the vertices of a regular icosahedron, 
we cannot find a regular dodecahedron (i.e., its dual) in $D$. 
Hence $\gamma(P') \neq I$. 

Consequently, $\gamma(P') \preceq D_2$ or $D_5$.

\noindent{\bf Case G. $P$ is a icosidodecahedron:~} 
$D$ forms an $\epsilon$-truncated icosahedron 
and we check the rotation group of any $30$-set of $D$. 
We will show $\gamma(P') \in \varrho(P) = \{C_5, C_3\}$. 

Assume that $\gamma(P')$ is $D_k$ or $C_k$ for some $k \geq 2$. 
Because we cannot find any regular $\ell$-gon in $D$ 
for $\ell=4, 6, 7, \ldots$, $k=2, 3$, or $5$. 

Assume $\gamma(P') = D_2$, 
then the $D_2$-decomposition of $P'$ contains at least one 
$U_{D_2, 2}$ (cardinality $2$) because $|P'| = 30$ is not divided by $4$. 
However, $D_2$-decomposition of $P'$ does not consist of only 
$U_{D_2, 2}$'s, otherwise $P'$ consists of $15$ $U_{D_2, 2}$'s 
and there is at least one $2$-fold axes of $D_2$ 
that contains more than two $U_{D_2, 2}$'s (i.e., $4$ points). 
However, there is no line containing more than two points of $D$ 
since $D$ is spherical. 
Hence $P'$ contains at least one $U_{D_2, 1}$, i.e., 
a sphenoid, a regular tetrahedron, a rectangle, or a square. 
Now consider $\epsilon \to 0$. If $U_{D_2, 1}$ from a
sphenoid or a regular tetrahedron, 
when $\epsilon = 0$, $U_{D_2, 1}$ forms a regular tetrahedron or
a sphenoid in the base polyhedron (a regular icosahedron).
There are two types of possibilities for the edges of these
regular tetrahedron or a sphenoid in the base polyhedron:
the edges of the regular icosahedron or the edges in the interior of the 
regular icosahedron connecting two vertices. 
For any edge of the regular icosahedron, we cannot find any vertex 
on the $2$-fold axis of it.
For the edges in the interior of the regular icosahedron,
there are two possibilities as shown in Figure~\ref{fig:edges-icosa},
but we cannot find any vertex on the $2$-fold axis of it. 
When $\epsilon > 0$, the arrangement of edges of $U_{D_2, 1}$
slightly moves from the regular tetrahedron or the sphenoid,
but we cannot find two points forming $U_{D_2, 2}$ on the
$2$-fold axis of $U_{D_2, 1}$.

In the same way, if $U_{D_2,1}$ forms a rectangle or a square,
when $\epsilon=0$, $U_{D_2,1}$ forms a line formed by two vertices of
the base polyhedron or a rectangle formed by four vertices of the
base polyhedron. 
The possible edges are same as the previous case 
and in the same way, we cannot find any $U_{D_2, 2}$. 
Hence $\gamma(P') \neq D_2, C_2$. 
\begin{figure}[t]
\centering 
\subfigure[]
 {\includegraphics[width=3cm]{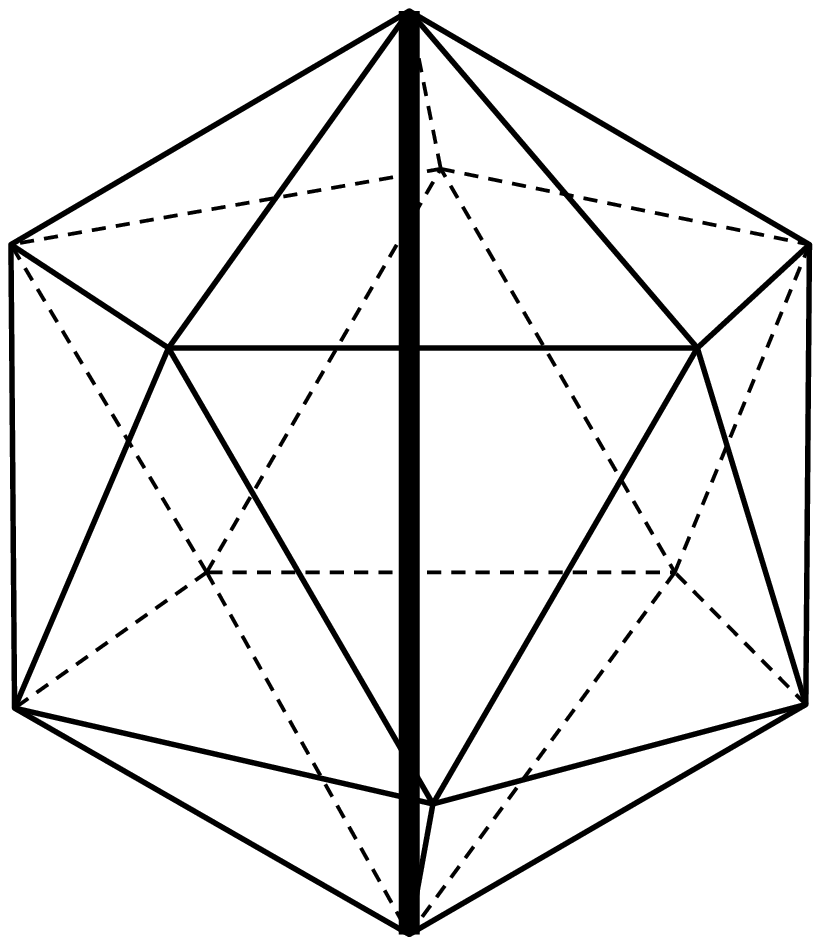}\label{fig:edges-icosa-1}}
\hspace{3mm}
\subfigure[]
 {\includegraphics[width=3cm]{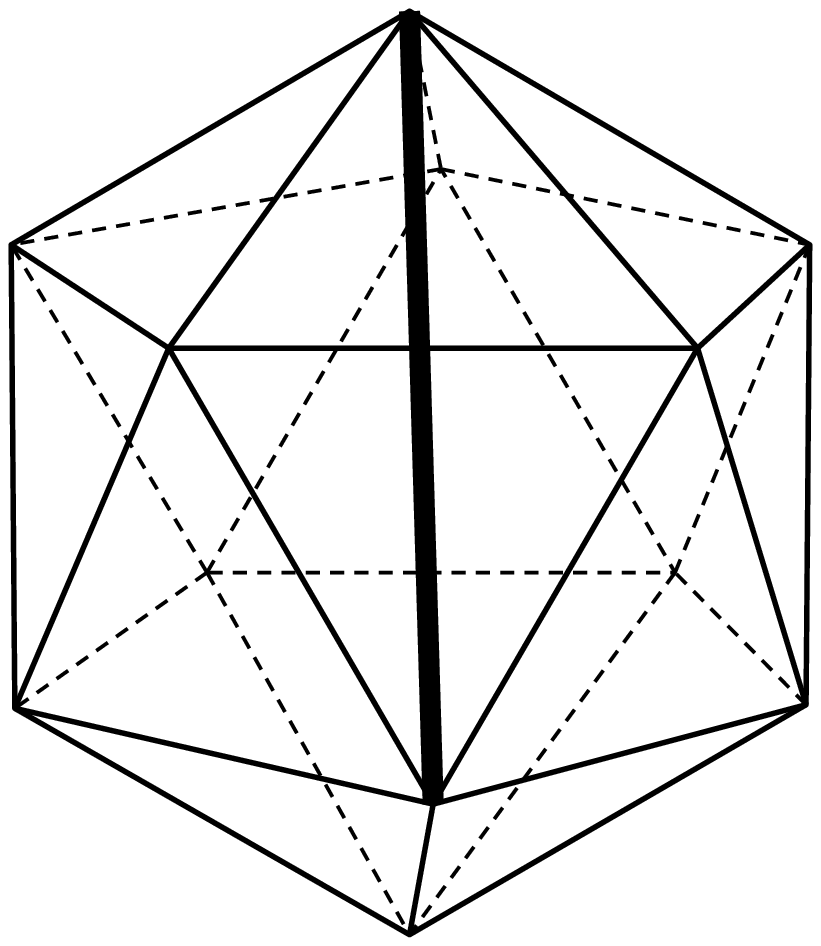}\label{fig:edges-icosa-2}}
\hspace{3mm}
\caption{Long edges in the interior of a regular icosahedron.} 
\label{fig:edges-icosa}
\end{figure}

Assume $\gamma(P') = D_5$, then 
$P'$ contains at least one regular pentagon. 
Figure~\ref{fig:penta-ticosa} shows all possible regular pentagons
in $D$ and any regular pentagon in $D$ is centered at a $5$-fold axis
of $\gamma(D) = I$. 
Because there is no point of $D$ on the $5$-fold axis of any
regular pentagon, the $D_5$-decomposition of $P'$
does not contain $U_{D_5, 5}$.  

\begin{figure}[t]
\centering 
\includegraphics[width=14cm]{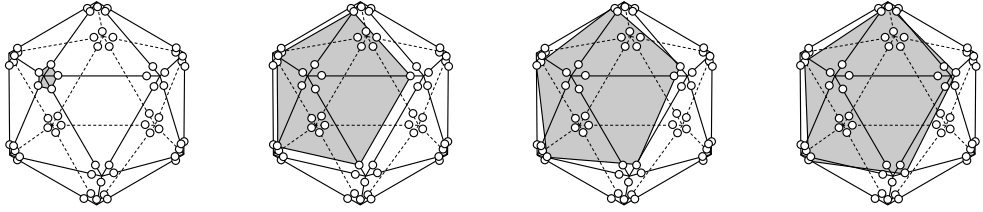}
\caption{Regular pentagons in an $\epsilon$-truncated icosahedron.}
\label{fig:penta-ticosa}
\end{figure}

 Assume that the $D_5$-decomposition of $P'$ consists of 
 three $U_{D_5, 1}$'s (cardinality $10$). 
 Thus the three $U_{D_5, 1}$'s share a $5$-rotation axis of $\gamma(D)$. 
 We divide the long edges of the $\epsilon$-truncated icosahedron into
 three groups base on a decomposition of the base polyhedron into 
 two regular pentagonal pyramids and one pentagonal anti-prism.  
The first group consists of the side edges of the two pyramids, 
the second group consists of the perimeter of the bases of
the pyramids (also the anti-prism), and the 
third group consists of the side edges of the anti-prism as shown in 
Figure~\ref{fig:cut-icosa-into5}.

\begin{figure}[t]
\centering 
\includegraphics[width=9cm]{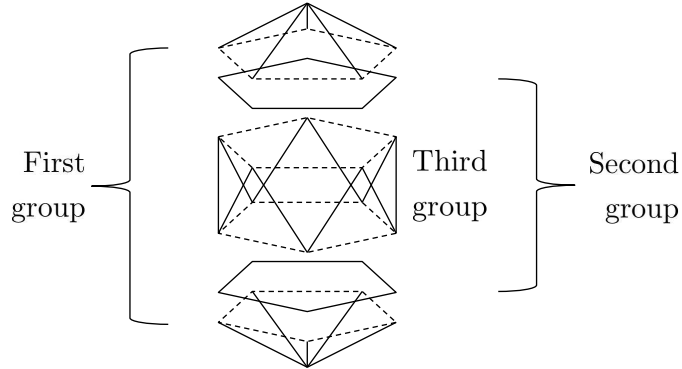}
\caption{Decomposition of the edges of an $\epsilon$-truncated
 icosahedron around a $5$-fold axis 
 by using the base polygon. The bold edges show the member edges
 of each group.}
\label{fig:cut-icosa-into5}
\end{figure}

 Remember that the two endpoints of
 each long edge of the $\epsilon$-truncated icosahedron
 are the destinations of one robot and $P'$ contains just one of them. 
 From the endpoints of the first group, 
 we can construct at most one pentagonal anti-prism
 (i.e., $U_{D_5, 1}$) and
 from the second group, 
 we can construct at most one pentagonal anti-prism. 
 Then we check the endpoints of the third group and we can form two 
 pentagonal anti-prisms each of which contains both endpoints of 
 long edges of $D$ because the anti-prisms should
 share an arrangement of $D_5$
 with the pentagonal anti-prisms formed by the first and the second groups.
 Thus $P'$ cannot contain three $U_{D_5,1}$'s.

 Assume that the $D_5$-decomposition of $P'$ contains $U_{D_5, 2}$
 (cardinality $5$). 
 Because $|P'|=30$, the number of $U_{D_5, 2}$'s is even and from 
 Figure~\ref{fig:penta-ticosa}, at most two regular pentagons from $D$
 are on the same plane. Thus we have two $U_{D_5, 2}$'s and
 two $U_{D_5, 1}$'s.
 The bases of $U_{D_5, 1}$'s are selected from the pentagons in
 Figure~\ref{fig:penta-ticosa} and opposite against the plane containing
 the two $U_{D_5, 2}$'s, but we cannot find any two $U_{D_5, 2}$'s
 on such a plane. 
 Thus we do not have this case. 
 Hence $\gamma(P') \neq D_5$. 
  
Assume $\gamma(P')=D_3$, then 
$P'$ contains at least one regular triangle. 
On the other hand, it does not contain $U_{D_3, 3}$ (cardinality $2$). 
Figure~\ref{fig:tri-ticosa} shows all regular triangles in $D$ and 
any regular triangle in $D$ is centered at a $3$-rotation axis
of $\gamma(D) = I$. 
Hence there is no point on the $3$-fold axis of any
regular pentagon in $D$ and the $D_3$-decomposition of $P'$
does not contain $U_{D_3, 3}$.  

\begin{figure}[t]
\centering 
\includegraphics[width=10.5cm]{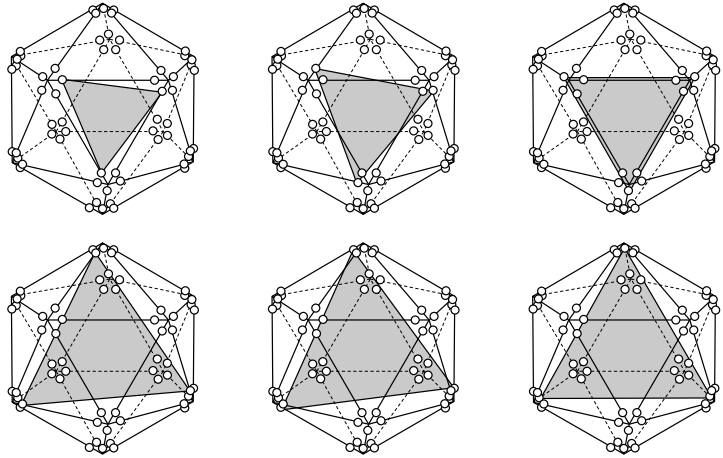}
\caption{Regular triangles in an $\epsilon$-truncated icosahedron.}
\label{fig:tri-ticosa}
\end{figure}

 Assume that the $D_3$-decomposition of $P'$ consists of five
 $U_{D_3, 1}$'s. 
 Hence the five $U_{D_3, 1}$'s share a $3$-rotation axis of $\gamma(D)$. 
 We divide the long edges of the $\epsilon$-truncated icosahedron based
 on a decomposition of the edges of a regular icosahedron 
 into four groups as shown in Figure~\ref{fig:cut-icosa-into3}.
 
\begin{figure}[t]
\centering 
\subfigure[First group]
{\includegraphics[width=2.5cm]{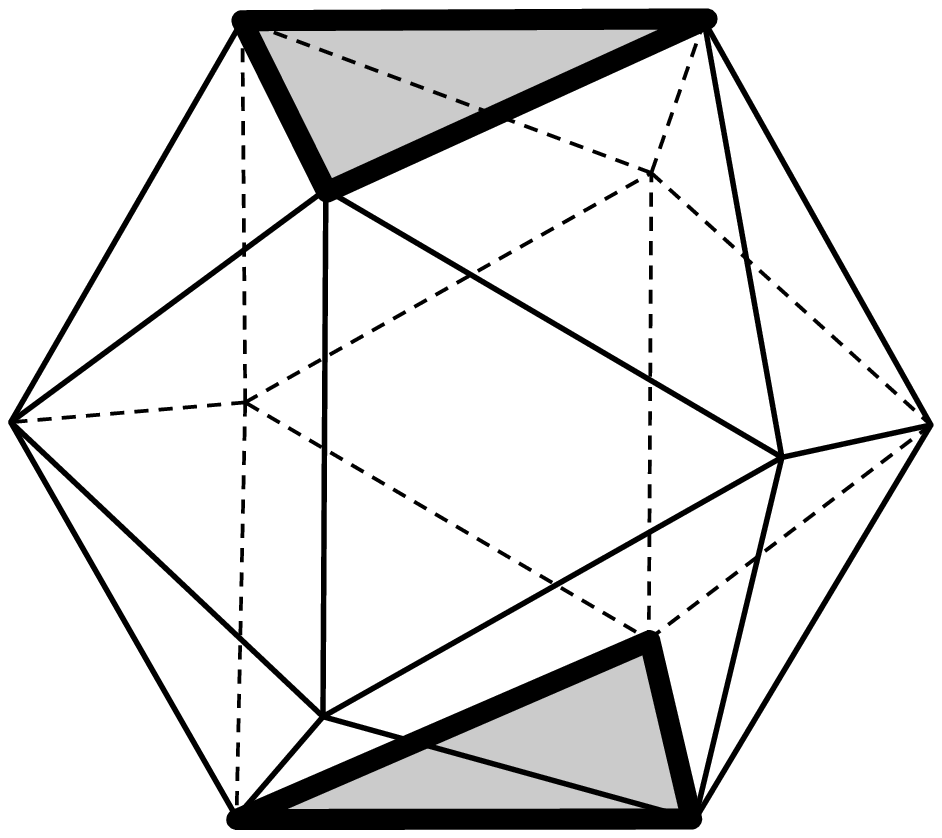}\label{fig:cut-icosa-into3-1}}
\hspace{3mm}
\subfigure[Second group]
{\includegraphics[width=2.5cm]{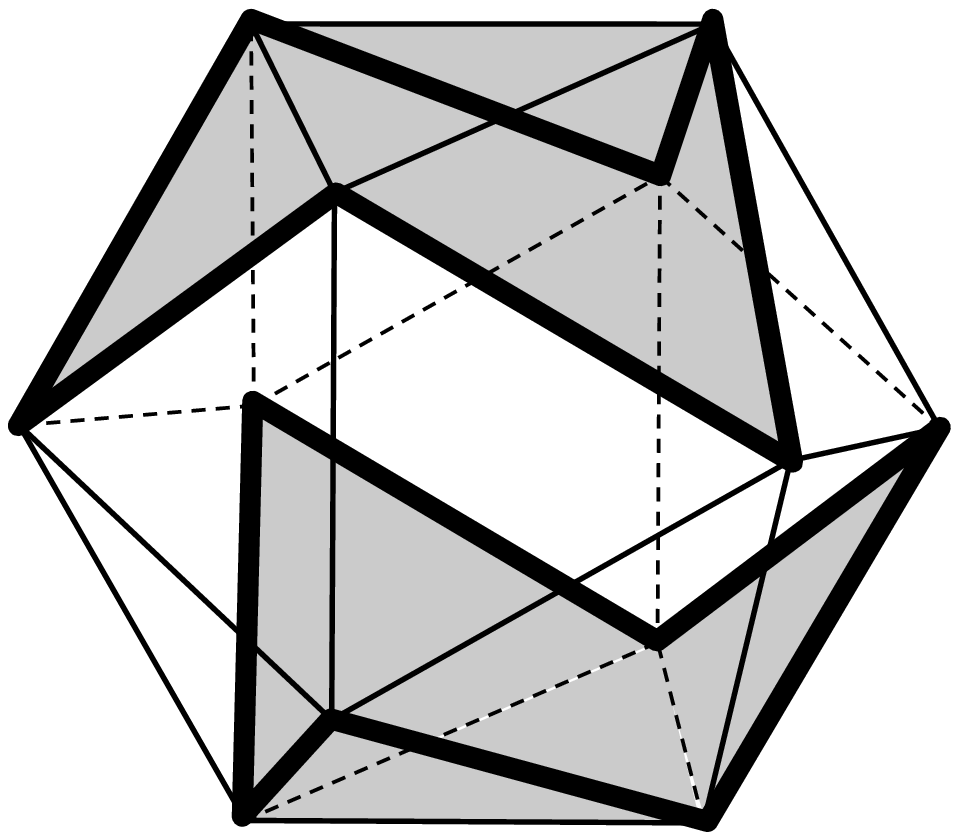}\label{fig:cut-icosa-into3-2}}
\hspace{3mm}
\subfigure[Third group]
{\includegraphics[width=2.5cm]{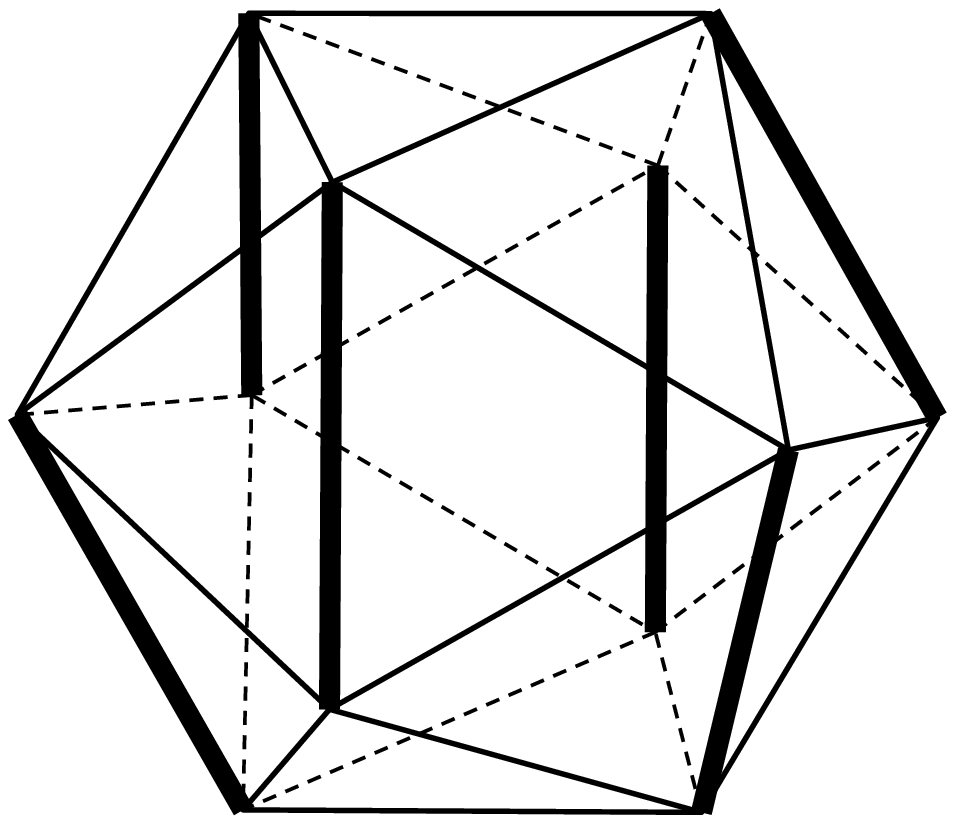}\label{fig:cut-icosa-into3-3}}
\hspace{3mm}
\subfigure[Fourth group]
{\includegraphics[width=2.5cm]{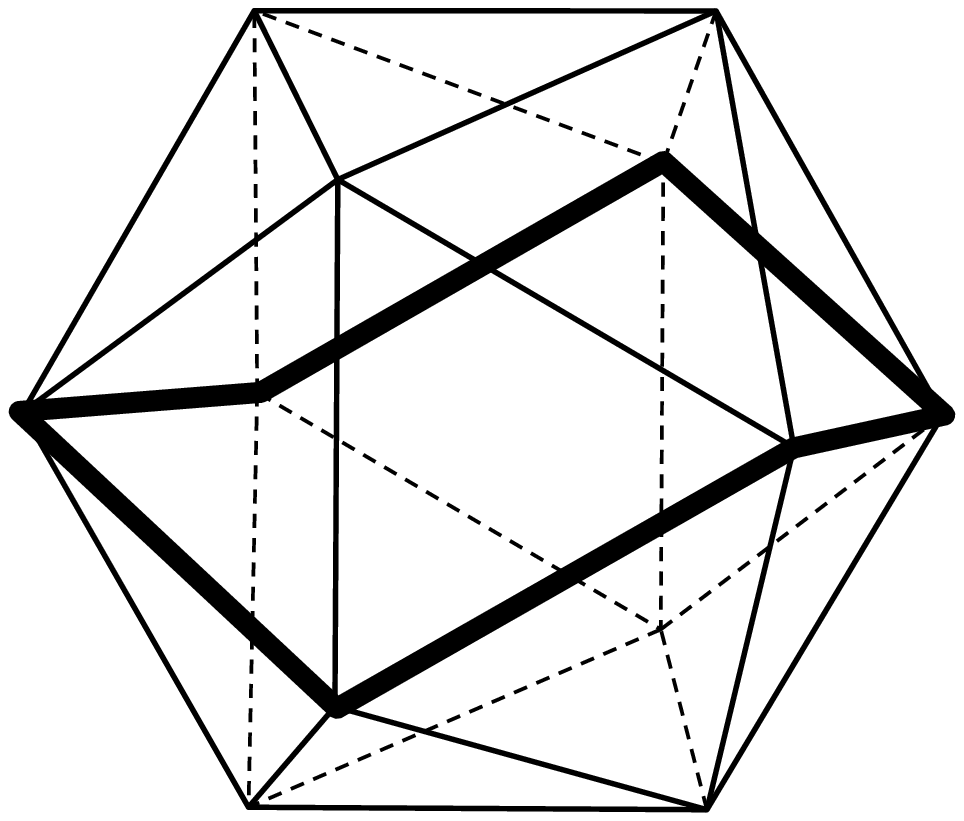}\label{fig:cut-icosa-into3-4}}
\hspace{3mm}
\caption{Decomposition of the edges of an $\epsilon$-truncated
 icosahedron around a $3$-fold axis 
 by using the base polygon. The bold edges show the member edges
 of each group.}
\label{fig:cut-icosa-into3}
\end{figure}

 Remember that the two endpoints of
 each long edge of the $\epsilon$-truncated icosahedron
 are the destinations of one robot and $P'$ contains just one of them. 
 From the endpoints of the first group,  
 we can construct at most one triangular anti-prism 
 (i.e., $U_{D_3, 1}$). 
 In the same way, we can construct at most two triangular anti-prisms
 and at most one triangular anti-prism from the second and the third group. 
 Then we check the endpoints of the fourth group and we can form
 two triangular anti-prisms each of which contains both endpoints of
 long edges of $D$ because they should share the arrangement of $D_3$
 with $U_{D_3, 1}$'s formed by the other groups.
 Thus $P'$ does not contain five $U_{D_3,1}$'s.

 Assume that the $D_3$-decomposition of $P'$ contains $U_{D_3, 2}$,
 then because $|P'|=30$, the number of $U_{D_3, 2}$ in
 the $D_3$-decomposition of $P'$ is even and they are on the same
 plane. From Figure~\ref{fig:tri-ticosa}, at most two triangles from $D$
 are on the same plane. Thus we have two $U_{D_3, 2}$'s and
 four $U_{D_3, 1}$'s.
 The bases of $U_{D_3, 1}$'s are selected from the regular triangles in
 Figure~\ref{fig:tri-ticosa} and opposite against the plane containing
 the two $U_{D_3, 2}$'s, but we cannot find any two $U_{D_3, 2}$'s
 on such a plane. 
 Thus we do not have this case. 
 Hence $\gamma(P') \neq D_3$. 

We check the 3D rotation groups. 
First, assume $\gamma(P') = T$, then the $T$-decomposition of $P'$ 
contains at least one $U_{T, 2}$ (cardinality $6$) 
because $|P'|=30$ is not divided by $4$. 
Because the vertices of an $\epsilon$-truncated icosahedron 
is around the vertices of a regular icosahedron, we cannot find any 
regular octahedron in $D$. 
Hence, $\gamma(P) \neq T$. 

Second, assume $\gamma(P') = O$, then the $O$-decomposition of $P'$ 
contains at least one $U_{O, 4}$ (cardinality $6$) 
 since $|P'|=30$ is not divided by $4$.
 In the same way as the previous case, we cannot find any 
regular octahedron in $D$. 
Hence $\gamma(P) \neq O$. 

Finally, assume $\gamma(P') = I$, then the $I$-decomposition of $P'$ 
forms $U_{I, 2}$ (cardinality $30$) since $|P'|=30$. 
Because the vertices of an $\epsilon$-truncated icosahedron 
is around the vertices of a regular icosahedron, we cannot find any 
icosidodecahedron in $D$. 
Hence $\gamma(P) \neq I$. 

Consequently, $\gamma(P') \preceq C_3$ or $C_5$. 
\qed 
\end{proof}

\subsection{Composition of transitive sets of points} 
\label{subsec:composition}

In this section, we show algorithm $\psi_{SYM}$ that
translates an initial configuration $P$ to another configuration $P'$
that satisfies $\gamma(P') \in \varrho(P)$.
When $P$ is transitive regarding a 3D rotation group,
$\psi_{SYM}$ makes the
robots execute the ``go-to-center'' algorithm
(Algorithm~\ref{alg:go-to-center}) and its correctness
is already shown in Section~\ref{subsec:base}.
In this section, we consider other initial configurations
where $\gamma(P) \not\in \varrho(P)$. 

For example, consider the case where $P$ consists of a 
cube and a regular octahedron with $\gamma(P) = O$ 
(Figure~\ref{fig:composi-1}). 
In this case, unoccupied rotation axes are the six 
$2$-fold axes, and $\varrho(P) = \{C_2\}$, 
because there are no three $2$-fold axes 
perpendicular to each other. 
As we have already shown, each of these regular polyhedra 
can show its symmetricity by executing Algorithm~\ref{alg:go-to-center} 
and the robots on the $3$-fold axes and $4$-fold 
axes eliminate these rotation axes. 
However there are executions that do not keep the 
the smallest enclosing ball of the robots (Figure~\ref{fig:composi-2}) 
and we cannot directly discuss the composite symmetricity.
Rather, algorithm $\psi_{SYM}$ makes each element of the
$\gamma(P)$-decomposition 
show its symmetricity one by one with keeping the 
smallest enclosing ball unchanged. 
For example, starting from an initial configuration shown in 
Figure~\ref{fig:composi-1}, 
the proposed algorithm first makes 
the robots forming a regular octahedron execute the go-to-center algorithm 
with the robots forming a cube keeping the smallest enclosing circle. 
Let $P'$ be the configuration after the regular octahedron is broken.
Because the robots formed the regular octahedron approaches to the
center, the robots forming the cube do not move during the transition from $P$
to $P'$ and they also keep the rotation axes of $P$, 
i.e., any rotation applicable to $P'$ should also applicable 
to the cube. 
Hence, we can check $\gamma(P')$ by checking the rotation axes of $O$. 
For example, because a robot on a $4$-fold axes in $P$ 
(i.e., a robot forming a regular octahedron) left the rotation axes 
and there are no three corresponding robots to keep the rotation axes, 
$\gamma(P')$ does not have any $4$-rotation axes. 
Hence, $\gamma(P') \preceq D_3$. 
When $\gamma(P')=D_3$, 
$P'$ is a configuration where the principal rotation axis is 
occupied by two robots  forming the cube. 
The robots can translate $P'$ to another configuration $P''$ 
with $\gamma(P'') \preceq C_2$ by the two robots leaving the principal
axis. 
Other possibility is $\gamma(P')=C_3$, the robots can recognize 
a single robot on the single rotation axis 
and the robots can translate $P'$ to $P''$ where $\gamma(P'') = C_1$ 
by this robot leaving the axis. 
In this way, the robots can translate $P$ to another configuration 
$P''$ that satisfies $\gamma(P'') \preceq C_2 \in \varrho(P)$. 

\begin{figure}[t]
\centering 
\subfigure[]
{\includegraphics[width=3cm]{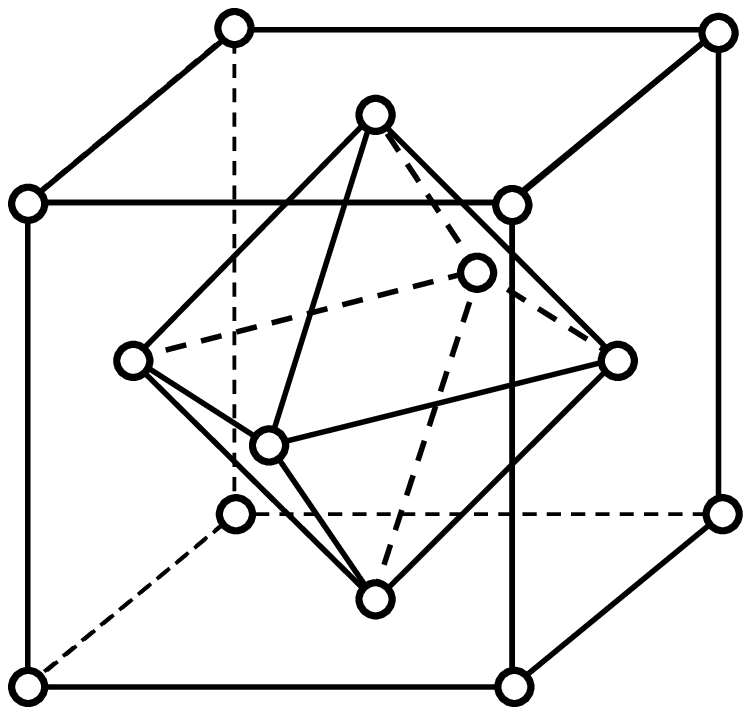}\label{fig:composi-1}}
\hspace{3mm}
\subfigure[]
{\includegraphics[width=3cm]{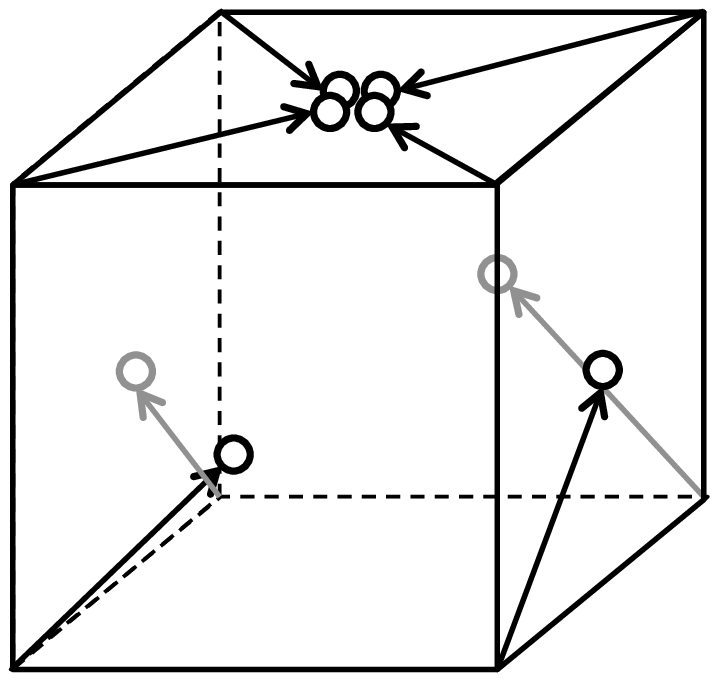}\label{fig:composi-2}}
\caption{Symmetry breaking from a composite initial configuration.}
\label{fig:composi}
\end{figure}

\subsubsection{Algorithm $\psi_{SYM}$}

The proposed algorithm $\psi_{SYM}$ is shown in
Algorithm~\ref{alg:psi-sym}.
Algorithm $\psi_{SYM}$ 
makes the robots on each type (i.e., fold) of the rotation axes 
leave the positions by repeating the following
procedure: 
$\psi_{SYM}$ first selects an element of the
$\gamma(P)$-decomposition of the current configuration $P$
that is on the rotation axes $\gamma(P)$ and make the element
shrink toward $b(P)$ so that the other robots keep the smallest
enclosing ball (Procedure {\texttt shrink} in
Algorithm~\ref{proc:psi-sym}).
This movement does not change the rotation group of the positions of
robots. 
Let $P'$ be a resulting configuration. 
Then, $\psi_{SYM}$ makes the innermost robots leave the rotation axes
of $\gamma(P')$. 
Depending on $\gamma(P')$, there are three procedures; 
{\texttt go-to-sphere} (when $\gamma(P')$ is cyclic or $\gamma(P')$ is
dihedral and its principal axis is occupied),
{\texttt go-to-corner} (when $\gamma(P')$ is dihedral and its secondary
axes are occupied,)
{\texttt go-to-center} (Algorithm~\ref{alg:go-to-center},
when $\gamma(P')$ is a 3D rotation group). 
These three procedures send the robots on $I(P')$ to some point
in the interior of $I(P')$ and not on any rotation axes of
$\gamma(P')$. 
By repeating these two phases, $\psi_{SYM}$ removes 
occupied rotation axes. 

Any terminal configuration $P$ of Algorithm~\ref{alg:psi-sym} satisfies 
one of the following two properties:
\begin{description}
\item[(i)] If $\gamma(P) \neq C_1$, then $P$ is a regular $n$-gon
or no robot is on the rotation axes of $\gamma(P)$. 
\item[(ii)] $\gamma(P)=C_1$. 
\end{description}
When an arbitrary configuration $P$ satisfies one of the above two
conditions, $\psi_{SYM}$ outputs $\emptyset$ at all robots.  
For any terminal configuration $P$, 
the $\gamma(P)$-decomposition of $P$ 
consists of elements of size $|\gamma(P)|$, 
which is shown to be useful 
in the pattern formation algorithm in Section~\ref{sec:suf}.  
Here, $C_1$-decomposition of $P$ divides $P$ into $n$ subsets. 

In $\psi_{SYM}$, 
we use the following notations. 
Let $P$ and $\{P_1, P_2, \ldots, P_m\}$ be an initial configuration and 
its $\gamma(P)$-decomposition. 
\begin{itemize}
\item $P_{ip}$: The element on the (principal) axis 
when $\gamma(P)$ is cyclic or dihedral. 
If there are multiple such elements, 
$\psi_{SYM}$ selects the element with the minimum index. 
\item $P_{is}$: The element on the secondary axis 
when $\gamma(P)$ is dihedral.  
If there are multiple such elements, 
$\psi_{SYM}$ selects the element with the minimum index. 
\item $P_{imax}$: The element on some occupied rotation axes 
with the maximum fold 
when $\gamma(P) \in \{T, O, I\}$. 
If there are multiple such elements, 
$\psi_{SYM}$ selects the element with the minimum index. 
\end{itemize}
For example, if $P$ consists of a regular octahedron and a cube 
(Figure~\ref{fig:composi-1}), 
all the $3$-fold axes and $4$-fold axes are occupied, and 
$P_{imax}$ is the element forming the regular octahedron.
The robots can agree on $P_{ip}$, $P_{is}$, and $P_{imax}$ if any 
irrespective of local coordinate systems. 

When $\gamma(P)$ is dihedral, $\psi_{SYM}$ uses the 
{\em reference prism} that inscribed in a ball $Ball(b(P), rad(I(P))/2)$. 
Consider a cylinder with radius is $rad(I(P)/4)$ that is 
parallel to the principal axes and inscribed in $I(P)$. 
The corners of the reference prism are the intersection of this 
cylinder and the plane formed by the principal axes and a 
secondary axes (Figure~\ref{fig:ref-prism}). 
Hence, when $\gamma(P) = D_k$, the bases of the reference prism 
are regular $2k$-gons. 

\begin{figure}[t]
\centering 
\includegraphics[width=2.5cm]{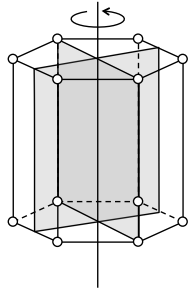}
\caption{Reference prism for $D_6$. There are $6$ planes formed by 
the principal axes and a $2$-fold axes, and 
the reference prism consists of regular $12$-gon faces. }
\label{fig:ref-prism}
\end{figure}

\begin{algorithm}
\caption{$\psi_{SYM}$ for robot $r_i \in R$}
\label{alg:psi-sym} 
\begin{tabbing}
xxx \= xxx \= xxx \= xxx \= xxx \= xxx \= xxx \= xxx \= xxx \= xxx
\kill 
{\bf Notation} \\ 
\> $P$: The positions of robots observed in $Z_i$.\\ 
\> $\{P_1, P_2, \ldots, P_m\}$: $\gamma(P)$-decomposition of $P$. \\ 
\> $p_i$: current position of $r_i$. \\ 
\\ 
{\bf Algorithm}  \\ 
\> {\bf If} $p_i = b(P)$ {\bf then} \\
\> \> Execute \texttt{go-to-sphere} \\ 
\> {\bf Else} \\ 
\> \> {\bf If} $\gamma(P) \neq C_k$ and 
$P \cap B(P) \neq P_m$ {\bf then} \\ 
\> \> \> {\bf If} $p_i \in P_m$ {\bf then} \\ 
\> \> \> \> Execute \texttt{Expand} \\ 
\> \> \> \> // To guarantee that there are at least two robots 
on the smallest enclosing ball of the robots. \\ 
\> \> \> {\bf Endif} \\ 
\> \> {\bf Else} \\ 
\> \> \> {\bf Switch} ($\gamma(P)$) {\bf do} \\ 

\> \> \> \> {\bf Case} $C_k$ ($k \geq 1$): \\ 
\> \> \> \> \> {\bf If} $k \neq 1$ and
 $P_{ip}$ is determined and $p_i \in P_{ip}$ {\bf then} \\ 
\> \> \> \> \> \> {\bf If} $P_{ip} \neq P_{1}$ {\bf then} \\ 
\> \> \> \> \> \> \> Execute \texttt{Shrink}($P \setminus P_{ip}$). \\ 
\> \> \> \> \> \> {\bf Else} \\  
\> \> \> \> \> \> \> Execute \texttt{go-to-sphere} \\ 
\> \> \> \> \> \> {\bf Endif} \\ 
\> \> \> \> \> {\bf Endif} \\ 

\> \> \> \> {\bf Case} $D_{\ell}$ ($\ell \geq 2$): \\ 
\> \> \> \> \> {\bf If} $P_{ip}$ is determined and $p_i \in P_{ip}$ {\bf then} \\ 
\> \> \> \> \> \> {\bf If} $P_{ip} \neq P_1$ {\bf then} \\ 
\> \> \> \> \> \> \> Execute \texttt{Shrink}($P$). \\ 
\> \> \> \> \> \> {\bf Else} \\  
\> \> \> \> \> \> \> Execute \texttt{go-to-corner} \\ 
\> \> \> \> \> \> {\bf Endif} \\ 
\> \> \> \> \> {\bf Else} \\ 
\> \> \> \> \> \> {\bf If} $P_{is}$ is determined, $P \neq P_{is}$ and 
$p_i \in P_{is}$ {\bf then} \\ 
\> \> \> \> \> \> \> {\bf If} $P_{is} \neq P_1$ {\bf then} \\ 
\> \> \> \> \> \> \> \> Execute \texttt{Shrink}$(P)$. \\ 
\> \> \> \> \> \> \> {\bf Else} \\  
\> \> \> \> \> \> \> \> Execute \texttt{go-to-corner} \\ 
\> \> \> \> \> \> \> {\bf Endif} \\ 
\> \> \> \> \> \> {\bf Endif} \\ 
\> \> \> \> \> {\bf Endif} \\ 

\> \> \> \> {\bf Default} // $\gamma(P) \in \{T, O, I\}$: \\
\> \> \> \> \> {\bf If} $P_{imax}$ is determined and $p_i \in P_{imax}$
 {\bf then} \\ 
\> \> \> \> \> \> {\bf If} $P_{imax} \neq P_1$ {\bf then} \\ 
\> \> \> \> \> \> \> Execute \texttt{Shrink}($P$). \\ 
\> \> \> \> \> \> {\bf Else} \\  
 \> \> \> \> \> \> \> Execute \texttt{go-to-center}($P_{imax}$)
// Algorithm~\ref{alg:go-to-center}. \\ 
\> \> \> \> \> \> {\bf Endif} \\ 
\> \> \> \> \> {\bf Endif} \\ 

\> \> \> {\bf Enddo} \\ 

\> \> {\bf Endif} \\ 
\> {\bf Endif}
\end{tabbing}
\end{algorithm}

\begin{algorithm}
\caption{Procedure for Algorithm~\ref{alg:psi-sym}}
\label{proc:psi-sym} 
\begin{tabbing}
xxx \= xxx \= xxx \= xxx \= xxx \= xxx \= xxx \= xxx \= xxx \= xxx
\kill 
\> \texttt{Expand} \\ 
\> \> Let $d_i$ be the intersection of $Ball(b(P), 2rad(I(P)))$ 
and the half line from $b(P)$ that passes $p_i$. \\ 
\> \> Go to $d_i$. \\ 
\\ 
\> \texttt{Shrink}($Q$) \\ 
\> \> Let $d_i$ be the intersection of $Ball(b(Q), I(Q)/2)$ 
and the line $\overline{p_i b(Q)}$. \\ 
\> \> Go to $d_i$. \\  
\\ 
\> \texttt{go-to-sphere} \\ 
\> \> Select arbitrary point on $Ball(b(P), rad(I(P)/2)$ with avoiding 
 the intersections \\
\> \> with rotation axes of $\gamma(P)$ and the 
equator (if $\gamma(P)$ is a 2D rotation group). \\ 
\> \> Move to the point. \\ 
\\ 
\> \texttt{go-to-corner} \\ 
\> \> Select a nearest vertex of the reference prism. \\ 
\> \> Go to the selected vertex. 
\end{tabbing}
\end{algorithm}

\subsubsection{Correctness of $\psi_{SYM}$}

We show the correctness of Algorithm~\ref{alg:psi-sym}. 
Let $P(0), P(1), P(2), \ldots $ be an execution of 
Algorithm~\ref{alg:psi-sym} from an initial configuration $P(0)$. 
If $P(0)$ is transitive, Algorithm~\ref{alg:psi-sym} executes
Algorithm~\ref{alg:go-to-center} and its correctness is already shown in
Section~\ref{subsec:base}. 
In the following, we assume that $P(0)$ is not transitive and
$\gamma(P(0))$-decomposition of $P(0)$ consists of at 
least two elements.
We first show that $\psi_{SYM}$ does
not allow the robots to form one transitive set of points and 
for any $t \geq 0$, 
$\gamma(P(t))$-decomposition of $P(t)$ consists of at least two elements
(Lemma~\ref{lemma:keep-vt}). 

For each transition from $P(t)$ to $P(t+1)$, 
if some robots move, 
they execute one of the five procedures \texttt{Expand}, 
\texttt{Shrink}, \texttt{go-to-sphere}, \texttt{go-to-corner} 
(shown in Algorithm~\ref{proc:psi-sym}) and 
\texttt{go-to-center} (Algorithm~\ref{alg:go-to-center}).  
The moving robots execute the same procedure,
because they can agree on $\gamma(P(t))$ and 
the $\gamma(P(t))$-decomposition of $P(t)$ that determines 
the procedure they execute. 

We will show that the rotation axes 
occupied by robots in $P(0)$ are gradually eliminated  
without adding any new rotation axis. 
Let $\gamma^+(P(t))$ ($\gamma^-(P(t))$, respectively) be the 
arrangement of occupied rotation axes (unoccupied rotation axes) 
of $\gamma(P(t))$. 
We also consider each of them as a set of rotation axes. 
In the following, we focus not on the fold of a rotation axes, 
but on the existence of the rotation axes. 
Even when the fold of a rotation axis becomes smaller, 
during a transition from $P(t)$ to $P(t+1)$,
we say the rotation axis remains in $\gamma(P(t+1))$. 

We will show the following two properties. 
\begin{enumerate}
\item When robots execute \texttt{go-to-sphere}, 
\texttt{go-to-corner}, or \texttt{go-to-center}, 
$\gamma^+(P(t+1))$ is a proper subset of $\gamma^+(P(t))$. 
\item $\gamma^-(P(t+1))$ is a subset of $\gamma^-(P(t))$.
\end{enumerate}
Starting from an initial configuration $P(0)$, the first property 
guarantees that all occupied rotation axes are gradually eliminated 
because $\psi_{SYM}$ executes these three 
procedures as long as there is an occupied rotation axis. 
The second property guarantees that the movement of robots does not 
add any new rotation axis. 
Hence the robot system eventually reaches a configuration $P(t^*)$ 
that have no occupied rotation axes and 
$\gamma(P(t^*)) \preceq \gamma^-(P(0))$. 
This property of $\gamma(P(t^*))$ satisfies the definition of 
the symmetricity and $\psi_{SYM}$ succeeds in making 
robots show their symmetricity. 

\begin{lemma}
\label{lemma:keep-vt}
Let $P(0), P(1), P(2), \ldots$ be an execution of 
Algorithm~\ref{alg:psi-sym} by oblivious FSYNC robots
from an initial configuration $P(0)$. 
If $P(0)$ is not transitive, then 
$P(t)$ is not transitive for all $t \geq 0$. 
\end{lemma}
\begin{proof}
 Let $\{P_1(t), P_2(t), \ldots, P_{m(t)}(t)\}$ be the
 $\gamma(P(t))$-decomposition of $P(t)$.
 Because $\psi_{SYM}$ outputs nothing at any robot when
 $\gamma(P(t))=C_1$, we assume $\gamma(P(t)) \neq C_1$.
 Hence, each $Ball(P_i(t))$ is a ball centered at $b(P(t))$. 
 We will show that $\psi_{SYM}$ does not move any
 $P_i(t)$ to $Ball(P_j(t))$ for any $i \neq j$.
 
 When $t=0$, from the assumption, $P(0)$ is not transitive and
 the $\gamma(P(0))$-decomposition of $P(0)$ consists of 
 at least two elements.
 Assume $\psi_{SYM}$ selects $P_i(0)$ to move the robots forming it.
 Then these robots synchronously execute one of the 
 five procedures, \texttt{Expand}, \texttt{Shrink},
 \texttt{go-to-sphere}, \texttt{go-to-corner}, and
 \texttt{go-to center}.
 When these robots execute \texttt{Expand},
 they move to the exterior of $B(P(0))$.
 In the same way, when they execute \texttt{Shrink},
 they move to the interior of $I(P(0))$.
 When they execute  \texttt{go-to-sphere}, \texttt{go-to-corner}, or 
 \texttt{go-to-center},
 they are on $I(P(0))$ and their destinations are also in the interior of
 $I(P(0))$.
 Hence during the transition from $P(0)$ to $P(1)$,
 $\psi_{SYM}$ does not move any
 $P_i(0)$ to $Ball(P_j(0))$ for any $i \neq j$.

 In the same way, for any transition from $P(t)$ to $P(t+1)$,
 $\psi_{SYM}$ does not move any
 $P_i(t)$ to $Ball(P_j(t))$ for any $i \neq j$,
 because the destinations of the five procedures are selected
 based on $I(P(t))$ and $B(P(t))$. 
 Hence we obtain the lemma. 
 \qed
\end{proof}

\begin{lemma}
Let $P(0), P(1), P(2), \ldots$ be an execution of 
Algorithm~\ref{alg:psi-sym} by oblivious FSYNC robots
from a non transitive 
initial configuration $P(0)$. 
If some robots execute \texttt{Shrink} or \texttt{Expand} during 
the transition from $P(t)$ to $P(t+1)$, then we have 
$\gamma^+(P(t+1)) = \gamma^+(P(t))$ and 
$\gamma^-(P(t+1)) = \gamma^-(P(t))$.  
\end{lemma}
\begin{proof}
As already mentioned, the robots that execute 
\texttt{Shrink} or \texttt{Expand} during the transition from 
$P(t)$ to $P(t+1)$ form an element of the $\gamma(P(t))$-decomposition 
of $P(t)$, and they agree on which procedure they execute. 
Let $R'$ and $P'$ be the set of robots 
that execute one of these two procedures 
during the transition from $P(t)$ to $P(t+1)$, 
and their positions in $P(t)$. 

 If \texttt{Shrink} is executed,
 because the movement of robots of $R'$ are radial against $b(P(t))$,
 it does not add any new rotation axis to $\gamma(P(t))$. 
 On the other hand, there is at least another 
 element of the $\gamma(P(t))$-decomposition of $P(t)$ in
 $P(t) \setminus P'$ and they keep the rotation axes of $P(t)$.
 When $\gamma(P(t))$ is cyclic,
 there is another element that forms $U_{\gamma(P(t)), 1}$
 (otherwise $\gamma(P(t))$ is not cyclic),
 and keeps the rotation axes of $\gamma(P(t))$.
 When $\gamma(P(t))$ is dihedral,
 there is another element that forms $U_{\gamma(P(t)), 1}$ or
 $U_{\gamma(P(t)), 2}$ (otherwise $\gamma(P(t))$ is not dihedral),
 and keeps the rotation axes of $\gamma(P(t))$.
 When $\gamma(P(t)) \in \{T, O, I\}$, there is another
 element that forms $U_{\gamma(P(t)), \mu}$ for $\mu < |\gamma(P(t))|$
 and keeps the rotation axes of $\gamma(P(t))$.
Hence this movement of $R'$ does not change the occupied 
and unoccupied rotation axis. 

If \texttt{Expand} is executed, $\gamma(P(t))$ is a dihedral 
group or a 3D rotation group. 
In the same way, the movement is radially against $b(P(t))$, 
and it does not change the rotation group and 
occupied and unoccupied rotation axes. 
\qed 
\end{proof}

\begin{lemma}
\label{lemma:eliminate}
Let $P(0), P(1), P(2), \ldots$ be an execution of 
 Algorithm~\ref{alg:psi-sym} by oblivious FSYNC robots
 from a non transitive 
initial configuration $P(0)$. 
If some robots execute \texttt{go-to-center}, \texttt{go-to-corner}, 
or \texttt{go-to-sphere} during the transition from $P(t)$ to $P(t+1)$, 
then we have 
$\gamma^+(P(t+1)) \subset \gamma^+(P(t))$ and 
$\gamma^-(P(t+1)) \subseteq \gamma^-(P(t))$.
\end{lemma}
\begin{proof}
As already mentioned, robots that move during the transition from 
$P(t)$ to $P(t+1)$ form the first element of the 
$\gamma(P(t))$-decomposition of $P(t)$, and 
execute the same procedure selected from the 
three procedures \texttt{go-to-center}, \texttt{go-to-corner}, 
or \texttt{go-to-sphere}. 
Let $R'$ be the set of these moving robots. 

From Lemma~\ref{lemma:keep-vt}, 
$P(t)$ is not transitive and 
there exits at least one another element of the $\gamma(P(t))$-decomposition 
of $P(t)$ that does not move during the transition. 
and keep $B(P(t))$.
Clearly, $\gamma(B(P(t)) \cap P(t)) = \gamma(P(t))$. 
No robot moves to the sphere of $B(P(t))$ during the transition 
because the robots of $R'$ moves to the interior of $I(P(t))$. 
Hence if we can apply some rotation to $P(t+1)$, 
then we can also apply it to $P(t+1) \cap B(P(t+1))$. 
In other words, $\gamma(P(t+1)) \preceq \gamma(P(t))$. 

We first show $\gamma^+(P(t+1)) \subset \gamma^+(P(t))$. 
Observe that the three procedures \texttt{go-to-center}, 
\texttt{go-to-corner}, and \texttt{go-to-sphere} assigns 
candidate destinations to each robot in $P(t)$ 
so that the set of candidate destinations 
of robots in $P(t)$ are disjoint. 
Let $a$ be the rotation axes that is occupied by a robot $r$
that executes one of these three procedures. 
To keep this axes $a$ in $P(t+1)$, at least one point that 
is symmetric for the destination of $r$ regarding axes $a$ 
should be occupied by another robot in $P(t+1)$. 
Such a point is a candidate destination of $r$ 
and no robot will occupy this point.  
Hence, $\gamma(P(t+1))$ does not have this rotation axes. 
 Hence, $\gamma^+(P(t+1)) \subset \gamma^+(P(t))$.
 
Additionally, the three procedures \texttt{go-to-corner}, 
\texttt{go-to-center}, \texttt{go-to-sphere} do not allow the 
robots to move a point on the rotation axis of $\gamma(P(t))$. 
Hence, the unoccupied rotation axes of $\gamma(P(t))$ remains 
unoccupied in $P(t)$. 
Thus, we have $\gamma^-(P(t+1)) \subseteq \gamma^-(P(t))$.

Consequently, $P(t+1)$ satisfies the two properties. 
\qed
\end{proof}

\begin{theorem}
Let $P$ be an arbitrary initial configuration of oblivious FSYNC robots.  
Algorithm~\ref{alg:psi-sym} translates $P$ into another 
configuration $P'$ that satisfies the terminal condition of $\psi_{SYM}$ 
and $\gamma(P) \in \varrho(P)$ in at most $7$ 
steps. 
\end{theorem}
\begin{proof}
Let $P(0)(=P), P(1), P(2), \ldots$ be the execution of 
Algorithm~\ref{alg:psi-sym} from $P$. 
In the worst case, one of the three procedures 
\texttt{go-to-center}, \texttt{go-to-corner}, \texttt{go-to-sphere} 
is executed in every two steps, 
with the exception that for the first and the second
cycle the robots execute \texttt{Expand} and \texttt{Shrink}. 
From Lemma~\ref{lemma:eliminate}, 
every time the three procedures are executed, 
at least one type of occupied rotation axes are eliminated without 
increasing the occupied rotation axes. 
Hence, the system eventually reaches a configuration $P(t^*)$ where 
$\gamma^+(P(t^*))$ is empty. 
Additionally, Lemma~\ref{lemma:eliminate} guarantees 
\begin{equation*}
 \gamma^-(P(t^*)) = \gamma(P(t^*))
  \subseteq \gamma(P(t^*-1)) \subseteq \cdots 
\subseteq \gamma(P(0)). 
\end{equation*}
Hence, $P(t^*)$ satisfies the claim. 

In the worst case, $t^* =7$ because in each iteration,
just one type of rotation axes is eliminated.
\qed 
\end{proof}


\section{Necessity of Theorem~\ref{theorem:main}}
\label{sec:nec}

In this section, we prove the necessity of Theorem~\ref{theorem:main}. 

\begin{theorem}
\label{theorem:nec}
Regardless of obliviousness, 
FSYNC robots can form a target pattern $F$ from 
an initial configuration $P$ 
only if 
$\varrho(P) \subseteq \varrho(F)$. 
\end{theorem}

To prove Theorem~\ref{theorem:nec}, we show the
following lemma on the relationship between 
$\sigma(P)$ and $\varrho(P)$ of a configuration $P$
without multiplicity. 
\begin{lemma}
\label{lemma:sigma-rho}
For a configuration $P$ without multiplicity, 
$\sigma(P) \in \varrho(P)$. 
\end{lemma}
\begin{proof} 
 The proof is by Property \ref{prop:size-sigma}. 
 \qed 
\end{proof}

\begin{lemma}
\label{lemma:nec-oblivious}
Oblivious FSYNC robots can form a target pattern $F$ 
from an initial configuration $P$ 
only if 
$\varrho(P) \subseteq \varrho(F)$. 
\end{lemma}
\begin{proof} 
Let $P$ and $F$ be a given initial configuration 
and a target pattern without multiplicity 
that does not satisfy $\varrho(P) \subseteq \varrho(F)$. 
Hence there exists $G \in \varrho(P)$ such that 
$G \not\in \varrho(F)$. 

Assume that there exists an algorithm $\psi$ that forms $F$ 
from $P$, for contradiction. 
Consider an initial arrangement of local coordinate systems that
 satisfies $\sigma(P) = G$.
 Such arrangement exists from the definition. 
From the assumption, there exists at least one execution 
 $P(0)(=P), P(1), P(2), \ldots $ that satisfies
 $P(t) \simeq F$ for some $t >0$. 
From Lemma~\ref{lemma:sigma-impossibility}, 
we have $\sigma(P(t)) \succeq \sigma(P(0)) = G$. 

From the definition of the symmetricity, we have the following: 
if $G' \in \varrho(F)$, any subgroup of $G'$ is in $\varrho(F)$. 
In other words, if there exists a subgroup of $G'$ 
that is not in $\varrho(F)$, then $G' \not\in \varrho(F)$. 
Because $\sigma(P(0)) = G$ is not in $\varrho(F)$, 
its supergroup $\sigma(P(t))$ is not in $\varrho(F)$. 

This contradicts Lemma~\ref{lemma:sigma-rho} that guarantees 
$\sigma(P(t)) \in \varrho(P(t))$. 
\qed 
\end{proof}

In the same way, we have the non-oblivious robot version 
directly from Lemma~\ref{lemma:sigma-impossibility-memory}.
Consequently, we have Theorem~\ref{theorem:nec}.


\section{Sufficiency of Theorem~\ref{theorem:main}} 
\label{sec:suf}

We have shown that the necessity is derived from the symmetricity of
an initial configuration.
We will show that the condition of Theorem~\ref{theorem:nec}
is also a sufficient condition for FSYNC robots to form a given target
pattern. 

\begin{theorem}
\label{theorem:suf}
Regardless of obliviousness, 
FSYNC robots can form a target pattern $F$ from 
an initial configuration $P$ 
if 
$\varrho(P) \subseteq \varrho(F)$. 
\end{theorem}

We present a pattern formation algorithm $\psi_{PF}$ that 
makes oblivious FSYNC robots form a target pattern $F$ from 
a given initial configuration $P$ if $P$ and $F$ satisfy 
the condition of Theorem~\ref{theorem:suf}.
Non-oblivious robots can also execute $\psi_{PF}$ correctly 
by ignoring local memory contents. 

As we have already seen in Section~\ref{sec:show-sym}, 
algorithm $\psi_{SYM}$ translates an initial configuration $P$ 
to another configuration $P'$ that satisfies
(i) $\gamma(P') \in \varrho(P)$, 
(ii) If $\gamma(P') \neq C_1$, then $P'$ is a regular $n$-gon
or no robot is on the rotation axes of $\gamma(P')$. 
Let $\{P_1', P_2', \ldots, P_m'\}$ be the $\gamma(P')$-decomposition of 
$P'$. 
From the first property, 
we have $\gamma(P') \in \varrho(F)$. 
From the second property, the size of each element $|P_i'|$ 
($1 \leq i \leq m$) is $|\gamma(P')|$. 
The second property implies $\sigma(P') = \gamma(P')$ 
in the worst case and  the robots forming each element of the
$\gamma(P')$-decomposition of $P'$ may
forever move symmetric positions regarding $\gamma(P')$. 

The proposed pattern formation algorithm $\psi_{PF}$ first
makes the robots agree on an embedding of $F$ in $P'$ so that
$\gamma(P')$ overlaps unoccupied rotation axes of $\gamma(F)$.
Because $\gamma(P') \in \varrho(F)$, such embedding exists and 
it guarantees that $\gamma(P')$-decomposition of $F$
consists of elements of size $|\gamma(P')|$. 
Thus we can overcome the symmetric movement of each element of the
$\gamma(P')$-decomposition of $P'$ 
by assigning it to an element of the $\gamma(P')$-decomposition of $F$. 
We denote the embedded target pattern by $\widetilde{F}$. 
Then $\psi_{PF}$ makes them compute a perfect matching between
$P'$ and $\widetilde{F}$, denoted by $M(P, \widetilde{F})$
to assign final destination to each robot. 

We will show how the robots agree on $\widetilde{F}$ in
Section~\ref{subsec:embedding} 
and on the perfect matching between $P'$ and $\widetilde{F}$ in 
Section~\ref{subsec:matching}.

\subsection{Embedding the target pattern}
\label{subsec:embedding}

Let $P$ be a current configuration that is a terminal configuration 
of $\psi_{SYM}$ and
no rotation axis of $P$ is occupied unless $P$ is on one plane. 
From the condition of Theorem~\ref{theorem:suf}, 
$\gamma(P) \in \varrho(F)$. 

To form the target pattern $F$, the robots first fix an image of the 
target pattern $F$. 
We denote this image by $\widetilde{F}$. 
The robots fix $\widetilde{F}$ so that $B(\widetilde{F}) = B(P)$
and unoccupied rotation axes of $\gamma(\widetilde{F})$ overlaps
the rotation axes of $\gamma(P)$. 
Algorithm $\psi_{PF}$ first fixes the arrangement of
$\gamma(F)$ in $P$ instead of $\widetilde{F}$. 
The robots construct an agreement on the arrangement of
rotation axes of $\gamma(F)$ that $\gamma(P)$
does not have.
In the following, we refer to an arrangement of a rotation group $G$ 
by using $U_{G, \mu}$ for some $\mu > 1$ because it helps our
understanding
with illustration. 
Depending on $\gamma(P)$, we have the following two cases. 

\medskip 

\noindent{\bf When $\gamma(P)$ is a 3D rotation group.~}
We have only two cases to consider
because $O \not\preceq I$, i.e, 
the first case is when $\gamma(P)=T$ and $\gamma(F)=O$ and
the second case is when $\gamma(P) = O$ and $\gamma(F)=I$.

First, when $\gamma(P) = T$ and $\gamma(F) = O$,
the $3$-fold axes and $4$-fold axes of $\gamma(F)$
are not occupied, otherwise $T \not\in \varrho(F)$. 
The robots fix the arrangement of $O$ in $T$ as shown in
Figure~\ref{fig:exp-TtoO}: 
$\psi_{PF}$ replace the $2$-fold axes of $T$ with
$4$-fold axes of $O$.
The set of $3$-fold axes of $T$ 
and new $4$-fold axes generate the remaining $2$-fold axes
of $O$.\footnote{Because the rotation around the $4$-fold 
axes and the $3$-fold axes are the generator of $O$.} 
Thus the rotation axes of $\gamma(P)$ corresponds to
the $3$-fold axes and the $4$-fold axes of $\gamma(F)$
and $\gamma(P)$ overlaps unoccupied rotation axes of $\gamma(F)$.

Second, when $\gamma(P)= T$ and $\gamma(F) = I$, 
the $3$-fold axes and $2$-fold axes of $\gamma(F)$
are not occupied, otherwise $T \not\in \varrho(F)$. 
The robots first fix the arrangement of $O$ in $T$, then
further add some rotation axes to fix the arrangement of $I$.
We use the orientation of the $3$-fold axes of $T$ in this
procedure. 
The robots extend $T$ to $O$ in the same way as the previous case
and consider a unit cube (in their local coordinate systems).
Then they put
two types of 3-blade fan components to each of the vertices of the cube
(Figure~\ref{fig:exp-TtoI}).
The two types of fans are mirror image of each other and 
we call one of them ``black'' and the other ``white''. 
The robots put a black component on the vertex of the cube
on the positive direction of the $3$-fold axis of $T$
(the vertices of the regular tetrahedron in Figure~\ref{fig:exp-TtoI}) 
and a white component on the vertex of the cube on the
negative direction of the $3$-fold axis of $T$. 
We need such procedure
because there are two types of embeddings of a cube to a regular 
dodecahedron (thus, the arrangement of $I$) 
(Figure~\ref{fig:dodeca-cube-1} and Figure~\ref{fig:dodeca-cube-2}), 
and robots have to agree on one of the two arrangements of $I$.
The rotation axes of $\gamma(P)$ corresponds to
the $3$-fold axes and the $2$-fold axes of $\gamma(F)$
and $\gamma(P)$ overlaps unoccupied rotation axes of $\gamma(F)$.

\begin{figure}[t]
\centering 
\subfigure[From $T$ to $O$]
{\includegraphics[height=3cm]{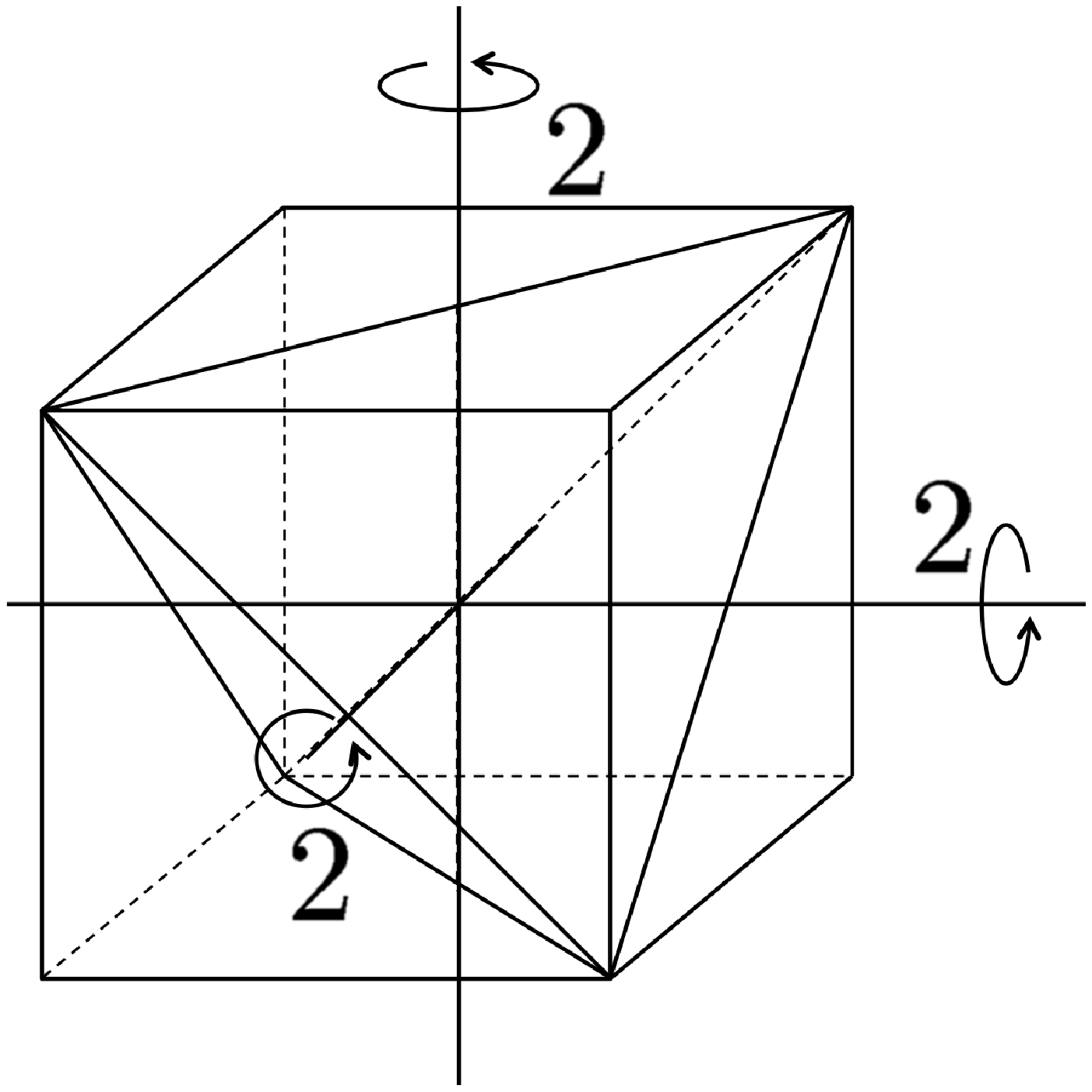}\label{fig:exp-TtoO}}
\hspace{3mm}
\subfigure[From $T$ to $I$]
{\includegraphics[height=2.5cm]{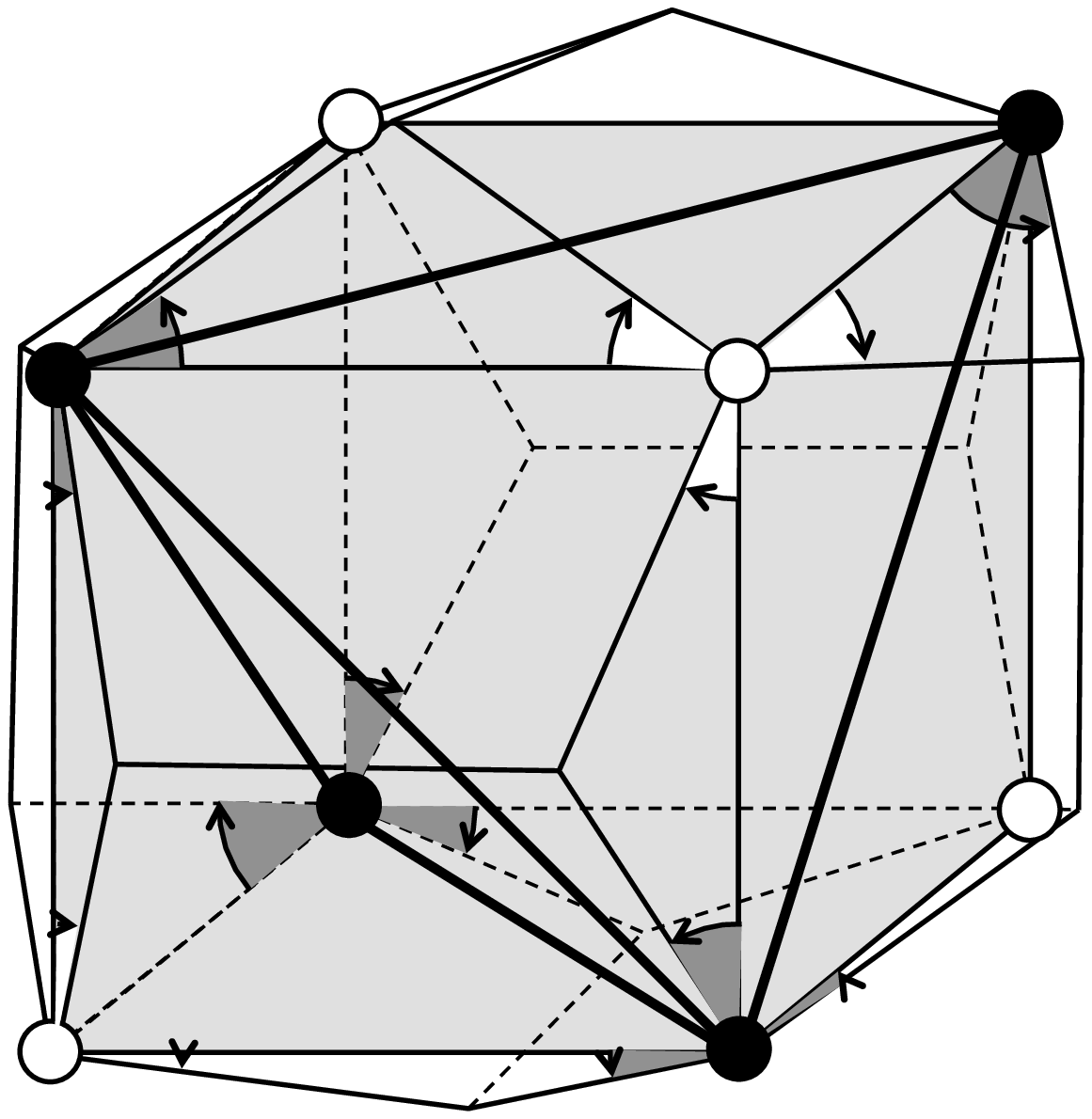}\label{fig:exp-TtoI}}
\hspace{3mm}
\subfigure[$C_2$ to $T$]
{\includegraphics[height=2.5cm]{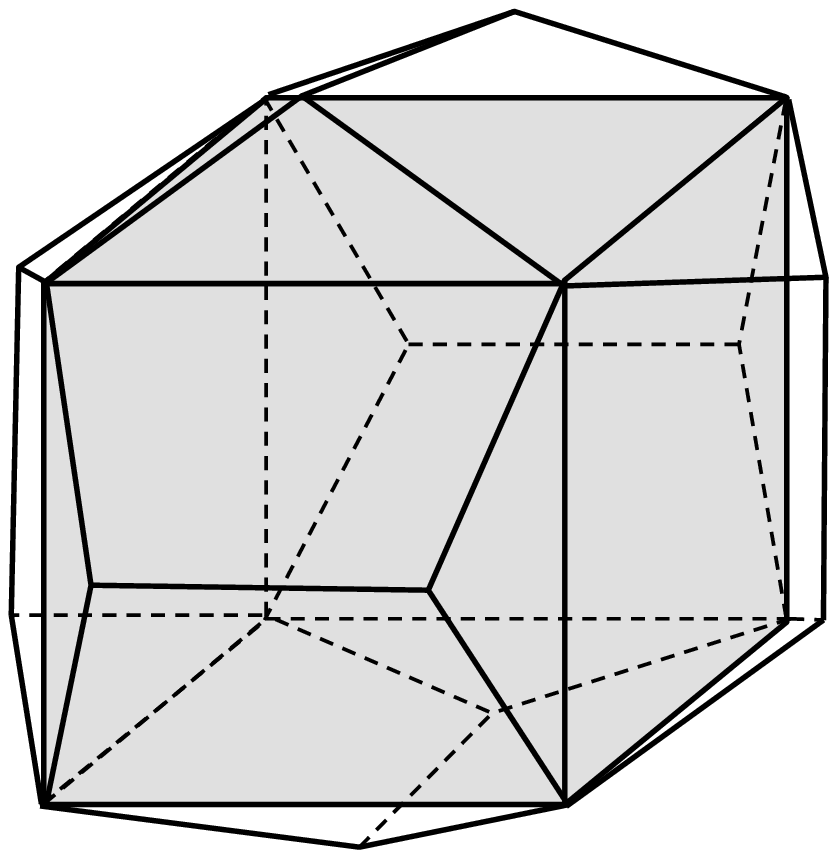}\label{fig:dodeca-cube-1}}
\hspace{3mm}
\subfigure[$C_3$ to $T$]
{\includegraphics[height=2.4cm]{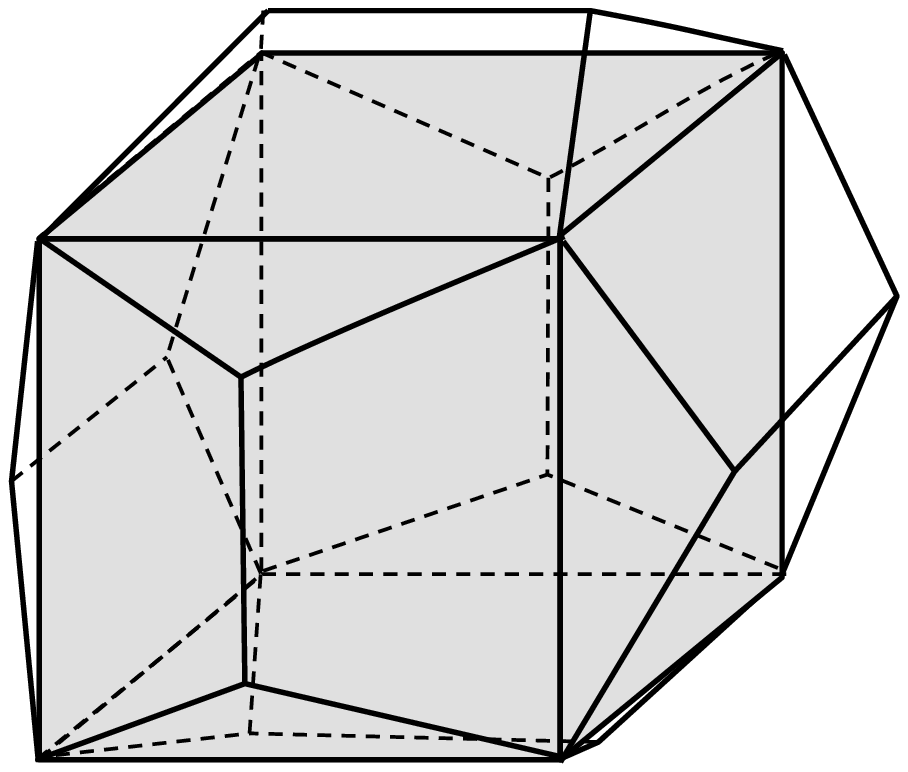}\label{fig:dodeca-cube-2}}
 \caption{Fixing $\gamma(F)$ in $\gamma(P')$.
 Instead of drawing rotation axes, we show a typical regular polyhedron
 for $\gamma(F)$ that shows the arrangement of rotation axes of
 $\gamma(F)$. }
\label{fig:3D-expand}
\end{figure}

In the above two cases, once $\gamma(F)$ is fixed,
$\widetilde{F}$ is also fixed with the condition that
$B(\widetilde{F}) = B(P')$. 

\medskip

\noindent{\bf When $\gamma(P)$ is a 2D rotation group.~}
In this case, we cannot always fix $\gamma(F)$
by using only the rotation axes of $\gamma(P)$.
Consider the case where $\gamma(P) = C_k$ and $\gamma(F)=D_k$ for some
positive integer $k$.
In this case, the robots cannot agree on the arrangement of
secondary axes of $D_k$ by using only the single rotation axis of $C_k$.
In such cases, we use the positions of point of $P$ to fix 
the rotation axes of $\gamma(F)$.
A {\em reference polygon} of $P$ for a specified $k$-fold axis
of $\gamma(P)$ is 
a regular $k$-gon on a plane that is perpendicular to $\gamma(P)$ and
contains $b(P)$.\footnote{When the principal axis of a dihedral group
$D_{\ell}$ is specified, the reference polygon also determines the
orientation of secondary axis if $\ell$ is odd. } 
We consider the specified axis as the earth's axis of $B(P)$ and 
consider the {\em equator plane} for the axis. 
Thus the reference polygon is on the equator plane.  

We also consider a reference polygon of $F$ 
in the same way. 
The role of the reference polygon is twofold; one is for $P$ to 
show reference points when fixing an arrangement of additional
rotation axes 
and the other is 
for $F$ to recover all the arrangement of rotation axes when given
the specified rotation axes and its reference polygon. 

Specifically, the robots agree on the reference polygon as follows: 

\noindent{\bf Case A. $\gamma(P) = C_k$.~}
In this case, there is only the single rotation axis. 
Let $\{P_1, P_2, \ldots, P_m\}$ be the $\gamma(P)$-decomposition of $P$. 
When $k \geq 2$, there exists at least one element in the
$\gamma(P)$-decomposition of $P$ that form a regular $k$-gon
perpendicular to the single rotation axis. 
Let $P_i$ be the element that has the minimum index among such
elements. 
Then the reference polygon is
the projection of the regular $k$-gon formed by $P_i$
onto the equator plane.

When $\gamma(P) = C_1$, each element of the $\gamma(P)$-decomposition
is a $1$-set. 
Let $P_1 = \{p_1\}$. 
Because $\gamma(P) = C_1$, there is at least one element 
that is not on the line containing $b(P)$ and $p_1$.
Let $P_i = \{p_i \}$ be the element with the minimum index among such
elements.
Then the reference polygon is the projection of $p_i$ onto the equator
plane. 

\noindent{\bf Case B. $\gamma(P) = D_{\ell}$.~}
In this case, there are two choices for the rotation axis for the
reference polygon. 
However, we do not have to consider the secondary axis of $\gamma(P)$
because
$\gamma(P) \preceq \gamma(F)$,
$\gamma(F)$ has a rotation axis whose fold is a multiple of $\ell$,
and $\psi_{PF}$ makes the principal axis of $D_{\ell}$ overlap 
such rotation axis of $\gamma(F)$.

Let $\{P_1, P_2, \ldots, P_m\}$ be the $\gamma(P)$-decomposition of $P$.
When the $\gamma(P)$-decomposition of $P$ contains $U_{D_{\ell}, 2}$,
the reference polygon is defined by the element with the minimum
index among such elements because such regular $\ell$-gons is on the
equator plane. 
Otherwise, the $\gamma(P)$-decomposition of $P$ contains at least
one $U_{D_{\ell}, 1}$ and let $P_i$ be the element with the minimum
index among such elements. 
When $P_i$ forms a regular prism, the reference polygon is
defined by the projection of the bases onto the equator plane.
When $P_i$ is not a regular prism, the robots consider
twisting the bases so that the vertices overlap the nearest plane formed
by a $2$-fold axis and the principal axis of $\gamma(P)$.
The twist follows the right screw rule with considering the
direction from a base to $b(P)$ as the positive direction
(Figure~\ref{fig:right-screw-twist}).
This twisting enables the robots agree on a regular prism 
and the reference polygon is defined in the same way as the above case. 

\begin{figure}[t]
\centering 
\includegraphics[width=2cm]{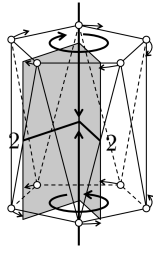}
\caption{Twisting $U_{D_5,1}$ to form a pentagonal prism.}
\label{fig:right-screw-twist}
\end{figure}

\noindent{\bf Case C. $\gamma(F)$ is cyclic or dihedral.~}
In this case, when the single rotation axis or the principal axis
is specified, the reference polygon is defined in the same way
as Case A and B.

When $\gamma(F)$ is dihedral and a secondary axis is specified,
the reference polygon is defined by the intersection of the
principal axis and the equator. 

\noindent{\bf Case D. $\gamma(F) \in \{T, O, I\}$.~}
There are three rotation axes to define a reference polygon
because $\gamma(F)$ consists of three types of rotation axes.
Figure~\ref{fig:expand} shows reference polygons
for each type of the rotation axis.
Readers can easily find that the reference polygons are defined
in the samey way as Case B, specifically,
when the specified rotation axis is oriented,
the reference polygon is defined in the same way as when
$\gamma(F)$ is cyclic (Figure~\ref{fig:exp-3toT}) 
and when the specified rotation axis is not
oriented, it is defined in the same way as when $\gamma(F)$ is
dihedral.
Find that with the specified rotation axis and its reference polygon,
we can uniquely fix the arrangement of all other rotation axes of
$\gamma(F)$. 

\begin{figure}[t]
\centering 
\subfigure[$T$ by a $2$-fold axis]
{\includegraphics[height=3cm]{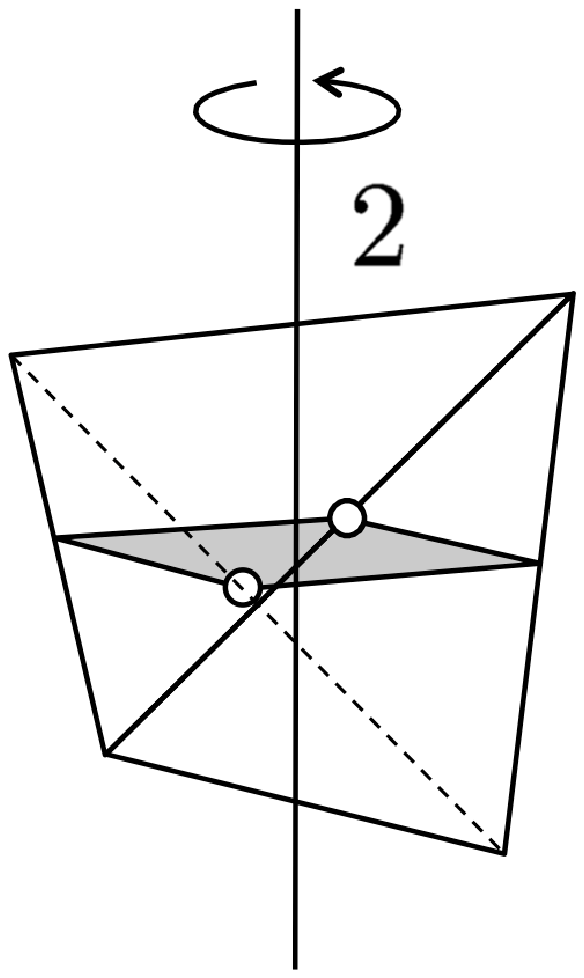}\label{fig:exp-2toT}}
\hspace{3mm}
\subfigure[$T$ by a $3$-fold axis]
{\includegraphics[height=3cm]{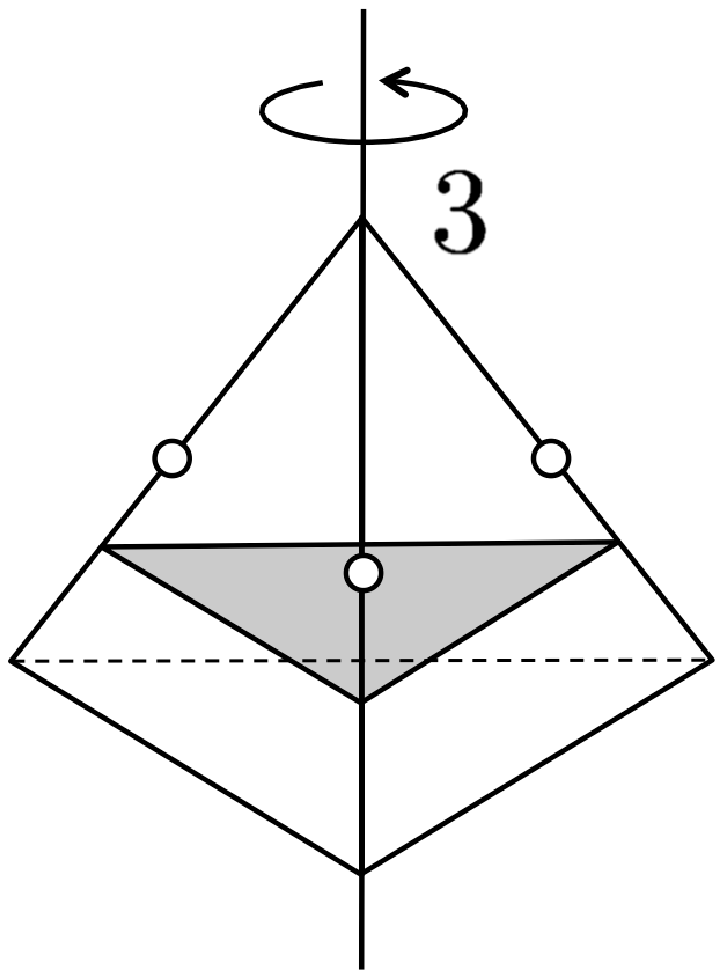}\label{fig:exp-3toT}}
\hspace{3mm}
\subfigure[$O$ by a $2$-fold axis]
{\includegraphics[height=3cm]{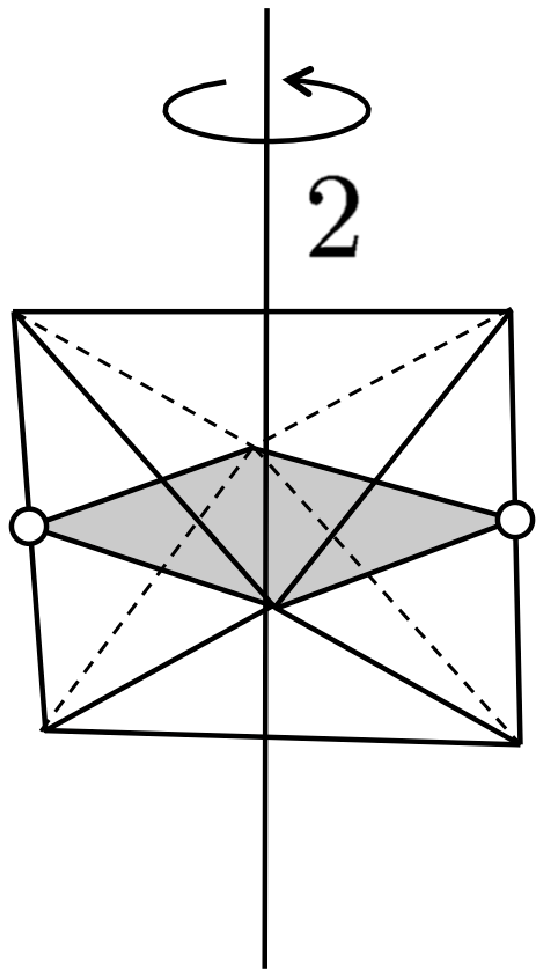}\label{fig:exp-2toO}}
\hspace{3mm}
\subfigure[$O$ by a $3$-fold axis]
{\includegraphics[height=3cm]{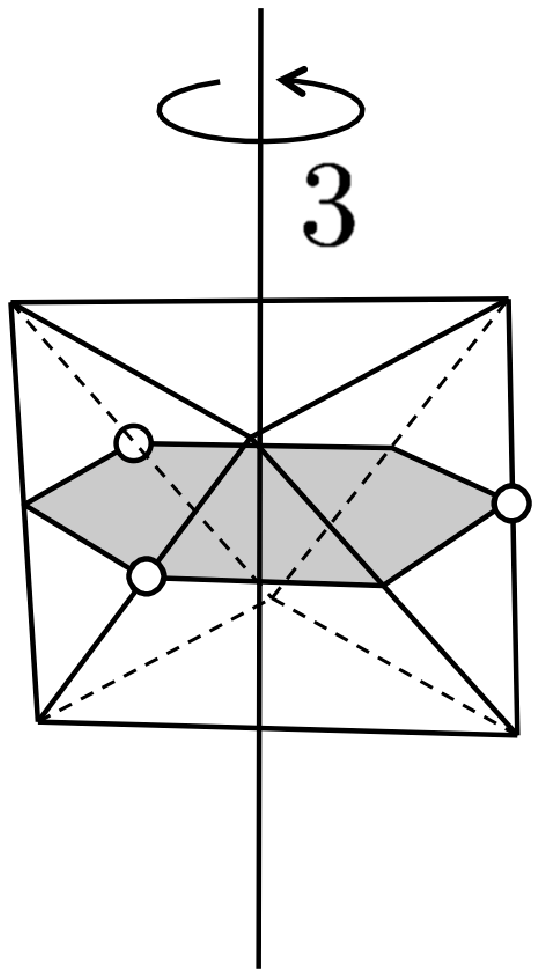}\label{fig:exp-3toO}}
\hspace{3mm}
\subfigure[$O$ by a $4$-fold axis]
{\includegraphics[height=3cm]{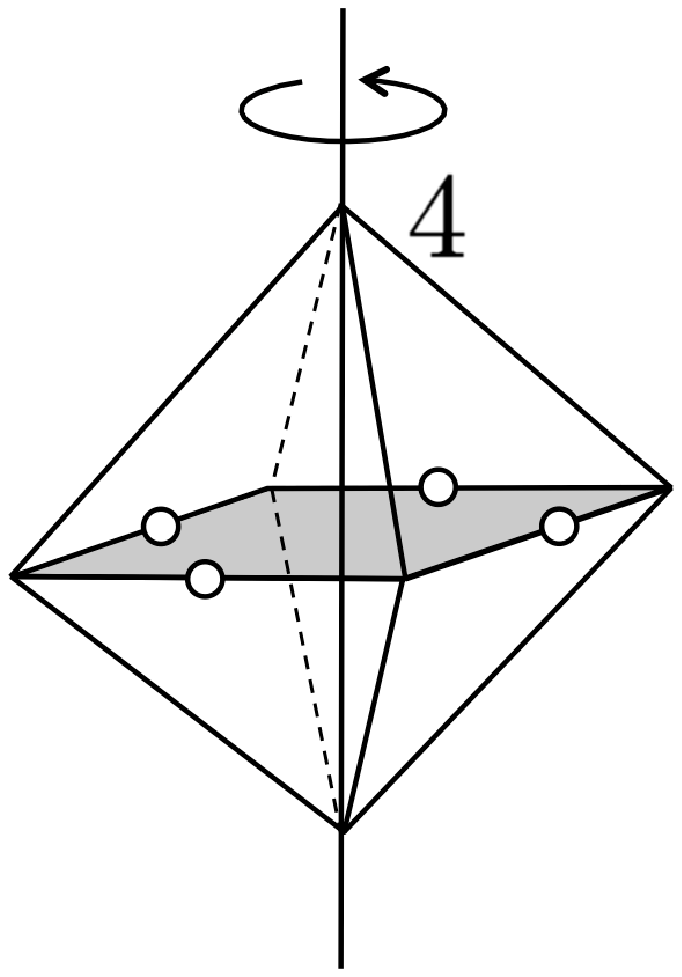}\label{fig:exp-4toO}}
\hspace{3mm}
\subfigure[$I$ by a $2$-fold axis]
{\includegraphics[height=3cm]{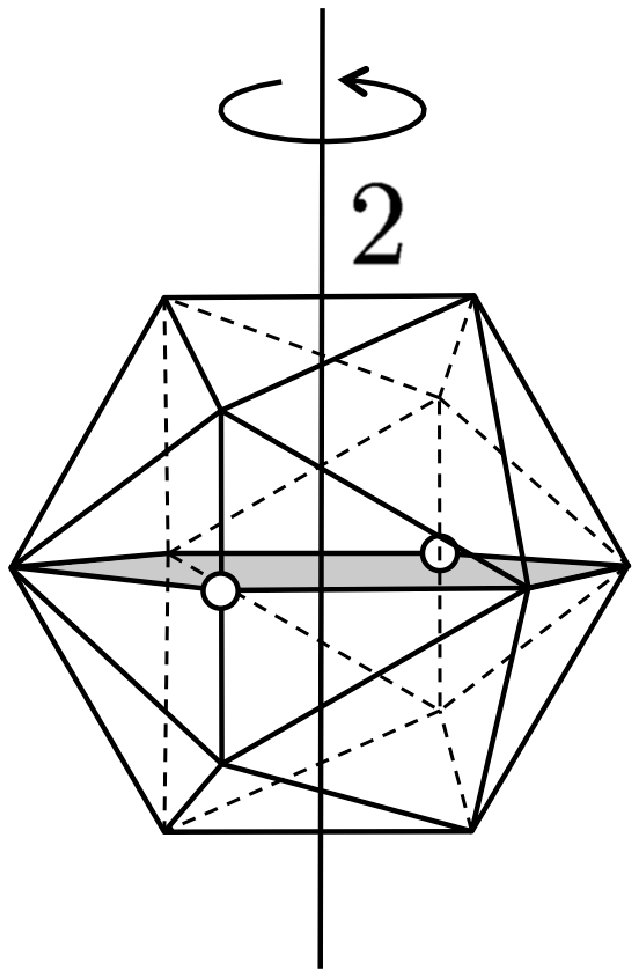}\label{fig:exp-2toI}}
\hspace{3mm}
\subfigure[$I$ by a $3$-fold axis]
{\includegraphics[height=3cm]{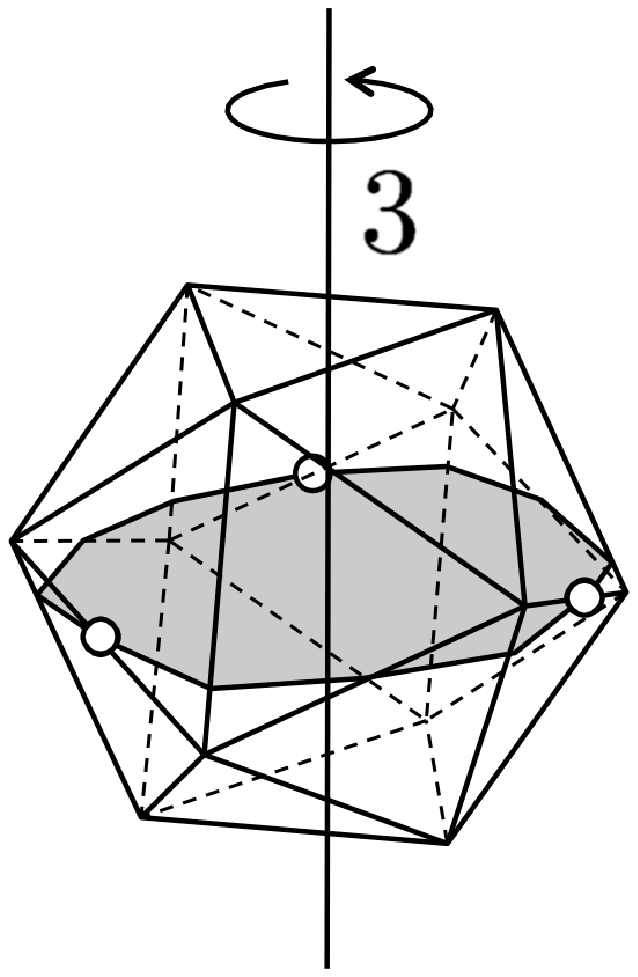}\label{fig:exp-3toI}}
\hspace{3mm}
\subfigure[$I$ by a $5$-fold axis]
{\includegraphics[height=3cm]{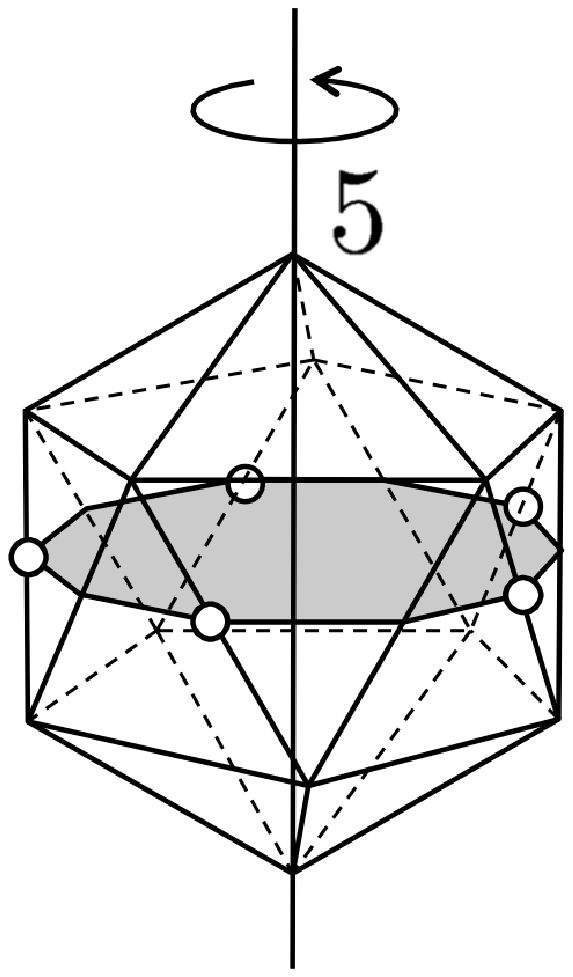}\label{fig:exp-5toI}}
 \caption{Reference polygon for $T, O, I$ when one rotation axis is
 specified. Instead of drawing rotation axes,
 we show a typical regular polyhedron
 that shows the arrangement of rotation axes of $\gamma(F)$.} 
\label{fig:expand}
\end{figure}

Now we describe how the robots fix $\gamma(F)$ in $\gamma(P)$
when $\gamma(P)$ is a cyclic group or a dihedral group.
Let the single rotation axis (or the principal axis) of $\gamma(P)$
be a $k$-fold axis. 
To fix the arrangement of $\gamma(F)$ in $\gamma(P)$,
we first consider an embedding of $\gamma(P)$ to unoccupied
rotation axes of $\gamma(F)$.
If there are multiple ways to embed $\gamma(P)$,
we select an embedding where the single (or principal) rotation axis of
$\gamma(P)$ corresponds to the maximum fold of $\gamma(F)$. 
There may be still multiple ways to embed $\gamma(P)$,
but the vertices of the reference polygon of $P$ overlap those $F$. 
For example, consider the case where $\gamma(P) = C_1$ and
$\gamma(F)=T$. In this case, $\psi_{PF}$ embeds $C_1$ to a
$3$-fold axis of $T$, and the arrangement
of $T$ does not depend on which $3$-fold axis is selected.
Algorithm $\psi_{PF}$ realizes such arrangement of $\gamma(F)$
in $\gamma(P)$ by using the specified rotation axes and
reference polygons of $P$ and $F$. 

Finally, if embedding of $\gamma(F)$ to $\gamma(P)$ does not 
fix $\widetilde{F}$,
$\psi_{PF}$ fixes $\widetilde{F}$ by overlapping the reference polygon
of $F$ to that to $P$. 
For example, when $P$ is a pyramid with a regular $k$-gon base and
$F$ is a pyramid with a regular $2k$-gon base,
$\psi_{PF}$ fixes $\widetilde{F}$ by overlapping the
regular $2k$-gon reference polygon of $F$ to
the regular $k$-gon reference polygon of $P$.

\subsection{Assigning the final position}
\label{subsec:matching} 

Let $P$ and $\widetilde{F}$ be
a terminal configuration of $\psi_{SYM}$ and 
the target pattern fixed in $P$, respectively. 
The robots now compute a perfect matching between the 
points of $P$ and the points of $\widetilde{F}$ to 
finally form the target pattern. 

We now consider the rotation group of $P \cup \widetilde{F}$. 
We consider the rotations that matches the points of $P$ to $P$ itself 
and those of $\widetilde{F}$ to $\widetilde{F}$. 
Hence, $\gamma(P \cup \widetilde{F}) \preceq \gamma(P)$. 
Actually, $\gamma(P \cup \widetilde{F}) = \gamma(P)$ because 
$\gamma(P) \in \varrho(\widetilde{F})$ and each 
$G \in \varrho(\widetilde{F}) \preceq \gamma(F)$, i.e., 
any rotation of $\gamma(P)$ is applicable to $\widetilde{F}$. 
The group action of $\gamma(P)$ divides 
$P \cup \widetilde{F}$ to a transitive set of points regarding 
$\gamma(P)$ so that each element consists of only the point of $P$ 
or only those of $F$. 
Additionally, each element consists of $|\gamma(P)|$ points 
since no robot is on $\gamma(P)$. 

Now, in the same way as \cite{YUKY15}, the robots can order the 
elements. 
Let $\{P_1, P_2, \ldots, P_m\}$ and $\{F_1, F_2, \ldots, F_m\}$ 
be the elements of $P$ and those of $\widetilde{F}$ 
that appears the entire decomposition in this order. 
Then, the proposed algorithm makes the robots forming 
$P_i$ to the positions of $F_i$ for each $1 \leq i \leq m$. 
In each element $P_i$, each robot selects the nearest point in $F_i$ 
as its destination. 
We first show that there exists a minimum weight perfect matching 
between the points of $P_i$ and $F_i$, where the weight is the 
sum of distances between matched points. 

\begin{lemma}
For each element $P_i$ and $F_i$, 
there exists a minimum weight perfect matching between the 
points of $P_i$ and the points of $F_i$. 
\end{lemma}
\begin{proof}
For an arbitrary $p_j \in P_i$, 
let $f_j$ be one of the nearest point in $F_i$. 
 Because $P_i$ is transitive regarding $\gamma(P)$ and
 $|P_i| = |\gamma(P)|$, 
for each $p_k \in P_i$, there exists $g_k \in \gamma(P)$ 
such that $g_k * p_j = p_k$, and for any $p_k \neq p_{\ell}$, 
$g_k \neq g_{\ell}$. 
We apply $g_k$ to $f_j$ so that we obtain the matching 
point for each $p_k \in P_i$. 
 Because $F_i$ is also transitive regarding $\gamma(P)$ and
 $|F_i| = |\gamma(P)|$, 
 this procedure produces distinct matching points for each $p_k \in P_i$. 
Consequently, we obtain a minimum weight perfect matching 
between $P_i$ and $F_i$. 
\qed
\end{proof}

However, each point of $p \in P_i$ may have multiple nearest 
destinations. 
Figure~\ref{fig:multiple-dest} shows an example where 
$P_i$ forms a expanded cube and $F_i$ forms a truncated cube.
For each robot (white circle), there are two nearest destinations
(black circles) around the nearest corner of the cube. 
In this case, we can show that the conflict forms a cycle around a 
rotation axis and the robots can resolve it 
by a right-screw rule around the rotation axis. 
The following lemma shows that we can apply this idea to
resolve any conflict. 

\begin{figure}[t]
\centering 
\includegraphics[width=3cm]{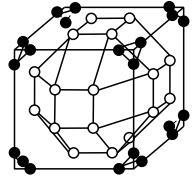}
\caption{Minimum weight perfect matching between 
 elements of $P$ and $\widetilde{F}$. The white circles are positions of
 the robots, 
 and the black circles are the positions of destinations.
 Thus $P$ forms a expanded cube and $\widetilde{F}$ forms a
 truncated cube. 
 Each robots has two nearest target points.}
\label{fig:multiple-dest}
\end{figure}

\begin{lemma} 
Consider the graph $G$ formed by vertices of $P_i \cup F_i$. 
For each $p \in P_i$, if $f \in F_i$ is a nearest destination, 
then we have edge $(p,f)$. 
Then, if $G$ contains a cycle, it is around a rotation axes of 
$\gamma(P)$ and such rotation axis is uniquely determined for each 
cycle. 
\end{lemma}
\begin{proof}
Clearly, each $p_j \in P_i$ has at least one nearest destination. 
First, we prove that each $p_j \in P_i$ has at most two 
nearest points of $F_i$. 
 Assume $p_j$ has $k>2$ nearest points of $F_i$.
 Then these points are on a circle on $B(F_i)$
 and because they are nearest to $p_j$, 
no point of $F_i$ is in this circle. 
Hence these $k$ points are on one face of $F_i$. 
Because $F_i$ is transitive set of points, 
they are on a sphere, these $k$ points form a face of the convex hull of 
$F_i$, 
 and $p_j$ is on the line connecting the center of the face
 and $b(P)$. 
Because $|F_i| = |\gamma(P)|$, the center of each face of $F_i$ 
 intersects with a rotation axis and $p_i$ is on a rotation axis,
 which contradicts 
the fact that $|P_i| = |\gamma(P)|$. 

Second, we will show that if $p_j \in P_i$ has two nearest destinations, 
 then there is a cycle in $G$ containing $p_j$ around a rotation axis of
 $\gamma(P)$.
 Let $f_{j1}, f_{j2}$ be the two nearest destinations of $p_j$. 
 Then in the same way as the above discussion, each $p_k \in P_i$ has
 such two nearest destinations. 
 By the counting argument, there exists at least one $p_k \in P_i$ such
 that $\{f_{j1}, f_{j2}\} \cap \{f_{k1}, f_{k2}\} \neq \emptyset$.
 Let $f_{j2} = f_{k1}$.
 Because $P_i$ is transitive, for $p_k$, there exists
 $g_k \in \gamma(P)$ such that $g_k * p_j = p_k$. 
 Then, there exists another point $p_{\ell} = g_k * p_k \in P_i$
 such that $f_{k2} = f_{\ell1}$.
 Be repeating this argument, there exists a subset of $P_i$
 that share their nearest destinations.
 Remember that $g_k$ is a rotation around some rotation axis of
 $\gamma(P)$. Thus $g_k$ forms a cyclic group and these subsets are
 around this rotation axis. Thus any cycle in $G$ is around some
 rotation axis of $\gamma(P)$.
 Additionally such rotation axis is uniquely determined. 
\qed 
\end{proof}

For each cycle, the robots can recognize the unique 
rotation axis that produces the cycle, 
and they can resolve the conflict by the right-handed screw rule with
the positive direction being $b(P)$, 
i.e., they select the nearest element in the clockwise direction 
around the rotation axis. 

We denote the entire matching obtained by these rules 
by $M(P, \widetilde{F})$. 
Remember that all computations consisting of 
finding reference polygon, 
fixing $\widetilde{F}$, 
decomposition $P \cup \widetilde{F}$, and 
computing $M(P, \widetilde{F})$ is done in one Compute phase 
in a terminal configuration of $\psi_{SYM}$. 
Finally, robots move the corresponding position in $M(P, \widetilde{F})$ 
to complete the pattern formation. 

We finally show the proposed pattern formation algorithm $\psi_{PF}$
in Algorithm~\ref{alg:pf}. 

\begin{algorithm}
\caption{Pattern formation algorithm $\psi_{PF}$ for robot $r_i \in R$}
\label{alg:pf} 
\begin{tabbing}
xxx \= xxx \= xxx \= xxx \= xxx \= xxx \= xxx \= xxx \= xxx \= xxx
\kill 
{\bf Notation} \\ 
\> P: The positions of robots observed in $Z_i$. \\ 
\> $p_i$: current position of $r_i$. \\ 
\\ 
{\bf Algorithm}  \\ 
\> {\bf If} $P$ is not a terminal configuration of $\psi_{SYM}$ {\bf then} \\ 
\> \> Execute $\psi_{SYM}$. \\ 
\> {\bf Else} \\ 
\> \> Let $\widetilde{F}$ be the target pattern fixed in $P$. \\ 
\> \> Move to the matched point in $M(P, \widetilde{F})$. \\ 
\> {\bf Endif}
\end{tabbing}
\end{algorithm}

\section{Discussion and conclusion} 
\label{sec:concl}

We have shown a necessary and sufficient condition for
FSYNC robots to form a given target pattern. 
We introduce the notion of symmetricity of positions of robots 
in 3D-space and used it to characterize the pattern formation problem.

In this section, we consider target patterns with multiplicity and
show that we have the same characterization.
To define the symmetricity of target patterns with multiplicities,
we extend the notion of symmetricity as follows:
A multiset of points $P$ is transitive regarding a
rotation group $G \in {\mathbb S}$ if 
it is one orbit regarding $G$ and
the multiplicity of point $p \in P$ on a $k$-fold
rotation axis of $G$ is $k$. 

\begin{definition}
\label{def:multi-symmetricity}
Let $P$ be a multiset of points. 
The symmetricity of $P$, denoted by $\varrho(P)$, 
is the set of rotation groups $G \in {\mathbb S}$
such that there is an arrangement of $G$ that decomposes $P$ 
into transitive multisets or transitive sets of points regarding $G$. 
\end{definition}

The above definition adds the decomposition of $F$
into transitive multiset of points when $F$ contains multiplicity. 
For example, consider a target pattern $F$ whose
points occupy the vertices of a cube but each
vertex contains three points of $F$. 
Thus $|F| = 24$ and $\varrho(F) = \{O\}$.
Clearly, from an initial configuration $P$ that forms
a truncated cube, the oblivious FSYNC robots can form $F$
by each robot gathering the nearest vertex of the cube.
The following theorem directly follows from the
discussions through this paper.

\begin{theorem}
\label{theorem:main-mult}
Regardless of obliviousness, 
FSYNC robots can form a target pattern $F$ with
multiplicity 
from 
an initial configuration $P$ 
if and only if 
$\varrho(P) \subseteq \varrho(F)$. 
\end{theorem}

The necessity is clear from the discussion in Section~\ref{sec:nec}.
The proposed pattern formation algorithm can be easily extended to
target patterns with multiplicity.
The only difference is when we consider the $\gamma(P')$-decomposition
of $P' \cup \widetilde{F}$ where the robots cannot agree on a unique
ordering of the 
elements formed by $\widetilde{F}$ because of the multiplicity.
However, this procedure does not require the robots to agree on the
ordering of elements formed by $\widetilde{F}$
occupying the same positions.
The robots just agree on the ordering among the elements formed by $P'$,
and which elements is assigned to such positions with multiplicity.

Our future direction is to consider the pattern formation problem
for weaker robot models, for example, 
\begin{itemize}
 \item SSYNC or ASYNC robots, 
 \item Robots with non-rigid movement, 
 \item Robots with limited visibility, and 
 \item Robots without chirality. 
\end{itemize}
Another question is whether there exists a clear separation
between the ability of robots in 2D-space and
that of robots in 3D-space.

\newpage

\appendix 

\section{Property of rotation groups}

\label{app:rotation-groups}

\noindent{\bf Property~\ref{property:d2-principal}.~} 
{\it
Let $P \in {\cal P}_n^3$ be a set of points. 
If $D_2$ acts on $P$ and we cannot distinguish the principal axis of 
(an arbitrary embedding of) $D_2$, then $\gamma(P) \succ D_2$. 
}
\begin{proof}
Without loss of generality, 
we can assume that $x$-$y$-$z$ axes of the global coordinate system $Z_0$
 are the $2$-fold axes of 
 $D_2$.\footnote{There exists a translation consisting of
 rotation and translation that overlaps the $2$-fold axis of
 $\gamma(P)$ to the three axes.}
We define the octant according to $Z_0$ as shown in 
Figure~\ref{fig:rot3-0} and Table~\ref{table:octant}. 

\begin{table}[t]
\begin{center}
\caption{Definition of octant} 
\label{table:octant} 
\begin{tabular}[t]{|c|c|c|c|}
Number & $x$ & $y$ & $z$\\
\hline 
1 & + & + & + \\
2 & - & + & + \\
3 & - & - & + \\
4 & + & - & + \\
5 & + & + & - \\
6 & - & + & - \\
7 & - & - & - \\
8 & + & - & - \\
\end{tabular}
\end{center}
\end{table}

We consider the positions of points of $P$ in the first octant,  
which defines the positions of points of $P$ in the 
third, sixth, and the eighth octant by the rotations of $D_2$. 
The discussion also holds symmetrically in the second octant, 
that determines the positions of points in the 
fourth, fifth, and seventh octant. 

 We focus on a point $p \in P$ and
 depending on the position of $p$, we have the 
following five cases. 
\begin{itemize}
\item $p$ is on the $x$-axis (thus, the discussion follows for 
$y$-axis and $z$-axis, respectively). 
\item $p$ is on the $x$-$y$ plane (thus, the discussion follows for 
$y$-$z$ plane and $z$-$x$ plane, respectively). 
\item $p$ is on the line $x=y=z$. 
\item other cases. 
\end{itemize}
We will show that in any of the four cases, 
if we cannot recognize the principal axis, then 
we can rotate $P$ around the four $3$-fold axis $x=y=z$, 
$-x=y=z$, $-x=-y=z$, and $x=-y=z$. 

\begin{figure}[t]
\centering 
\subfigure[] 
{\includegraphics[width=4cm]{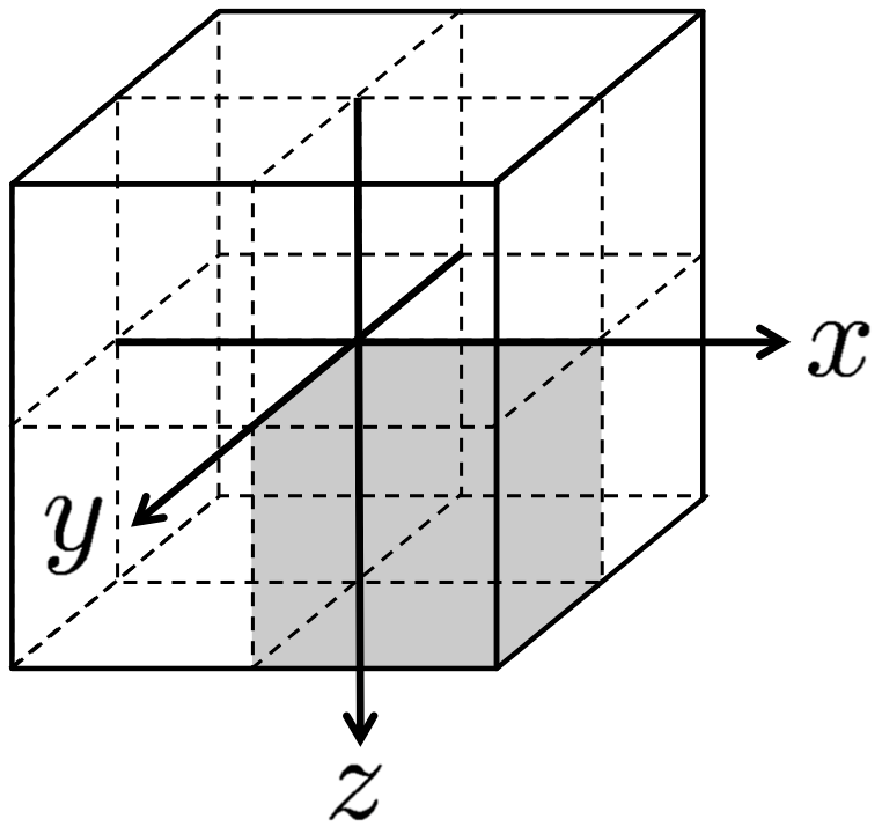}\label{fig:rot3-0}}
\hspace{3mm}
\subfigure[] 
{\includegraphics[width=4cm]{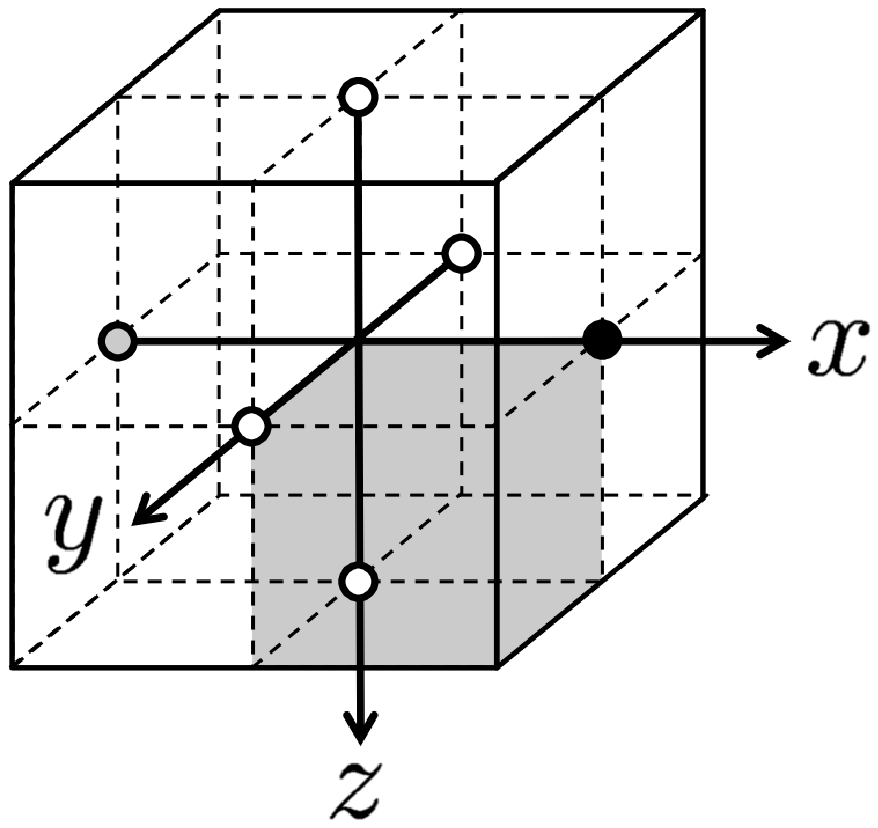}\label{fig:rot3-1}}
\hspace{3mm}
\subfigure[] 
{\includegraphics[width=4cm]{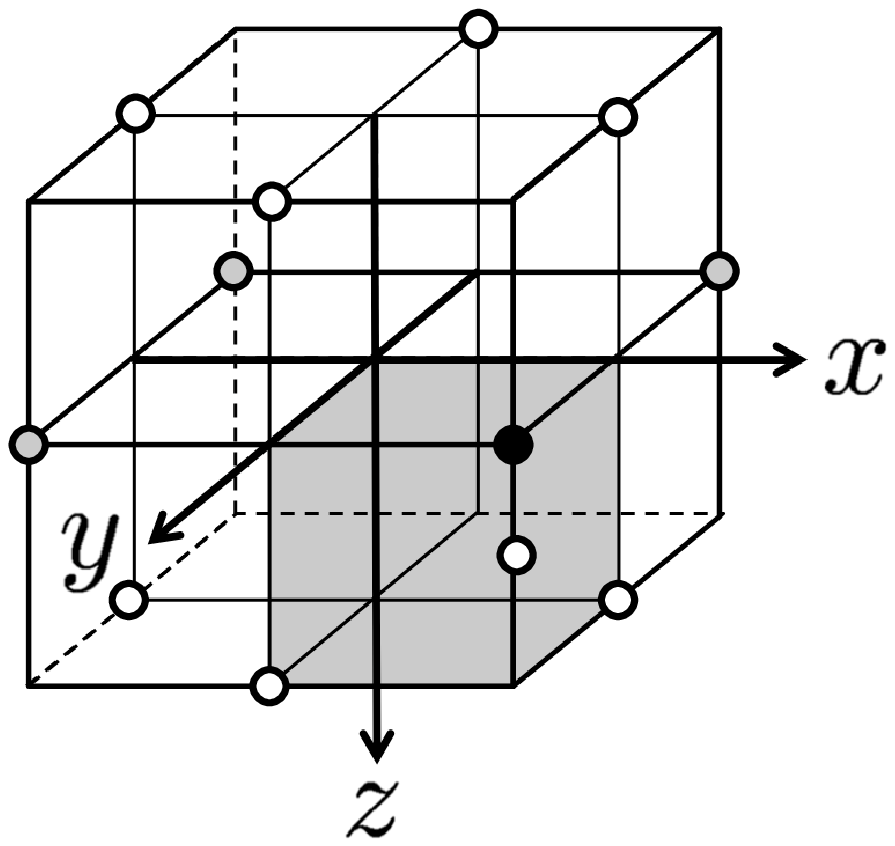}\label{fig:rot3-2}}
\hspace{3mm}
\subfigure[] 
{\includegraphics[width=4cm]{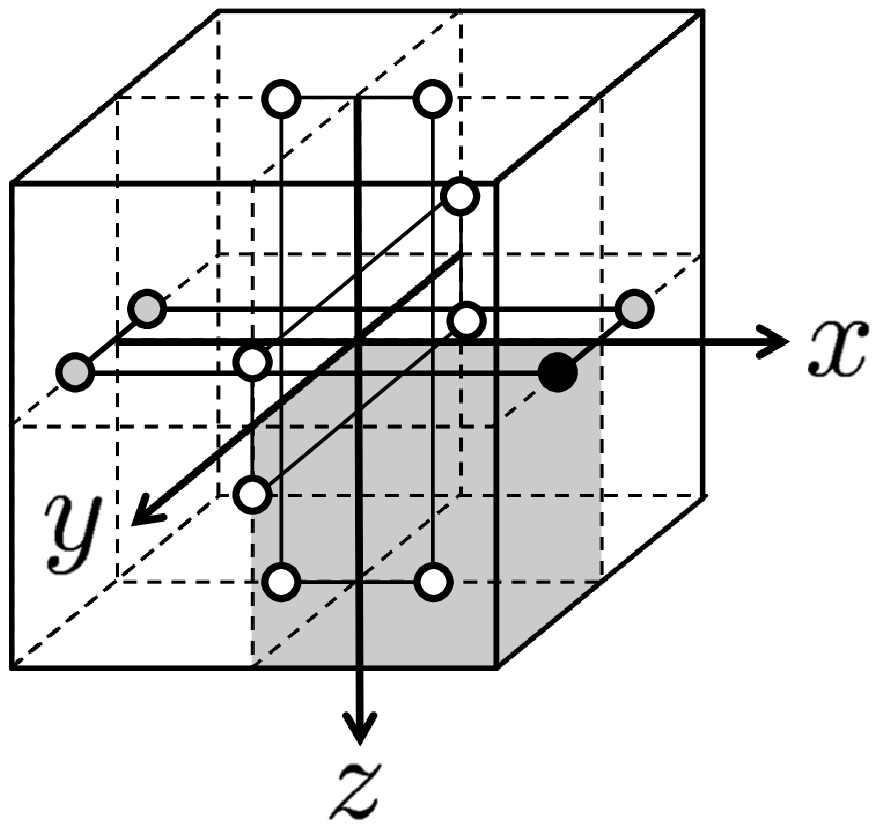}\label{fig:rot3-3}}
\hspace{3mm}
\subfigure[] 
{\includegraphics[width=4cm]{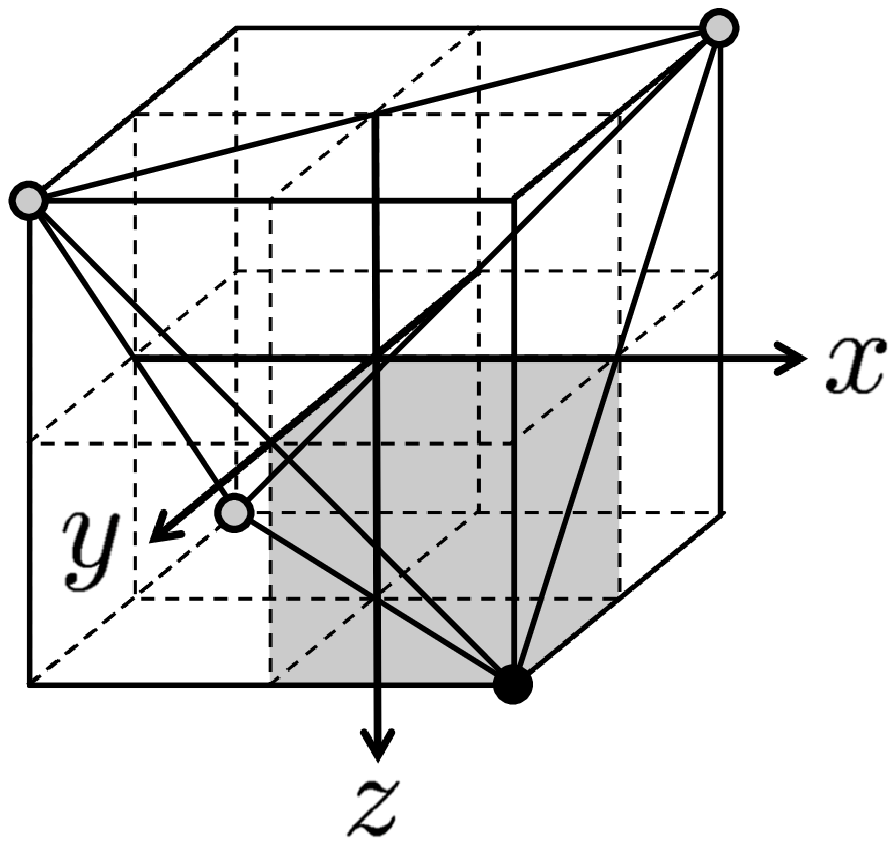}\label{fig:rot3-4}}
\hspace{3mm}
\subfigure[] 
{\includegraphics[width=4cm]{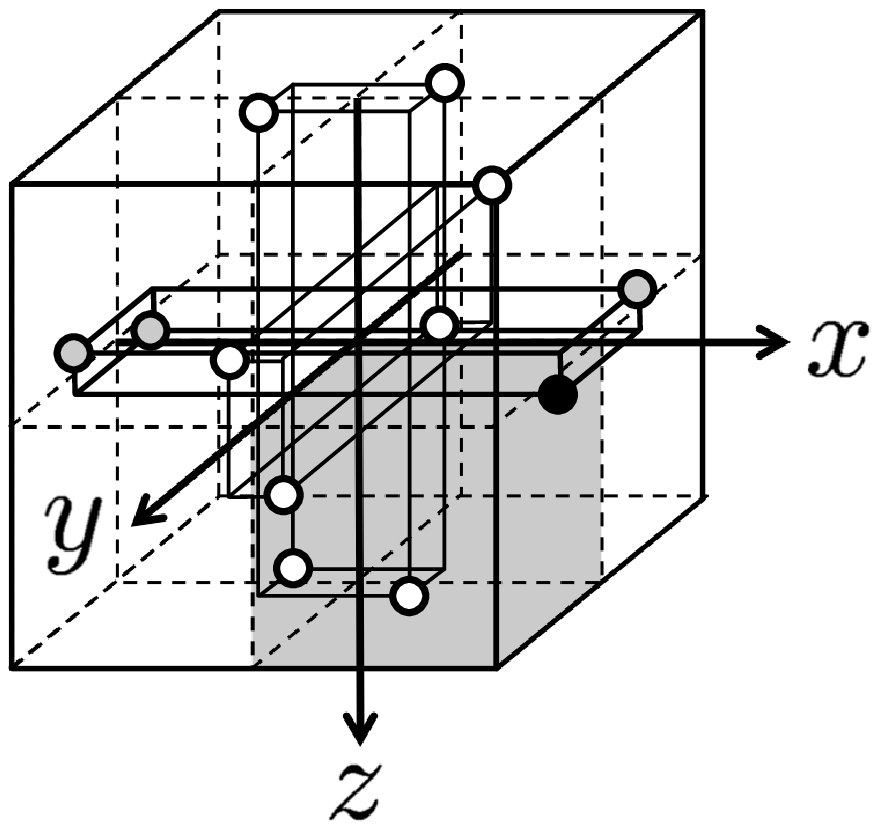}\label{fig:rot3-5}}
\caption{Position of a point of $P$ in the first octant, 
and the corresponding points generated by the $D_2$. 
The first octant is shown in the gray box in (a). 
The black circle is a point of $P$, and the gray circles are 
the points generated by $D_2$. 
The white circles are generated so that none of the three 
rotation axes is recognized.  }
\label{fig:rot3}
\end{figure}

\noindent{\bf Case A:~} When $p \in P$ is on the $x$-axis. 
Because $\gamma(P) = D_2$, we have a corresponding point on 
the negative $x$-axis (Figure~\ref{fig:rot3-1}). 
This allows us to recognize the $x$-axis from the $y$-axis and 
$z$-axis, hence $P$ should have corresponding points on 
$y$-axis and $z$-axis. 
In this case, we can rotate the corresponding six points 
around the four $3$-fold axes. 

\noindent{\bf Case B:~} When $p \in P$ is on the $x$-$y$ plane. 
First consider the case where a point $p \in P$ is on the line $x=y$. 
Because $\gamma(P) = D_2$, we have four corresponding points on 
the $x$-$y$ plane that forms a square (Figure~\ref{fig:rot3-2}). 
This allows us to recognize the $z$-axis from the other two 
axes, hence 
$y$-$z$ plane and $z$-$x$ plane also have the corresponding squares. 
Hence, the twelve points form a cuboctahedron, and 
we can rotate them around the four $3$-fold axes. 

When $p$ is not on the line $x=y$, 
because $\gamma(P) = D_2$, we have four corresponding points on 
the $x$-$y$ plane that forms a rectangle (Figure~\ref{fig:rot3-3}). 
This allows us to recognize the principal axis. 
In the same way as the above case, 
there are two rectangles on the $y$-$z$ plane and $z$-$x$ plane. 
The obtained polyhedron consists of $12$ vertices 
and we can rotate it around the four $3$-fold axes. 

\noindent{\bf Case C:~} When $p \in P$ is on the line $x=y=z$. 

Because $\gamma(P) = D_2$, we have four corresponding points in 
the third, sixth, and the eighth octant, 
that forms a regular tetrahedron (Figure~\ref{fig:rot3-4}). 
In this case, we can rotate the corresponding four points 
around the four $3$-fold axes. 

\noindent{\bf Case D:~} Other cases. 

For a point $p \in P$ in the first octant, 
because $\gamma(P) = D_2$, we have corresponding four points in 
the third, sixth, and the eighth octant, 
that forms a sphenoid (Figure~\ref{fig:rot3-5}). 
This allows us to recognize the $z$-axis from the others, 
hence $y$-axis and $x$-axis also have the corresponding sphenoids. 
The obtained polyhedron consists of $12$ vertices 
and we can rotate it around the four $3$-fold axes. 

Consequently when $D_2$ acts on $P$ 
but we cannot recognize the principal axis, 
we can rotate $P$ around the four $3$-fold axes. 
Thus $\gamma(P) \succeq T$. 

\qed 
\end{proof}

Clearly, Property~\ref{property:d2-principal} holds for
the robots since the above discussion 
does not depend on the local coordinate systems.


\begin{thebibliography}{99}

\bibitem{AOSY99}
H.~Ando, Y.~Oasa, I.~Suzuki, and M.~Yamashita,
Distributed memoryless point convergence algorithm for mobile
robots with limited visibility,
{\em IEEE Trans. Robotics and Automation}, 15, 5, pp.818--828,
1999.

\bibitem{A88}
M.A.~Armstrong, Groups and symmetry,
Springer-Verlag New York Inc., 1988. 

\bibitem{CFPS12} 
M.~Cieliebak, P.~Flocchini, G.~Prencipe, and N.~Santoro, 
Distributed computing by mobile robots: gathering, 
{\em SIAM J. Comput.,} 41, 4, pp.829--879, 2012. 
	
\bibitem{C73} 
H.S.M.~Coxeter,
Regular polytopes, 
Dover Publications, 1973. 

\bibitem{C97} 
P.~Cromwell, 
Polyhedra, 
University Press, 1997. 

\bibitem{DFPSY16}
S.~Das, P.~Flocchini, G.~Prencipe, N.~Santoro, and
M.~Yamashita,
Autonomous mobile robots with lights,
{\em Theor. Comput. Sci,} 609, pp.171--184, 2016. 

\bibitem{DFSY15}
S.~Das, P.~Flocchini, N.~Santoro, and M.~Yamashita,
Forming sequence of geometric patterns with oblivious mobile
robots,
{\em Distrib. Comput.}, 28, pp.131--145, 2015. 

\bibitem{D74}
E.W.~Dijkstra, 
Self-stabilization in spite of distributed control, 
{\em Communications of the ACM},
{\bf 17}, 11, pp.643--644, 1974. 

\bibitem{EP09}
A.~Efrima and D.~Peleg,
Distributed algorithms for partitioning a swarm of autonomous
mobile robots,
{\em Theor. Comput. Sci.,} 410, pp.1355--1368, 2009. 
	
\bibitem{FPS12} 
P.~Flocchini, G.~Prencipe, and N.~Santoro,
Distributed Computing by Oblivious Mobile Robots,
Morgan \& Claypool, 2012.
	
\bibitem{FPSV14}
P.~Flocchini, G.~Prencipe, N.~Santoro, and 
G.~Viglietta, 
Distributed computing by mobile robots: 
Solving the uniform circle formation problem, 
{\em In Proc. of OPODIS'14}, pp.217-232, 2014. 
 
\bibitem{FPSW05}
P.~Flocchini, G.~Prencipe, N.~Santoro, and P.~Widmayer,
Gathering of asynchronous robots with limited visibility,
{\em Theor. Comput. Sci.}, 337, pp.147--168, 2005. 

\bibitem{FPSW08} 
P.~Flocchini, G.~Prencipe, N.~Santoro, and P.~Widmayer, 
Arbitrary pattern formation by asynchronous, anonymous, oblivious 
robots, 
{\em Theor. Comput. Sci.,} 407, pp.412--447, 2008.

\bibitem{FYOKY15} 
N.~Fujinaga, Y.~Yamauchi, H. Ono, S.~Kijima, and M.~Yamashita, 
Pattern formation by oblivious asynchronous mobile robots,
{\em SIAM J. Comput.}, 44, 3, pp.740--785, 2015. 

\bibitem{IKY14}
T.~Izumi, S.~Kamei, and Y.~Yamauchi,
Approximation algorithms for the set cover formation by
oblivious mobile robots,
{\em In Proc. of OPODIS 2014}, pp.233--247, 2014. 
	
\bibitem{P05}
D.~Peleg,
Distributed coordination algorithms for mobile robot swarms:
New directions and challenges,
{\em In Proc. of IWDC 2005}, pp.1--12, 2005. 
	
\bibitem{SY99} 
I.~Suzuki and M.~Yamashita, 
Distributed anonymous mobile robots: Formation of geometric patterns, 
{\em SIAM J. Comput.}, 28, 4, pp.1347--1363, 1999. 

\bibitem{YS10} 
M.~Yamashita and I.~Suzuki, 
Characterizing geometric patterns formable by oblivious anonymous mobile 
robots, 
{\em Theor. Comput. Sci.}, 411, pp.2433--2453, 2010. 

\bibitem{YY13}
Y.~Yamauchi and M.~Yamashita, 
Pattern formation by mobile robots with limited visibility,
{\em In Proc. of SIROCCO 2013}, pp.201--212, 2013. 
	
\bibitem{YY14}
Y.~Yamauchi and M.~Yamashita,
Randomized pattern formation algorithm for asynchronous
oblivious mobile robots,
{\em In Proc. of DISC 2014}, pp.137--151, 2014.

\bibitem{YUKY15}
Y.~Yamauchi, T.~Uehara, S.~Kijima, M.~Yamashita, 
Plane Formation by Synchronous Mobile Robots in the Three Dimensional
Euclidean Space, {\em In Proc. of DISC 2015}, pp.92-106, 2015.	

\end{thebibliography}
\end{document}